\documentclass[format=acmsmall, review=false, screen=true]{acmart}

\usepackage{algorithm}%
\usepackage[noend]{algpseudocode}
\makeatletter
\def\BState{\State\hskip-\ALG@thistlm}
\makeatother

\usepackage{capt-of}
\usepackage{amssymb}
\usepackage{multirow}%
\usepackage{pst-node}%
\usepackage{pst-rel-points}%
\usepackage{xspace}%
\usepackage{paralist}%
\usepackage{calligra} %
\usepackage{amsthm}
\usepackage{stmaryrd}

\usepackage{caption, subcaption}

\DeclareFontFamily{U}{mathb}{\hyphenchar\font45}
\DeclareFontShape{U}{mathb}{m}{n}{
      <5> <6> <7> <8> <9> <10> gen * mathb
      <10.95> mathb10 <12> <14.4> <17.28> <20.74> <24.88> mathb12
      }{}
\DeclareSymbolFont{mathb}{U}{mathb}{m}{n}

\DeclareMathSymbol{\sqsubsetneqq}   {3}{mathb}{"90}
\DeclareMathSymbol{\sqsupsetneqq}   {3}{mathb}{"91}

\newrgbcolor{lightblue}{.75 0.85 1}
\newcmykcolor{lpink}{0 .2 0 0}

\DeclareMathAlphabet{\mathcalligra}{T1}{calligra}{m}{n}
\DeclareFontShape{T1}{calligra}{m}{n}{<->s*[2.2]callig15}{}

\newcommand{\prea}{\text{$\text{\sf\em PRE}_{1}$}}
\newcommand{\preb}{\text{$\text{\sf\em PRE}_{2}$}}
\newcommand{\tia}{\text{$\text{\sf\em TI}_{1}$}}
\newcommand{\tib}{\text{$\text{\sf\em TI}_{2}$}}

\newcommand{\lb}{\text{$\text{\sf\em LB}$}}
\newcommand{\co}{\text{$\text{\sf\em CO}$}}
\newcommand{\su}{\text{$\text{\sf\em SU}$}}
\newcommand{\ddnew}{\text{$\text{\sf\em DD}$}}

\newcommand{\pred}{\text {\sf\em{pred}\/}}
\newcommand{\boundary}{\text{\sf\em BI}\xspace}

\newcommand{\Start}[1]{\text{$\mbox{\sf\em Start}_{#1}$}}
\newcommand{\End}[1]{\text{$\mbox{\sf\em End}_{#1}$}}

\newcommand{\allocsite}{\text{\sf\em allocsite\/}\xspace}

\newcommand{\origin}{\text{\sf\em origin\/}\xspace}

\newcommand{\lhs}{\text{$\ell$}\xspace}
\newcommand{\rhs}{\text{\scalebox{1.15}{\slshape r}}\xspace\ }

\newcommand{\balpha}{\text{$\overline{\alpha}$}\xspace}
\newcommand{\bbeta}{\text{$\overline{\beta}$}\xspace}

\newcommand{\blhs}{\text{$\overline{\ell}$}\xspace}

\newcommand{\brhs}{\text{$\overline{\scalebox{1.15}{\slshape r}}$}\xspace}

\newcommand{\oldDin}{\text{\sf\em OldDin\/}\xspace}
\newcommand{\oldAout}{\text{\sf\em OldAout\/}\xspace}

\newcommand{\ain}{\text{\sf\em Ain\/}\xspace}
\newcommand{\aout}{\text{\sf\em Aout\/}\xspace}
\newcommand{\din}{\text{\sf\em Din\/}\xspace}
\newcommand{\dout}{\text{\sf\em Dout\/}\xspace}

\newcommand{\Succ}{\text{\sf\em succ\/}\xspace}
\newcommand{\Pred}{\text{\sf\em pred\/}\xspace}
\newcommand{\AKill}{\text{\sf\em Akill\/}\xspace}
\newcommand{\AGen}{\text{\sf\em Agen\/}\xspace}
\newcommand{\dkill}{\text{\sf\em Dkill\/}\xspace}
\newcommand{\dgen}{\text{\sf\em Dgen\/}\xspace}
\newcommand{\ldgen}{\text{\sf\em LDgen\/}\xspace}
\newcommand{\rdgen}{\text{\sf\em RDgen\/}\xspace}

\newcommand{\nodes}{\text{$\Sigma$}\xspace}

\newcommand{\base}{\text{\sf\em base\/}\xspace}
\newcommand{\Null}{\text{\sf\em null\/}\xspace}
\newcommand{\new}{\text{\sf\em new\/}\xspace}

\newcommand{\dn}{\text{\sf\em absName\/}\xspace}
\newcommand{\vars}{\text{$\mathbb{V}$\/}\xspace}
\newcommand{\flds}{\text{$\mathbb{F}$\/}\xspace}
\newcommand{\pflds}{\text{{\sf\em F}$_p$}\xspace}

\protect{\newlength{\BRACEL}}
\protect{\newlength{\TAL}} 
\newlength{\codeLineLength}
\newcommand{\codeLine}[4]{\darkgray\footnotesize\sf #1&%
\psframebox[framesep=0,fillstyle=solid,fillcolor=#4,
	linestyle=none]{\makebox[\codeLineLength][l]{%
	\rule[-.3em]{0em}{1.em}{\tt%
	 \NL{#2}{\text{#3}}}}}
\\ }

\usepackage{amsthm}
\newtheorem{thm}{Theorem}
\newtheorem{lemma}{Lemma}

\newtheorem{corollary}{Corollary}

\theoremstyle{definition}

\newcommand{\vname}{\text{\sf\em var}\xspace}
\newcommand{\aeset}{\text{$\mathbb{A}$}\xspace}
\newcommand{\pset}{\text{$\mathbb{P}$}\xspace}
\newcommand{\tset}{\text{$\mathbb{T}$}\xspace}
\newcommand{\oneae}{\text{$\alpha$}\xspace}
\newcommand{\addr}{\text{\sf\em isAddr\/}\xspace}
\newcommand{\addrexpr}{\text{\sf\em addrExpr}\xspace}
\newcommand{\addrTaken}{\text{\sf\em addrTaken}\xspace}

\newcommand{\aliasing}{\text{\sf\em A}\xspace}
\newcommand{\typing}{\text{\sf\em T}\xspace}

\newtheorem{exmp}{Example}
\newenvironment{example}{\exmp}{$\Box$}

\newcommand{\vecf}{\text{$\mathcal{F}$}\xspace}
\newcommand{\vecd}{\text{$\mathcal{D}$}\xspace}
\newcommand{\veca}{\text{$\mathcal{A}$}\xspace}
\newcommand{\bekic}{\text{Beki\'{c}}\xspace}
\newcommand{\dval}{\text{$\mathbf{D}$}\xspace}
\newcommand{\aval}{\text{$\mathbf{A}$}\xspace}
\newcommand{\mfp}{\text{\em MFP\/}\xspace}
\newcommand{\mop}{\text{\em MoP\/}\xspace}
\newcommand{\paths}[1]{\text{$\Pi_{#1}$}\xspace}
\newcommand{\fpaths}[1]{\text{$\overrightarrow{\Pi}_{#1}$}\xspace}
\newcommand{\bpaths}[1]{\text{$\overleftarrow{\Pi}_{#1}$}\xspace}

\newcommand{\mfpdval}{\text{$\dval^\text{\fp}$}\xspace}
\newcommand{\mfpaval}{\text{$\aval^\text{\fp}$}\xspace}
\newcommand{\mfpdvalseries}{\text{$\demand^\text{\fp}_n$}\xspace}
\newcommand{\mfpavalseries}{\text{$\aliasing^\text{\fp}_n$}\xspace}

\newcommand{\mfpvecd}{\text{$\mathcal{D}^\text{\fp}$}\xspace}
\newcommand{\mfpveca}{\text{$\mathcal{A}^\text{\fp}$}\xspace}
\newcommand{\mopdval}{\text{$\dval^\text{\op}$}\xspace}
\newcommand{\mopaval}{\text{$\aval^\text{\op}$}\xspace}
\newcommand{\mopdvalseries}{\text{$\demand^\text{\op}_n$}\xspace}
\newcommand{\mopavalseries}{\text{$\aliasing^\text{\op}_n$}\xspace}

\newcommand{\mopvecd}{\text{$\mathcal{D}^\text{\op}$}\xspace}
\newcommand{\mopveca}{\text{$\mathcal{A}^\text{\op}$}\xspace}
\newcommand{\eqmfpdval}[1]{\text{$\demand^\text{\fp}_{#1}$}\xspace}
\newcommand{\eqmfpaval}[1]{\text{$\aliasing^\text{\fp}_{#1}$}\xspace}

\newcommand{\eqmoptdval}[2]{\text{$\demand^{\,#1}_{#2}$}\xspace}
\newcommand{\eqmoptaval}[2]{\text{$\aliasing^{#1}_{#2}$}\xspace}
\newcommand{\eqmopdval}[1]{\text{$\demand^\text{op}_{#1}$}\xspace}

\newcommand{\ndd}{\text{$\delta$}\xspace}
\newcommand{\nda}{\text{$\gamma$}\xspace}
\newcommand{\gfunname}{\text{$\mathbf{X}$}\xspace}
\newcommand{\gfun}[2]{\text{$\dfv_{#1}^{\,#2}$}\xspace}
\newcommand{\gfunnode}[1]{\text{$\dfv_{#1}$}\xspace}
\newcommand{\gfunnodein}[2]{\text{\sf\em Xin$_{#1}^{\,#2}$}\xspace}
\newcommand{\gfunnodeout}[2]{\text{\sf\em Xout$_{#1}^{\,#2}$}\xspace}
\newcommand{\val}[1]{\text{$\mathbf{I}_{#1}$}\xspace}
\newcommand{\initval}[1]{\text{${\text{\sf\em I}}_{#1}$}\xspace}
\newcommand{\prog}{\text{$\mathcal{P}$}\xspace}
\psset{unit=1mm,arrowsize=1.5mm}
\newcommand{\tba}{\text{\sf\em tba\/}\xspace}
\newcommand{\asb}{\text{\sf\em asb\/}\xspace}
\newcommand{\id}{\text{\sf\em Id\/}\xspace}
\newcommand{\dd}{\text{\sf\em Dd\/}\xspace}
\newcommand{\jd}{\text{\sf\em Jd\/}\xspace}
\newcommand{\cd}{\text{\sf\em Cd\/}\xspace}
\newcommand{\ex}{\text{\sf\em Ex\/}\xspace}
\newcommand{\cm}{\text{\sf\em cm\/}\xspace}

\newcommand{\fp}{\text{\sf\em fp\/}\xspace}
\newcommand{\op}{\text{\sf\em op\/}\xspace}
\newcommand{\aid}{\text{$\aval^{\id}$}\xspace}
\newcommand{\add}{\text{$\aval^{\dd}$}\xspace}

\newcommand{\aex}{\text{$\aval^{\ex}$}\xspace}
\newcommand{\demand}{\text{\sf\em D}\xspace}

\newcommand{\ddd}{\text{$\dval^{\dd}$}\xspace}

\newcommand{\dcd}{\text{$\dval^{\cd}$}\xspace}

\newcommand{\afun}[2]{\text{$\aliasing_{#1}^{\,#2}$}\xspace}
\newcommand{\dfun}[2]{\text{$\demand_{#1}^{\,#2}$}\xspace}

\newcommand{\jop}{\text{\sf\em jp}\xspace}
\newcommand{\djoin}[2]{\text{$\demand_{#1}$}\xspace}

\newcommand{\afunin}[2]{\text{{\sf\em Ain}$_{#1}^{\;\,#2}$}\xspace}
\newcommand{\afunout}[2]{\text{{\sf\em Aout}$_{#1}^{\;\,#2}$}\xspace}
\newcommand{\dfunout}[2]{\text{{\sf\em Dout}$_{#1}^{\;\,#2}$}\xspace}
 \newcommand{\vecfid}{\text{$\mathcal{F^{\id}}$}\xspace}
 \newcommand{\vecfdd}{\text{$\mathcal{F^{\dd}}$}\xspace}
 \newcommand{\vecfjd}{\text{$\mathcal{F^{\jd}}$}\xspace}
\newcommand{\vecfcd}{\text{$\mathcal{F^{\cd}}$}\xspace}
\newcommand{\vecgex}{\text{$\mathcal{G^{\ex}}$}\xspace}

\newcommand{\restrict}[2]{\text{\sf Restrict ($#1$,$#2$)}} 
\newcommand{\Restrict}{\text{\sf Restrict}\xspace} 

\begin{document}

\editor{}

\title[Demand-driven Alias Analysis]{Demand-driven Alias Analysis : Formalizing Bidirectional Analyses for Soundness and Precision}


 \author{Swati Jaiswal}
 \email{swati.j@iitb.ac.in}
 \author{Uday P. Khedker}
 \email{uday@cse.iitb.ac.in}
 \author{Supratik Chakraborty}
 \email{supratik@cse.iitb.ac.in}
 
 \affiliation{%
\institution{Indian Institute of Technology Bombay}
\country{India}}





\begin{abstract}
A demand-driven approach to program analysis computes only the information that is
needed to serve a target demand. In contrast, an exhaustive approach computes all information in anticipation of
it being used later. 
Demand-driven methods have primarily been viewed as efficient algorithms as they compute only the information that is required 
to meet a given 
set of demands. However, for a given set of demands, they 
are believed to compute the same information that would be computed by the corresponding exhaustive methods. 
We investigate the precision and bidirectional nature of 
demand-driven methods and show that:
\begin{inparaenum}[(a)]
\item demand-driven methods can be formalized inherently as bidirectional data flow analysis, and
\item for some analyses, demand-driven method can compute more precise information than the corresponding
	exhaustive method.
\end{inparaenum}

The formalization as a bidirectional analysis follows because
the demands are propagated against the control flow and the information to satisfy the demands is
propagated along the control flow. 
We extend the formalization of the \emph{Meet Over Paths} solution to bidirectional flows, 
by 
introducing the concept of \emph{qualified control flow paths}
to explicate the forward and backward flows.
This formalization helps us to prove the soundness and precision of our analysis.

The approximation caused by data abstraction used for heap locations
(e.g. allocation-site-based abstraction or type-based abstraction) is a source of imprecision 
in pointer analysis.
Since a demand-driven method computes information for a smaller set of demands,
it has to deal with less imprecision caused by the data abstractions.
This could lead to more precise results than an exhaustive method.
We show that while this is indeed the case for Java, for C/C++, the precision
critically hinges on how indirect assignments are handled. We use this insight and propose a demand-driven alias
analysis that is more precise than an exhaustive analysis for C/C++ too.
We have chosen static resolution of virtual function calls as an application 
to demonstrate the precision gain of our demand-driven alias analysis for C++.

Our measurements show that our
method is more precise and more efficient (for both allocation-site-based and type-based abstractions) than both, the existing demand-driven method, as well as 
the corresponding exhaustive method. This precision is measured in terms of
the number of monomorphic call-sites, 
the number of virtual call edges, and 
the number of class types discovered by the methods. 
\end{abstract}

\maketitle







\section{Introduction} \label{sec:intro}

Pointer analysis has received a lot of attention
because most analyses and applications 
need to disambiguate indirect manipulation of data and control that arises in the presence of pointers.
Some applications of pointer analysis require high efficiency
whereas some other applications need high precision.
Some applications require information of all pointers 
which requires an \textit{exhaustive method}.
Some other applications need only partial information which can be computed using \textit{demand-driven method}.
This paper investigates two aspects of demand-driven methods 
that do not seem to
have been explored so far.
\begin{itemize}
 \item 
Demand-driven methods compute only the information that is
required to meet a set of demands rather than all possible information. Hence,
they have always been thought of as efficient algorithms.
They have also been known to compute the same information as computed by the corresponding exhaustive methods to serve a demand.
We show that, in the presence of data abstraction, demand-driven methods can compute more precise information than the corresponding
exhaustive methods. In other words, they are both more efficient and more precise than exhaustive methods in some situations.

 \item Demand-driven methods are inherently bidirectional
with demands being raised and propagated against the control flow and the requisite information to
satisfy the demands flowing along the control flow. We formalize them as 
bidirectional analyses by defining a general concept of \emph{meet over paths} (\mop) solution
for bidirectional flows.
\end{itemize}

\subsection{Precision Gain in Demand-driven Methods}
The two dimensions - \textit{precision} and \textit{efficiency} play a major role in selecting the 
right mix of features for any program analysis.
We classify the factors governing precision and efficiency of 
pointer analysis as 
\begin{inparaenum}[(i)]
 \item control flow abstractions,
 \item data abstractions, and
 \item the quantum of information required.
\end{inparaenum}

\emph{Control flow abstractions} govern the over-approximation introduced by an analysis in the flow of control 
by striking a balance between precision and efficiency.
The intraprocedural control flow is abstracted in terms of
flow sensitivity (honoring control flow and computing distinct information for each control flow point)
or flow insensitivity (ignoring control flow and computing gross information common to all program points).
The interprocedural control flow is abstracted in terms of context sensitivity (distinguishing between
different calling contexts of a procedure) or context insensitivity (treating all calling contexts alike).

\emph{Data abstraction} governs 
the over-approximation introduced by an analysis in order to model objects on the stack or heap.
Generally heap is represented in terms of access paths~\cite{
Landi:1992:SAA:143095.143137,Khedker:2007:HRA:1290520.1290521}, 
allocation-site-based abstraction~\cite{Lhotak:2003:SJP:1765931.1765948, 
	Milanova:2005:POS:1044834.1044835,lhotak, Sridharan:2006:RCP:1133981.1134027},
or type-based abstraction~\cite{Palsberg:1991:OTI:117954.117965, Diwan:2001:UTA:383721.383732,Diwan:1998:TAA:277652.277670}.

The third factor governing the precision and efficiency of pointer 
analysis is the 
\emph{quantum} of the information required from an analysis.
If the analysis desires information of all the pointers, 
it can be computed using an exhaustive method.
If the information desired by an analysis is sparse, it can be computed using demand-driven method.
Some examples when such a partial information may be required are:
\begin{inparaenum}[\em(a)]
 \item by a client (like user of a debugger or a slicer) 
 ~\cite{Sridharan:2006:RCP:1133981.1134027, Guyer:2003:CPA:1760267.1760284, Sridharan:2005:DPA:1094811.1094817, Yan:2011:DCA:2001420.2001440}, 
 \item by an application (like resolving indirect calls~\cite{Agrawal:2002:EDD:647478.727927,Heintze:2001:DPA:378795.378802,Sridharan:2005:DPA:1094811.1094817},
taint analysis~\cite{spth_et_al:LIPIcs:2016:6116,Arzt:2014:FPC:2594291.2594299,Huang:2016:DSD:2950290.2950348}), or
 \item by an analysis (like compute 
information only for live data~\cite{lfcpa} 
or for an incremental change in the source program~\cite{Saha:2005:IDP:1069774.1069785}).
\end{inparaenum}

%

The effect of 
control flow abstractions and data abstractions on the precision and efficiency of pointer analysis
has been looked at in great detail in the literature.
While there has been a lot of work on demand-driven methods, almost all of it is primarily motivated by 
the quest for efficiency and precision is
assumed to be equivalent or less than the exhaustive 
counterpart~\cite{Heintze:2001:DPA:378795.378802, Zheng:2008:DAA:1328438.1328464, Duesterwald:1997:PFD:267959.269970,Agrawal:2002:EDD:647478.727927, Sridharan:2005:DPA:1094811.1094817}.

We argue that in the case of pointer
analyses, a demand-driven method could be more precise than
an exhaustive method in the presence of data abstraction. Although this has not been reported before, 
it is easy to see: 
since a demand-driven method computes only the information required to meet the demands, 
the imprecision
caused by data abstraction could also reduce. However, we observe the following subtlety for 
C/C++: in the presence of
indirect assignments, the conventional demand-driven methods (see Sections~\ref{sec:compare:cd:ex} and~\ref{sec:related.work})
fail to benefit from the possibility of reduced imprecision.
We identify the exact cause of this loss of precision
and this guides us to a different nature of raising demands that restores the relative precision gain
of a demand-driven method over an exhaustive method (Example~\ref{exmp:our.speculation} in Section~\ref{sec:key.idea}).

We choose static resolution of virtual function calls as an application and show that our
demand-driven method is also more efficient than the conventional demand-driven method 
apart from being more precise. 

\subsection{Bidirectional Nature of Demand-driven Methods}
At a more general level, we observe that demand-driven methods are inherently bidirectional
with demands being raised and propagated against the control flow and the requisite information to
satisfy the demands flowing along the control flow. Conventionally, these
bidirectional dependencies in demand-driven method have been viewed as a characteristic of an 
efficient algorithm that computes 
information related to the given set of demands.
In other words,
a demand-driven method is seen as two disparate analyses traversing the program in opposite directions
and whose interaction is controlled by an algorithm that examines the information computed by these analyses and
decides the analysis to be invoked on a need basis.
We move this dependency from an algorithm to 
the specification of the analysis. This is achieved by defining the interactions declaratively in terms of
a bidirectional data flow analysis and by formalizing the concept of the \emph{Meet Over Paths} Solution (\mop) 
of general bidirectional analyses. 

We also show that the corresponding \emph{Maximum Fixed Point} 
Solution (\mfp) is a sound approximation of the \mop. Both these have been known for 
unidirectional analyses~\cite{Kildall:1973:UAG:512927.512945,Nielson:1999:PPA:555142,Khedker:2009:DFA:1592955} and have been 
formally defined for them. The \mop for bidirectional analyses has been
formally defined only for bit vector frameworks in the context of partial redundancy elimination~\cite{Khedker:1994:GTB:186025.186043}.
Our formalization is applicable not only to all the demand-driven methods but also
to the liveness-based points-to analysis that explicitly uses
bidirectional flow in the same analysis~\cite{lfcpa}.

\subsection{Duality of Alias and Points-to Analysis}
We view pointer analysis to comprise of two different 
but related analyses that are duals of each other:
points-to analysis and alias analysis. 
If locations are named, which is the case with our data abstractions,
then points-to information can be derived from aliases and vice-versa.
However, in some cases, the alias view of the information is more convenient
whereas in some other cases, the points-to view of the same information is more convenient.

We formulate our analysis as a bidirectional alias analysis where
demands are raised against the control flow and aliases for the demand
raised are propagated along the control flow.
Formulation in terms of alias analysis is done for convenience and
to ease the readability.
It is easy to find the aliases of a pointer using an alias analysis.
Such information can be computed from a points-to graph as well,
 but it is not directly available.
Instead of cluttering the formulation to compute points-to graph
and compute aliases from them, we compute aliases directly.
We view these two analyses as similar, with the difference in 
representation.

Alias pair $(x, \&y)$ is nothing but a points-to edge of the 
form $(x, y)$ (i.e. $x$ points-to $y$).
When it is easy to view an alias pair as a points-to edge
we refer to it in the form of a points-to edge.
While explaining the motivating example, we present points-to graph
as it is more convenient form of representation and 
is very easy to understand.
However, understanding the source of imprecision in terms of aliases is
much easier.
Hence, for convenience, we
use aliases and points-to graph interchangeably.

\subsection{Our Contributions and Organization of the Paper}

The main contributions of this work are:

\begin{itemize}
\item 
We present a demand-driven alias analysis method with data abstraction for static resolution of virtual function calls.
\item We show that a slight shift in the nature of demands raised could increase the precision of
      our analysis significantly for languages like C and C++.
\item We formalize our demand-driven method as a bidirectional data flow analysis.
\item We prove that our analysis is sound and is more precise than the 
exhaustive method and 
the conventional demand-driven methods. 
\item We present compelling empirical evidence to show the precision and efficiency of our method.
\item 
We formalize \mop solution for general bidirectional analyses by introducing the concept of \emph{qualified control flow paths}
to explicate the forward and backward flows.
We also establish that \mfp is a sound approximation of \mop for bidirectional analyses.
\end{itemize}

Our empirical measurements shows that our proposed demand-driven method
is far efficient as compared to the exhaustive method
with a speedup factor greater than two in most cases.
Our measurements also show the precision of
the conventional demand-driven method and the exhaustive method is identical in terms of
the number of monomorphic call-sites identified, the number of virtual call edges discovered and the class types identified for
the objects used in the program. 
This is in concurrence with our formal proof that the precision of the two methods is
equivalent.

Our method out performs these two methods in terms of precision for both allocation-site-based and type-based abstractions:
The number of monomorphic call-sites discovered by our method is larger, 
the number of virtual call edges discovered in the call graph by our method is smaller and 
the number of class types discovered by our method is far smaller.
It is interesting to note that with type-based abstraction, 
we identify 24\% fewer types in 7 cases with the reduction increasing significantly to
 50\% and 60\%  for \texttt{motti} and \texttt{dealII} programs respectively.
 Similar traits are seen when allocation-site-based abstraction is
used---the reduction in the number of types is 39\% and 49\% 
 for \texttt{gperf} and \texttt{gengetopt} programs respectively.

The rest of the paper is organized as follows:
Section~\ref{sec:causes.of.imprecision} discusses the interplay between data abstraction and demand speculation.
It also presents an overview of our key idea to mitigate the imprecision caused by them. 
Section~\ref{sec:formulation} formulates our proposed demand-driven approach
as a data flow analysis.
Section~\ref{sec:formalizing-bidirectional} presents a big picture view of bidirectional analyses 
and introduces a generalization of control flow paths for bidirectional data flows
which is then used for defining \mfp and \mop for bidirectional data flow analyses.
The subsequent sections then return to our method:
Section~\ref{sec:mop:demand} 
instantiates the general \mfp and \mop to our method.
Section~\ref{sec:soundness-proof} and Section~\ref{sec:precision-proof} 
use the generalization in control flow paths to formally prove the
soundness and precision of our proposed analysis.
Section~\ref{sec:measurements} presents the empirical results. 
Section~\ref{sec:related.work} describes the related work.
Section~\ref{sec:conclusions} concludes the paper.

\section{The Roles of Data Abstraction and the Nature of Demands in Pointer Analysis}
\label{sec:causes.of.imprecision}

This section examines the interplay between data abstraction objects and the nature of demands. Our observations
lead to our key ideas that allow us to restore the precision gain of a demand-driven method over
an exhaustive method.

\begin{figure}[!t]
\centering
\begin{tabular}{@{}c|c@{}}
\begin{tabular}{@{}c@{}}

\setlength{\codeLineLength}{35mm}
	\begin{tabular}{@{}r@{}c}
	\codeLine{}{0}{X **p,\,**q; \rule{0em}{1.25em}}{white}
	\codeLine{}{0}{X *t,\,*x,\,*y,\,*z; \rule{0em}{1.25em}}{white}
\\
	\codeLine{03}{1}{{q = \&z;}}{white}
	\codeLine{04}{1}{{p = \&z;}}{white}
	\codeLine{05}{1}{{x = new X;}}{white}
	\codeLine{}{1}{$\vdots$}{white}
	\codeLine{14}{1}{{y = new X;}}{white}
	\codeLine{15}{1}{{*p = x;}}{white}
	\codeLine{}{1}{$\vdots$}{white}
	\codeLine{23}{1}{{x->f = new Y;}}{white}
	\codeLine{24}{1}{{y->f = new Z;}}{white}
	\codeLine{}{1}{$\vdots$}{white}
	\codeLine{27}{1}{{t = z->f;}}{white}
	\codeLine{28}{1}{{t->vfun ();}}{white}
	
	\end{tabular}
	\\
	(a) Example Program
	\rule{0em}{1.75em}
	\end{tabular}
& 
\begin{tabular}{@{}l@{}}
 	\begin{pspicture}(0,0)(30,22)
	\psset{nodesep=-1.5}
	
	\putnode{z}{origin}{12}{18}{\pscirclebox[linestyle=none]{$z$}}
	\putnode{p}{z}{-10}{0}{\pscirclebox[linestyle=none]{$p$}}
	\putnode{x}{z}{0}{-7}{\pscirclebox[linestyle=none]{$x$}}
	\putnode{y}{x}{0}{-7}{\pscirclebox[linestyle=none]{$y$}}

	\putnode{A}{x}{10}{0}{\pscirclebox[linestyle=none]{$X$}}
	\putnode{B}{A}{10}{5}{\pscirclebox[linestyle=none]{$Y$}}
	\putnode{C}{A}{10}{-5}{\pscirclebox[linestyle=none]{$Z$}}
	\putnode{t}{A}{20}{0}{\pscirclebox[linestyle=none]{$t$}}
	
	\ncline[nodesep=-1]{->}{p}{z}
	\ncline[nodesep=-1]{->}{x}{A}
	\ncline[nodesep=-1]{->}{y}{A}
	\ncline[nodesep=-1]{->}{z}{A}
	\ncline[nodesep=-1]{->}{A}{B}
	\aput[1pt](.4){$f$}
	\ncline[nodesep=-1]{->}{A}{C}
	\bput[1pt](.4){$f$}
	\ncline[nodesep=-1]{->}{t}{B}
	\ncline[nodesep=-1]{->}{t}{C}
	\end{pspicture} \\
	(b) Demand-driven analysis with\\
	 conventional speculation strategy \rule[-1.5em]{0em}{1em}
	\\\hline
		 	\begin{pspicture}(0,0)(30,22)
	\psset{nodesep=-1.5}
	
	\putnode{z}{origin}{12}{13}{\pscirclebox[linestyle=none]{$z$}}
	\putnode{p}{z}{-10}{5}{\pscirclebox[linestyle=none]{$p$}}
	\putnode{q}{z}{-10}{-5}{\pscirclebox[linestyle=none]{$q$}}
	\putnode{x}{z}{0}{-10}{\pscirclebox[linestyle=none]{$x$}}
	\putnode{t}{x}{10}{0}{\pscirclebox[linestyle=none]{$t$}}

	\putnode{A}{x}{10}{5}{\pscirclebox[linestyle=none]{$X$}}
	\putnode{B}{A}{10}{0}{\pscirclebox[linestyle=none]{$Y$}}
	
	\ncline[nodesep=-1]{->}{q}{z}
	\ncline[nodesep=-1]{->}{p}{z}
	\ncline[nodesep=-1]{->}{x}{A}
	\ncline[nodesep=-1]{->}{t}{B}
	\ncline[nodesep=-1]{->}{z}{A}
	\ncline[nodesep=-1]{->}{A}{B}
	\aput[1pt](.4){$f$}
	\end{pspicture}
	\\
	(c) Demand-driven analysis with\\ 
	our proposed speculation strategy
\end{tabular}

\end{tabular}
\caption{Points-to graph using type-based abstraction for static resolution of
virtual call at line 28. 
Objects and allocation sites are annotated by their respective 
types in points-to graph shown in (b) \& (c).
Virtual function \emph{vfun} is defined in all the classes $\{X,Y,Z\}$.
Member $f$ is pointer to class $X$ and is declared in class $X$.
Class hierarchy is $X \rightarrow Y \rightarrow Z$.
Detailed working of both the kinds of speculations is described in the Appendix~\ref{detailed-working}.
}
\label{m-eg-type}
\end{figure}

\subsection{Data Abstraction and its Effect on Precision in Pointer Analysis}

A pointer analysis needs to employ data abstraction to represent an unbounded heap 
which is generally represented by
a store-based model or a storeless model~\cite{Kanvar:2016:HAS:2966278.2931098}. 
A storeless model names locations in terms of access paths that are sequences of field names
following a variable. Each of these paths correspond to paths in the memory graph and hence the set of
access paths needs an explicit summarization to be bounded.
A store-based model names memory locations 
using allocation-site-based abstraction
or type-based abstraction. Since this bounds the number of heap locations, it
obviates the need for any other summarization.

We focus on these two abstractions for a store-based model.
An \emph{allocation-site-based abstraction}\footnote{Although the objects on stack are named using variable names (because no
heap allocation statements are involved), we ignore this
minor distinction and continue to call the abstraction as allocation-site-based abstraction.}
uses variable names for objects on stack and allocation site names for objects in 
heap~\cite{Lhotak:2003:SJP:1765931.1765948, Milanova:2005:POS:1044834.1044835,lhotak, Sridharan:2006:RCP:1133981.1134027}.
A \emph{type-based abstraction} names objects in terms of 
types for objects on both stack and heap~\cite{Palsberg:1991:OTI:117954.117965, Diwan:2001:UTA:383721.383732,
 Diwan:1998:TAA:277652.277670}.
In allocation-site-based abstraction, all objects created at the same allocation site are treated alike under the
assumption that they are likely to be used alike. In a type-based abstraction, all objects with the same type are treated alike.
It is easy to see that a type-based abstraction is more imprecise compared to allocation-site-based abstraction; however,
it is far more efficient to compute. It suffices for type dependent clients like call graph construction, virtual call resolution,
may-fail cast etc. as discussed in~\cite{Tan:2017:EPP:3062341.3062360} where its coarseness across different objects of the same type
does not matter.


We consider static resolution of virtual function calls as our application.
A demand for our approach originates at the point of a virtual function call.
For this type-dependent application, we explain the concepts using type-based abstraction to model objects on heap and stack.


\begin{example}
\label{exmp:type.imprecision}
(\emph{Imprecision caused by data abstraction}).
Consider the program in Figure~\ref{m-eg-type}(a) 
for pointer analysis with type-based abstraction.
Both $x$ and $y$ point to an object of type $X$ resulting in
$x$ and $y$ being considered as aliases
even when they are not aliases.
This leads to the following imprecision:
$x\rightarrow f$ may point-to an object of type $Z$ (apart from an
object of type $Y$) and $y\rightarrow f$ may point-to an object of type $Y$ (apart from an object of type $Z$).
It is important to note that this imprecision is not introduced by a control flow abstraction but is purely due to
the modelling of data abstractions.
If we can avoid computing the information of $x$ and $y$ together, they will not be considered as aliases,
avoiding the imprecision.
\end{example}

It is easy to see that data abstraction introduces redundant aliases resulting in imprecision.
Similar imprecision could also be introduced by allocation-site-based abstraction when two pointers point to the same allocation-site,
introducing redundant aliases. A coarser abstraction leads to a larger imprecision.

%

For the above example, we view precision in two different dimensions, which are:
\begin{itemize}
 \item  an \emph{aliasing} relation \text{$\aliasing \subseteq \pset \times \pset$}
		between pointer expressions, and
 \item a \emph{typing} relation \text{$\typing \subseteq \pset \times \tset$} between pointer expressions and types.
\end{itemize}

We explain this by making the following distinction between two kinds of aliases.
Two aliased pointer expressions are \emph{node aliases} when their subexpressions are not aliased; otherwise they are
\emph{link aliases}~\cite{Khedker:2007:HRA:1290520.1290521}. Then,
the data abstraction is represented by its influence on \aliasing and \typing which can be described as
\begin{quote}
Data abstraction causes (a sound) over-approximation of \aliasing by introducing redundant node aliases
which leads to (sound) over-approximation in the link aliases.
Such an over-approximation in node aliases does not over-approximate \typing, but
over-approximation in link aliases results in (sound) over-approximation of \typing.
\end{quote}

\begin{example}
\label{exmp:type.imprecision.1}
(\emph{The role of node and link aliasing in imprecision}).
In Example~\ref{exmp:type.imprecision}, \typing contains 
\text{$(x,X)$} and \text{$(y,X)$} which is precise; however it causes 
$x$ and $y$ to be identified as node aliases which is imprecise. This imprecision 
leads to \text{$x\rightarrow f$} and \text{$y\rightarrow f$} being spuriously considered as link aliases.
Since we have \text{$(x\rightarrow f, Y)$} and \text{$(y\rightarrow f, Z)$} in \typing, link aliasing
causes inclusion of \text{$(x\rightarrow f, Z)$} and \text{$(y\rightarrow f, Y)$} too.
This suggests that
both $x\rightarrow f$  and $y\rightarrow f$ 
hold objects of types $Y$ and $Z$ which is imprecise.
\end{example}


\subsection{The Effect of Speculation on the Precision of Demand-driven Pointer Analysis}
\label{sec:speculation}

We first explain the need of speculation and then describe how it causes imprecision in pointer analysis. 
This discussion forms the basis of our key idea in Section~\ref{sec:key.idea}.

\begin{figure}[!t]
\centering
\begin{tabular}{@{}c|c@{}}
\begin{tabular}{@{}c@{}}
\setlength{\codeLineLength}{25mm}
	\begin{tabular}{@{}r@{}c}
	\codeLine{}{1}{X **p,\,\,*z,\,\,a; }{white}
\\
	\codeLine{01}{1}{{p = \&z;}}{white}
	\codeLine{02}{1}{{*p = \&a}}{white}
	\codeLine{03}{1}{{z->vfun ();}}{white}	
	\end{tabular}
	\end{tabular}
& 
\begin{tabular}{@{}c@{}}
\setlength{\codeLineLength}{33mm}
	\begin{tabular}{@{}r@{}c}
	\codeLine{}{1}{X *x,\,\,*y,\,\,*z,\,\,a; }{white}
\\
	\codeLine{01}{1}{{y = \&a;}}{white}
	\codeLine{02}{1}{{x = \&a;}}{white}
	\codeLine{03}{1}{{y->f = \ldots}}{white}
	\codeLine{04}{1}{{z = x->f;}}{white}	
	\codeLine{05}{1}{{z->vfun ();}}{white}	
	\end{tabular}
	\end{tabular}
\\
(a) & (b)
	\end{tabular}

\caption{The need of speculation for meeting a demand in the presence of indirect pointer assignments 
}
\label{need-for-speculation-eg1}
\end{figure}


A demand-driven approach seeks to serve a target demand by computing as little information as possible
but it needs to raise further demands internally to meet the target demands.
Most of the internal demands are a consequence of an assignment statement that creates
aliases of the target or other internal demands. However, in some cases
in C/C++, a demand may have to be raised in order to \emph{find}
an alias of the target demand and not because an alias of the target demand
has been found. We call such a demand, a \emph{speculated} demand.
It is possible that eventually such a demand may turn out to be
irrelevant; worse still, it may cause imprecision because of data abstraction. 
This speculation
is required due to the use of the \emph{address-of} operator `\&'.

\begin{example}
\label{exmp:demand.imprecision.1}
\emph{(Conventional speculation to ensure soundness)}.
For the program in Figure~\ref{need-for-speculation-eg1}(a), 
raising a demand for $z$ is not sufficient to identify its aliases because there is
no assignment in which $z$ appears on LHS or RHS. Further, the demand for $z$ cannot be propagated across 
indirect assignments (such as statement 02) because such an assignment could have
a side-effect of defining $z$ (which is indeed the case in our example).
In order to handle such a situation,
the conventional speculation strategy
kills the demand propagation of $z$ and generates a demand for $p$ at line 02
to ensure soundness.

For the program in Figure~\ref{need-for-speculation-eg1}(b), we wish to identify the pointees of $x\rightarrow f$
in order to resolve the virtual function call at line 05.
The fact that $(x,y) \in \aliasing$ will not be discovered by raising demand for $x$ alone
because the alias is not created by a direct assignment between $x$ and $y$ but by assigning $\&a$ to both $x$ and
$y$ separately on lines 01 and 02.
In the absence of alias \text{$(x,y)$}, the effect of the indirect assignment of $x\rightarrow f$ in line 03
cannot be incorporated.
Thus, it is essential to 
perform some form of speculation for soundness.
The conventional speculation kills the demand for $x\rightarrow f$ and generates a demand for $y$. 
The demands for $x$ and $y$ help identify the alias relationship between them, thereby meeting
the demand of $x\rightarrow f$ soundly.
\end{example}

In the absence of `\&' operator in Java, no speculation is required
making a demand-driven method for Java more precise than an exhaustive method.


The conventional speculation raises an internal demand at each indirect assignment. Some of these demands could be
irrelevant thereby leading to
an avoidable loss of precision caused by data abstraction.

\begin{example}
\label{exmp:demand.imprecision.2}
(\emph{Conventional speculation causes imprecision}).
Consider the program in Figure~\ref{m-eg-type}(a).
In order to resolve the virtual function call at line 28, we wish to find out the pointees of $t$
which in turn depends on the pointees of $z\rightarrow f$ at line 27.
We thus raise demand for $z$ and $z\rightarrow f$.
Conventional speculation will in turn raise demand for $y$, $x$ and $p$ because of the indirect assignment
statements at line 24, 23 and 15 respectively.
Raising demand for $x$ and $y$ simultaneously introduces redundant node alias between them, 
which leads to imprecision in the link aliases as discussed in Example~\ref{exmp:type.imprecision.1}.
This further contributes to the imprecision by introducing aliasing between
$x\rightarrow f$ and $y\rightarrow f$. As a consequence, 
$x\rightarrow f$ spuriously points-to $Z$ and $y\rightarrow f$ spuriously points-to $Y$ as 
shown in the Figure~\ref{m-eg-type}(b).
\end{example}

\subsection{Our Key Idea}
\label{sec:key.idea}

Note that in the above examples, data abstraction introduced a spurious {node} alias between two variables which caused
spurious {link} aliases between their fields.
In order to eliminate the imprecision in link aliases, node aliases should be computed precisely.
Redundant node aliases can be eliminated if a redundant speculation at the indirect assignments can be avoided.
This can be achieved if we know the alias relationship and raise demand only when required.

Consider the situation when a demand for $z$ is raised. Such a demand seeks the pointees of $z$
which can be identified by taking into account the effect of the statements $z = \ldots$ or $\ldots = z$.
Since the demand $z$ does not seek its pointers,
 it will miss on capturing the effect of the statement $\ldots = \&z$.
The effect of such a statement, if required, can be captured by raising speculated demands for indirect assignment statements.

We propose an alternative speculation strategy where we additionally raise a demand for the address-of a variable
involved in the target demand.

\begin{example}
\label{exmp:our.speculation}
(\emph{Raising demand for address-of a variable leads to precision)}.
For the program in Example~(\ref{exmp:demand.imprecision.2}) we wish to identify the alias relation $(p, \&z)$.
Conventional speculation achieves it by raising demand for $p$.
We propose an alternative speculation by raising demand for $\&z$. 
Demand for $\&z$ will identify the alias relation $(p, \&z)$ without having to raise the demand for $p$
because a demand for $\&z$ is unaffected by an indirect assignment.
Thus in order to identify an alias relation, raising demand for either value of the alias pair should suffice.

It is easy to see that if there is an indirect assignment statement which influences our target demand, 
it will be captured by raising demand for address-of a variable involved in meeting the target demand;
invalidating the need to raise demand at each indirect assignment statement.

For the program in Example~(\ref{exmp:demand.imprecision.2}),
we raise a demand for $\&z$ apart from $z$ thereby seeking to find pointers to $z$.
This captures the effect of statement 04 without having to raise the demand at 
indirect assignments (statements 15, 23, and 24) thereby 
eliminating the speculation of the conventional approach and yet finding that
$p$ points-to $z$. Thus we avoid computing the pointer information for $y$ thereby eliminating the imprecision
and computing precise pointer information as shown in Figure~\ref{m-eg-type}(c).
Demand for $\&z$ also identifies that $q$ points-to $z$ at line 03.
This does not over-approximate \typing as 
$\&z$ seeks for a pointer
and not an object.
It also does not over-approximate \aliasing
as $p$ aliased to $q$ is a precise aliasing relation.
However, it is a redundant information as manipulation to $p$ or its further links does not take place through $q$.
\footnote{We propose a solution to eliminate computation of such redundant information in Section~\ref{sec:seek-pointers}.}
\end{example}

Both speculations compute sound results but our speculation is more precise than 
the conventional speculation in terms of \typing because:
\begin{quote}
Our speculation raises a demand seeking the pointer of a variable whereas the conventional speculation raises a demand
seeking the pointee of a (pointer) variable. The main difference between the two is that a pointer \emph{cannot} be an object
whereas the latter could be an object. Thus our speculation avoids the imprecision caused by type-based abstraction
while the conventional speculation cannot do so.
\end{quote}

Our proposed speculation strategy for demand-driven method is more precise as compared to 
\begin{inparaenum}[(i)]
\item exhaustive method and 
\item demand-driven method with the conventional speculation.
\end{inparaenum}
In the absence of `\&' operator in Java, speculation is not required.
In general, demand-driven method for Java is more precise than the exhaustive method. We prove this formally in Section~\ref{sec:precision-proof}.

\section{Formulating Demand-driven Method With Improved Speculation Strategy}\label{sec:formulation}
Our proposed demand-driven method is bidirectional: demands are raised in a backward flow whereas aliases of the demands raised are computed
in a forward analysis---this may in turn require raising more demands, and the process continues until no further demands
are raised or no further aliases are computed.

\subsection{Explanation of Basic Concepts}
We discuss the concept of abstract name and aliases computed by
our analysis.

\subsubsection*{Access expressions}

Let \vars be the set of variables, \text{$\pset\subseteq \vars$} be the
set of pointer variables, and \tset be the set of types (which is the set of classes for our application) involved in
class hierarchies containing virtual functions. Let \flds be the set of
field members of classes in \tset.
We assume the program to be in 3-address code and the assignment statements consist 
of  \emph{access expressions} appearing in the left hand side \lhs and the right hand side
\rhs.   Let \text{$x \in \vars$},
\text{$y \in \vars$},
\text{$f \in \flds$}, and
\text{$\tau \in \tset$}. Then, \lhs and \rhs{} are defined as:
\begin{align*}
\lhs & :=  x \mid *x \mid x\rightarrow f \mid x.f 
\\
\rhs & :=  y \mid *y \mid y\rightarrow f \mid y.f \mid \&y \mid \new\ \tau \mid \Null
\end{align*}

We restrict ourselves to the
access expressions appearing in the program. 
We use the following two functions that extract
parts of an access expression $\oneae$. For mathematical convenience, they
compute either a singleton set or $\emptyset$.

\begin{itemize}
\item Function $\vname(\oneae)$ identifies the variable in $\oneae$.
	For access expressions $x$, $x\rightarrow f$, $*x$, $x.f$, and $\&x$,
	$\vname(\oneae)$ is $\{x\}$; for \text{$\new\ \tau$} and \Null, it is $\emptyset$.

\item Function $\base(\oneae)$  identifies the pointer variable dereferenced by
      $\oneae$.  For access expressions $x\rightarrow f$ and $*x$,
      $\base(\oneae)$ is $\{x\}$; for $x$, $\&x$, $x.f$, \text{$\new\ \tau$} and \Null, it is $\emptyset$.
\end{itemize}
Besides, predicate $\addr(\oneae)$ holds if $\oneae$ is of the form $\&x$.
Predicate $\addrTaken(\oneae)$ holds if $\oneae$ which is of the form $\&x$ occurs in the program.

\subsubsection*{Aliases}


We compute aliases of access expressions to find out all possible
ways of accessing a location using the access expressions appearing in the program. Let \aeset
denote the set of all possible access expressions that could occur in a three-address version of the program. Clearly, \aeset is
finite. We use two  different views of aliases that are semantically equivalent: First, aliasing is a relation
\text{$\aliasing \subseteq \aeset \times \aeset$}. In this case, aliasing is a set of alias pairs of access expressions of the
kind \text{$(\oneae_1,\oneae_2)$}. Alternatively, and sometimes this view is more convenient, aliasing
is a function \text{$\aliasing: \aeset \mapsto 2^\aeset$}. In this case, we can identify all access expressions that are aliased to 
a given access expressions. We use the following notation with an alias relation $\aliasing$:
\begin{align}
\aliasing(\oneae_1) & =
		\left\{ \oneae_2 \mid (\oneae_1,\oneae_2) \in \aliasing, \oneae_2 \in \aeset \right\}
		\label{eq:aliasclosure1}
	\\
\aliasing(X) & =
		\left\{ \oneae_2 \mid (\oneae_1,\oneae_2) \in \aliasing, \oneae_1 \in X, \oneae_2 \in \aeset \right\}
		\label{eq:aliasclosure2}
\end{align}

We do not compute all possible aliases. We compute aliases on demand where demands for access expressions are raised
in terms of abstract names.

\subsubsection*{Abstract name}
\begin{table*}
\caption{Abstract names for different access expressions}
\label{table-dn}
\footnotesize
 \begin{tabular}{|c|c|c|c|c|c|c|c|c|}
\hline
$\alpha$	
	& $x$
	& $\&x$
	& $\&a$
	& $*y$
	& $x\rightarrow f$
	& $a.f$
	& new $\tau$
	& null
		\\ \hline\hline

$\dn (\oneae, \aliasing)_{\text{\em\sf asb}}$
\rule[-.6em]{0em}{1.6em}
	& $\{ x\}$
	& $\{ \&x\}$
	& $\{ \&a\}$
	& $\{ x\mid (y, \&x) \in \aliasing\}$
	& $\{ a.f\mid (x, \&a) \in \aliasing\}$
	& $\{a.f\}$
	& $\{\&\allocsite\}$
	& $\emptyset$
		\\ \hline
$\dn (\oneae, \aliasing)_{\text{\em\sf tba}}$
\rule[-.6em]{0em}{1.6em}
	& $\{ x\}$ 
	& $\{ \&x\}$
	& $\{\&\tau\}$
	& $\{x \mid (y, \&x) \in \aliasing\}$
	& $\{ \tau.f\mid (x, \&\tau) \in \aliasing\}$ 
	& $\{\tau.f\}$
	& $\{\&\tau\}$
	& $\emptyset$
		\\ \hline
 
 \end{tabular}

\end{table*}
We compute demand in terms of abstract names. 
The abstract names of an access expression  $\alpha$
are normalized access expressions obtained by 
\begin{itemize}
\item eliminating pointer indirections ($*$ and $\rightarrow$) in $\alpha$ by
	their pointee variables in
	\aliasing, and 
\item replacing the variables by 
its type in case of type-based abstraction.
\end{itemize}

Let $x$ be a pointer to an object, $y$ be a pointer to a pointer to an object and $a$ be an object of type $\tau$.
Abstract names for different access expressions is shown in Table~\ref{table-dn}.
Entry in row $\dn (\oneae, \aliasing)_{\text{\em\sf asb}}$ depicts the abstract names computed using allocation-site-based abstraction and 
$\dn (\oneae, \aliasing)_{\text{\em\sf tba}}$ depicts the abstract names computed using type-based abstraction.
Abstract name for access expression $*x$ will never be encountered.
This is because program will model such a use by copy of an object which copies all the members of one object to respective members of another object.

\algnewcommand\algorithmicforeach{\textbf{for each}}
\algdef{S}[FOR]{ForEach}[1]{\algorithmicforeach\ #1\ \algorithmicdo}

\begin{algorithm}[t]
\begin{algorithmic}[1]
\Procedure{\id}{}
\ForEach {node $n$ in program \prog}
\State Set $\din_n, \dout_n, \ain_n, \aout_n$ as $\emptyset$
\EndFor
\State $\text{Dworklist} = \{n \mid n\in \origin\}$
\State $\text{Aworklist} = \{\text{startNode} \}$
\While {Dworklist is not empty or Aworklist is not empty}
\While {Dworklist is not empty}
\State Select a node $n$ from Dworklist
\State $\oldDin_n = \din_n$
\State Compute $\dout_n$, $\dgen_n$, $\dkill_n$ and $\din_n$ using Equations~\ref{eq:dout},~\ref{eq:dgen},~\ref{eq:dkill},~\ref{eq:din}.
\If {$\oldDin_n != \din_n$}
\State add predecessors of $n$ to Dworklist and Aworklist
\EndIf
\EndWhile
\While {Aworklist is not empty}
\State Select a node $n$ from Aworklist
\State $\oldAout_n = \aout_n$
\State Compute $\ain_n$, $\AGen_n$, $\AKill_n$ and $\aout_n$ using Equations~\ref{eq:ain},~\ref{eq:agen},~\ref{eq:akill},~\ref{eq:aout}.


\If {$\oldAout_n != \aout_n$}
\State add successors of $n$ to Aworklist and Dworklist
\EndIf
\EndWhile
\EndWhile
\EndProcedure
\end{algorithmic}
\caption{Worklist based demand-driven alias analysis algorithm with improved speculation (\id)}\label{algo:id}
\end{algorithm}

\subsection{Data Flow Equations for Intraprocedural Version}

Let the virtual call statements be recorded in the set \origin\footnote{
The nature of statements recorded in the set \origin is 
governed by the application. The set \origin can be suitably 
redefined for any other demand-driven application.
}.	
%
%
Our data flow equations compute
the following for each statement \text{$n\!: \lhs_n = \rhs_n$}:
\begin{inparaenum}[(i)]
\item the demands $(\din_n/\dout_n)$, and
\item the alias relationships $(\ain_n/\aout_n)$.
\end{inparaenum}
The equations have a bidirectional dependency because of the dependence of $\din_n/\dout_n$ (Equations~\ref{eq:din} and~\ref{eq:dout}) on 
$\ain_n/$ $\aout_n$  (Equations~\ref{eq:ain} and~\ref{eq:aout}).

\Start{p} and \End{p} denote the entry and exit nodes of procedure
$p$. At the intraprocedural level, the boundary information
\boundary associated with these nodes  is $\emptyset$. 
The changes for interprocedural propagation is described towards the end of this section.

The algorithm to perform our proposed demand-driven method
with improved speculation (henceforth denoted as \id)
is presented in Algorithm~\ref{algo:id}.
The algorithm terminates when simultaneous fixed point computation of both
the aliases and demands are reached.
The inner loops represent individual fixed point computations of demands and 
aliases for a round 
of mutual dependence between them.
The first inner \emph{while} loop represents fixed point computation
of demands using aliases from the previous round and 
the second inner \emph{while} loop represents fixed point computation 
of aliases using demands from the same round.
This is more formally represented and proved to be the \mfp solution
 in Lemma~\ref{lemma:mfp}
 in Section~\ref{sec:soundness-proof}.

\subsubsection*{Computing the demand.}
The equations for $\din_n/\dout_n$ are backward data flow equations which 
raises demands in the form of abstract name. 
We thus transform $\lhs$ and $\rhs$ of statement $n$ to its appropriate abstract name form as
 \begin{align*}
 \blhs_n & = \dn(\lhs_n, \ain_n)
 \\
 \brhs_n & = \dn(\rhs_n, \ain_n)
 \end{align*}
 
 Complete demand at a program point is computed with the help of the alias closure (Equation~\ref{eq:aliasclosure2})
 for the demand stored in $\din_n/\dout_n$.
 Consider $\alpha \in \dout_n$ and an alias $(\alpha, \beta) \in \aout_n$, then alias closure of the demand will 
 identify $\{\alpha, \beta\} \subseteq \dout'_n$.
 We use the below notation to denote the complete demand computed by taking alias closure at a program point.
 \begin{align*}
 \dout'_n & = \aout_n(\dout_n)
 \end{align*}
 
 Equation~\ref{eq:dgen} represents the demand generated for statement $n$.
 When $\blhs_n$ belongs to the demand raised at out of statement $n$, we raise demand for $\rhs_n$ and 
 when $\brhs_n$ belongs to the demand raised at out of statement $n$, we raise demand for $\lhs_n$.
 This gives rise to four cases while computing $\dgen_n$ which are,
 \begin{inparaenum}[\em(a)] 
 \item when both $\lhs$ and $\rhs$ belongs to the demand at out,
 \item when only $\lhs$ belongs to the demand at out,
 \item when only $\rhs$ belongs to the demand at out, or
 \item when neither $\lhs$ nor $\rhs$ belongs to the demand at out.                           
 \end{inparaenum}
Also, demand needs to be raised at the point of virtual function call recorded in the set \origin.
This condition is combined with case (b) of the $\dgen_n$ Equation~\ref{eq:dgen}. 
We perform weak update due to the use of an abstraction to model objects on heap.
We kill demand only when $\lhs_n$ is of the form $x$. 

We overload $\subseteq$ operator such that $\blhs_n \subseteq \dout'_n$ implies $\blhs_n \neq \emptyset \wedge \blhs_n \subseteq \dout'_n$.
 \begin{align}
\din_n & =  (\dout_n - \dkill_n) \cup \dgen_n
\label{eq:din}
	\\
\dout_n & =  
		\begin{cases}
		\boundary & n \text{ is } \End{p}
		\\
		\displaystyle\bigcup_{s \in \Succ(n)} \din_s & \text{otherwise}
		\end{cases}
	\label{eq:dout}
	\\
\dgen_n & = 
	\begin{cases}
	\ldgen(\rhs_n, \ain_n) \; \cup \; \rdgen(\lhs_n)
		& \blhs_n \subseteq \dout'_n \wedge
		\brhs_n \subseteq \dout'_n
	\\
	\ldgen(\rhs_n, \ain_n) 
		& \blhs_n \subseteq \dout'_n \vee n \in \origin
	\\
	\rdgen(\lhs_n)
		& \brhs_n \subseteq \dout'_n
	\\
	\emptyset
		& \text{otherwise}
	\end{cases}
	\label{eq:dgen}
	\\
\dkill_n & =  \{\lhs_n \mid \lhs_n \equiv x\}
\label{eq:dkill}
\end{align}

We compute $\dgen$ as a union of $\ldgen$ and $\rdgen$.
$\ldgen$ is computed when abstract name of \lhs belongs to $\dout_n$. In such a case demand for \rhs{} needs to be raised.
$\rdgen$ is computed when abstract name of \rhs{} belongs to $\dout_n$. In such a case demand for \lhs needs to be raised.
Further, the demand for $\rhs$ is generated incrementally depending upon the aliases of the base of $\rhs$.

First case for $\ldgen$ involve access expression whose $\base (\rhs) \neq \emptyset$ which represents access
expressions of the form $*x$ or $x\rightarrow f$.
It raises demand for 
the \vname involved in the access expression and also 
the abstract name of the entire access expression.
Second case refers to the access expression of the form $\&x$ and 
third case
considers access expression of the form $x$ and $x.f$. 
Demand for address-of a variable is raised using $\addrexpr$ as per our proposed speculation.
 \begin{align}
\ldgen(\rhs,\aliasing) & = 
\begin{cases}
	\vname(\rhs) \cup \addrexpr(\rhs) \cup \dn(\rhs, \aliasing) & \base(\rhs) \neq \emptyset
	\\
	\{\rhs\}	&	\addr(\rhs)
	\\
	\dn(\rhs, \aliasing) \cup \addrexpr(\rhs)	&	\vname(\rhs) \neq \emptyset
	\\
	\emptyset 	& 		\text{otherwise}
	\end{cases}
	\label{eq:ldgen}
	\\
\rdgen(\lhs) & = 
	\base(\lhs) \cup \addrexpr(\lhs)
	\label{eq:rdgen}
\intertext{where,}
\addrexpr(\oneae) & =
	\{ \&x \mid x \in \vname(\oneae) \wedge \addrTaken(\&x)\}
\end{align}

\subsubsection*{Computing aliases.}
Aliasing and typing relations are both computed by
the equations \text{$\ain_n/\aout_n$} which are forward data flow equations.
%
Symmetric and transitive closure of alias relationships is handled internally.
We do not clutter the equations to represent the same.

We generate alias relationship between $\lhs$ and $\rhs$ of a statement $n$ 
when demand for $\blhs_n$ has been raised or when alias relationship for $\brhs_n$ has been computed 
as shown in the Equation~\ref{eq:agen} for $\AGen_n$.
We perform weak update and kill alias relationship only when $\lhs_n$ is of the form $x$ as shown in Equation~\ref{eq:akill} for $\AKill_n$.
{
 \begin{align}
\ain_n & =  
		\begin{cases}
		\boundary & n \text{ is } \Start{p}
		\\
		\displaystyle\bigcup_{p \in \Pred(n)} \aout_p & \text{otherwise}
		\\
		\end{cases}
		\label{eq:ain}
	\\
\aout_n & =  
		(\ain_n - \AKill_n) \cup
		\AGen_n
		\label{eq:aout}
	\\
	\AGen_n \;\; & = 
	\{\blhs_n \times \brhs_n \mid \blhs_n \subseteq \dout_n \vee \ain_n (\brhs_n) \neq \{\brhs_n\}\}
	\label{eq:agen}
\\
	\AKill_n \;\; & = 
	  \{(\lhs_n, \oneae)\mid \lhs_n \equiv x \wedge \oneae \not \equiv \lhs_n\}
	  	\label{eq:akill}
\end{align}
}
\subsubsection*{Lifting the Analysis to Interprocedural Level.}
We perform context insensitive analysis where for indirect calls arising from function pointers,
we identify the potential callees using the points-to analysis performed by GCC; 
and for indirect calls involving virtual functions, we identify potential 
callees on-the-fly with the help of alias information computed by our analysis.

\subsection{Enhancing Precision and Efficiency Even Further}
Precision can be further improved in case of type-based abstraction by maintaining an object store, which records the objects
accessed by an analysis.
Efficiency can be improved by seeking for only used pointers.
We discuss these improvements regarding precision and efficiency in this section.

\subsubsection*{Computing the object store.}
The use of object store enhances the precision by ensuring that type-based
abstraction does not include irrelevant objects. It also enhances the efficiency 
of the analysis by ignoring the statements that do not involve relevant
objects. 
While propagating a demand, a statement is considered relevant when we compute $\dgen$ for it as depicted in Equation~\ref{eq:dgen}.
If the access expression is of the form $\&a$ or $a.f$, 
its abstract name does not depend on the alias relationship.
Instead, the abstract name is computed based on the types of the objects as shown in Table~\ref{table-dn}. 
We thus maintain a record of objects being accessed by a statement to avoid generating abstract names for objects other than the ones recorded 
in the object store.
Figure~\ref{m-eg-type-obj-store} depicts the precision gain achieved by recording objects from the relevant access expressions.

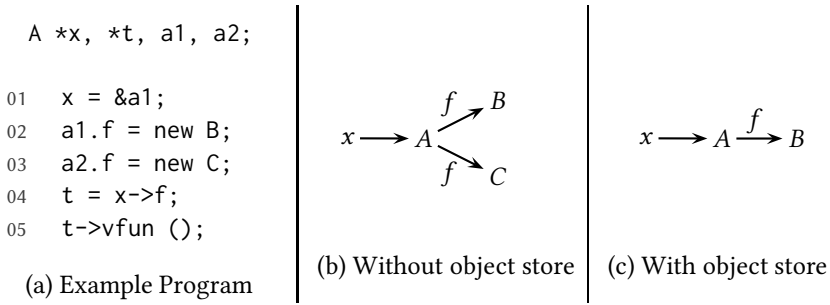
\begin{figure}[!t]
\centering
\begin{tabular}{@{}c|c|c@{}}
\begin{tabular}{@{}c@{}}

\setlength{\codeLineLength}{32mm}
	\begin{tabular}{@{}r@{}c}
	\codeLine{}{0}{A *x, *t, a1, a2; \rule{0em}{1.25em}}{white}
\\
	\codeLine{01}{1}{{x = \&a1;}}{white}
	\codeLine{02}{1}{{a1.f = new B;}}{white}	
	\codeLine{03}{1}{{a2.f = new C;}}{white}
	\codeLine{04}{1}{{t = x->f;}}{white}
	\codeLine{05}{1}{{t->vfun ();}}{white}
	
	\end{tabular}
	\\
	(a) Example Program
	\rule{0em}{1.75em}
	\end{tabular}
& 
\begin{tabular}{@{}l@{}}
 	\begin{pspicture}(0,15)(30,30)
	\psset{nodesep=-1.5}
	
	\putnode{x}{origin}{5}{15}{\pscirclebox[linestyle=none]{$x$}}
	\putnode{a1}{x}{10}{0}{\pscirclebox[linestyle=none]{$A$}}
	\putnode{B}{a1}{10}{5}{\pscirclebox[linestyle=none]{$B$}}
	\putnode{C}{a1}{10}{-5}{\pscirclebox[linestyle=none]{$C$}}
	
	\ncline[nodesep=-1]{->}{x}{a1}
	\ncline[nodesep=-1]{->}{a1}{B}
	\aput[1pt](.4){$f$}
	\ncline[nodesep=-1]{->}{a1}{C}
	\bput[1pt](.4){$f$}
	
	\end{pspicture} \\\\
	\rule[2.5em]{0em}{1em}
	(b) Without object store
	\end{tabular}
	&
\begin{tabular}{@{}l@{}}
 	\begin{pspicture}(0,15)(30,30)
	\psset{nodesep=-1.5}
	
	\putnode{x}{origin}{6}{15}{\pscirclebox[linestyle=none]{$x$}}
	\putnode{a1}{x}{10}{0}{\pscirclebox[linestyle=none]{$A$}}
	\putnode{B}{a1}{10}{0}{\pscirclebox[linestyle=none]{$B$}}
	
	\ncline[nodesep=-1]{->}{x}{a1}
	\ncline[nodesep=-1]{->}{a1}{B}
	\aput[1pt](.4){$f$}
	\end{pspicture} \\\\
	\rule[2.5em]{0em}{1em}
	(c) With object store
\end{tabular}

\end{tabular}
\caption{Illustrating precision gain with the use of an object store.
Abstract names for both the access expressions \text{\sf a1.f} and 
\text{\sf a2.f} 
is \text{\sf A.f}.
Without an object store, 
both the access expressions will be considered,
resulting in the points-to graph 
as shown in (b).
If we record the objects accessed with respect to demand \text{\sf x}, 
only \text{\sf a1} will belong to the object store.
Thus, statement 3 can be considered irrelevant;
resulting in a precise points-to graph as shown in (c).
}
\label{m-eg-type-obj-store}
\end{figure}

\subsubsection*{Seek for only used pointers}\label{sec:seek-pointers}
In our motivating example shown in Figure~\ref{m-eg-type} (a), we seek for pointers for $\&z$ using our proposed speculation.
We end up identifying that both $p$ and $q$ points-to $z$. 
As $q$ is not used in the program, it will not influence the typing relation \typing.
It is a redundant aliasing relation computed.
We can eliminate computation of such redundant aliasing relations by maintaining a store of the pointers involved in indirect assignments.
As $p$ is involved in indirect assignment on line 15, it will be recorded in the store.
When we seek to find pointers to $\&z$, we can now consider the statement in line 04 involving $p$
and eliminate considering the statement in line 03 involving $q$.
This will not have any effect on the precision but will make the analysis efficient by eliminating unnecessary information from the points-to graph.

\section{Formalizing the Solutions of Bidirectional Analyses}
\label{sec:formalizing-bidirectional}

Unlike a flow-insensitive analysis, a flow-sensitive analysis relies on the direction of control flow and 
can be either \emph{unidirectional} or \emph{bidirectional}. A unidirectional analysis may be
\begin{itemize}
\item a \emph{forward} analysis in which the information flows along the control flow i.e., the data flow values of a node
       are influenced by its ancestors (e.g. when computing alias information), or 
\item a \emph{backward} analysis in which the information flows against the control flow i.e., the data flow values of a node
       are influenced by its descendants (e.g. when computing demands).\footnote{%
	Ancestors and descendants represent the transitive closures of \Pred and
      \Succ relations, respectively.}
\end{itemize}
In a bidirectional flow analysis, the data flow values at a node are influenced by both ancestors as well as descendants.
Some examples of analyses that are recognized as bidirectional analyses are: partial redundancy 
elimination~\cite{Morel:1979:GOS:359060.359069} and its many variants,
type inferencing~\cite{Frade:2009:BDA:1480945.1480965,Khedker:2003:BDF:2295357.2295383}, and
liveness based points-to analysis~\cite{lfcpa}.

Interestingly, almost all demand-driven methods that compute flow-sensitive information, are
inherently bidirectional because demands are raised and propagated against control flow whereas the requisite information 
that satisfies the demands is discovered along the control flow. 
Yet, barring the liveness-based points-to analysis~\cite{lfcpa}, most demand-driven methods
have not been looked at as bidirectional analyses 
in the sense that the bidirectional dependency has been realised through an algorithm rather than
a declarative specification of the dependencies in terms of data flow equations.
Thus, in over two decades of their existence, they 
have been viewed as efficient versions of algorithms devised to compute 
only the information required to meet a given set of demands.
We lift these bidirectional dependencies from within
demand-driven algorithms to the specifications of analyses that the algorithms compute. 

\newcommand{\mylooparrow}{%
\psset{unit=1.75mm}
\begin{pspicture}(0,-.17)(2,1)
\pscurve[showpoints=false,linewidth=.115](0,0)(.85,.085)(1.2,.35)(1.3,.6)(1,.8)(.7,.6)(.8,.3)(1.15,.085)(2,0)
\psline[showpoints=false,linewidth=.115](.4,.4)(0,0)(.4,-.4)
\psline[showpoints=false,linewidth=.115](1.7,.3)(2,0)(1.7,-.3)
\end{pspicture}
}

\newcommand{\dfv}{\text{$\text{\sf\em X}$}\xspace}
\newcommand{\fdfv}{\text{$\overrightarrow{\dfv}$}\xspace}
\newcommand{\bdfv}{\text{$\overleftarrow{\dfv}$}\xspace}
\newcommand{\fbdfv}{\text{$\overleftrightarrow{\dfv}$}\xspace}
\newcommand{\odfv}{\text{$\overset{\mylooparrow}{\dfv}$}\xspace}

We formalize bidirectional dependencies by 
\begin{itemize}
\item defining a data flow value \dfv to contain multiple components distinguishing between the data flow
      information reaching a node from ancestors, descendants, or other nodes.

\item extending the concept of a \emph{meet over paths} solution (\mop) to the desired result of a bidirectional analysis, and
\item showing that the \emph{maximum fixed point} solution (\mfp) computed using the data flow equations defining
      a bidirectional analysis, is a sound approximation of the corresponding \mop.
\end{itemize}

Section~\ref{sec:bidirectional.big.picture} provides a big picture view of bidirectional dependencies. 
Section~\ref{sec:cfp.bidirectional} formalizes the control flow paths used for defining solutions of
a bidirectional analysis whereas Section~\ref{sec:mfp.mop.bidirectional} defines the \mop and \mfp solutions 
in terms of generic flow functions and shows the relationship between them.

\newcommand{\fe}{\text{$\overrightarrow{e}$}\xspace}
\newcommand{\fee}[1]{\text{$\overrightarrow{e_{#1}}$}\xspace}
\newcommand{\be}{\text{$\overleftarrow{e}$}\xspace}
\newcommand{\bee}[1]{\text{$\overleftarrow{e_{#1}}$}\xspace}
\newcommand{\In}[1]{\mbox{\em In$_{#1}$}}
\newcommand{\Out}[1]{\mbox{\em Out$_{#1}$}}
\newcommand{\ff}[1]{\mbox{$\overrightarrow{\!\!f}\!\!_{\!#1}$}}
\newcommand{\fb}[1]{\mbox{$\overleftarrow{f}\!\!_{#1}$}}

\subsection{The Big Picture View of Bidirectional Dependencies}
\label{sec:bidirectional.big.picture}
In this section we present the big picture view of bidirectional dependencies.

\subsubsection{Modelling Bidirectional Flows}

\begin{figure}[t]
\centering
\begin{tabular}{@{}c|c|c|c|c@{}}
\psset{unit=1.3cm}
\begin{pspicture}(2,-.4)(3.9,2.8)
\rput(3,2.7){Flow Functions}
\cnodeput{0}(3,2){n1}{\rule{0cm}{.6cm}}
\rput(3,2){$n$}
\rput(3,0){$m$}
\cnodeput{0}(3,0){n2}{\rule{0cm}{.6cm}}
\ncline{->}{n1}{n2}
\rput(2.4,2.4){\rnode{InN}{}}
\rput(2.4,1.5){\rnode{OutN}{}}
\rput(2.4,0.5){\rnode{InS}{}}
\rput(2.4,-.4){\rnode{OutS}{}}
\ncline[linewidth=.25mm,nodesepA=0mm,linestyle=dashed,dash=.6mm .5mm]{|->}{InN}{OutN}
\ncline[linewidth=.25mm,nodesepA=0mm,linestyle=dashed,dash=.6mm .5mm]{|->}{OutN}{InS}
\ncline[linewidth=.25mm,nodesepA=0mm,linestyle=dashed,dash=.6mm .5mm]{|->}{InS}{OutS}
\rput(2.45,2.05){\rnode{fn1}{\psframebox[linestyle=none,fillstyle=solid,framesep=0]{\ff{n}}}}
\rput(2.4,1){\rnode{fe}{\psframebox[linestyle=none,fillstyle=solid,framesep=0]{\rule{1mm}{0mm}\rule{0mm}{3mm}}}}
\rput(2.45,1.05){\rnode{fe}{\psframebox[linestyle=none,framesep=0]{\ff{n\rightarrow m}}}}
\rput(2.45,.10){\rnode{fn2}{\psframebox[linestyle=none,fillstyle=solid,framesep=0]{\ff{m}}}}
\rput(3.6,2.4){\rnode{InN}{}}
\rput(3.6,1.5){\rnode{OutN}{}}
\rput(3.6,0.5){\rnode{InS}{}}
\rput(3.6,-.4){\rnode{OutS}{}}
\ncline[linewidth=.25mm,nodesepB=0mm,linestyle=dashed,dash=.6mm .5mm]{<-|}{InN}{OutN}
\ncline[linewidth=.25mm,nodesepB=0mm,linestyle=dashed,dash=.6mm .5mm]{<-|}{OutN}{InS}
\ncline[linewidth=.25mm,nodesepB=0mm,linestyle=dashed,dash=.6mm .5mm]{<-|}{InS}{OutS}
\rput(3.6,1.95){\rnode{bn1}{\psframebox[linestyle=none,fillstyle=solid,framesep=0]{\fb{n}}}}
\rput(3.6,.96){\rnode{fe}{\psframebox[linestyle=none,fillstyle=solid,framesep=0]{\rule{1mm}{0mm}\rule{0mm}{3mm}}}}
\rput(3.6,1.0){\rnode{be}{\psframebox[linestyle=none,framesep=0]{\fb{n\rightarrow m}}}}
\rput(3.6,0){\rnode{bn2}{\psframebox[linestyle=none,fillstyle=solid,framesep=0]{\fb{m}}}}
\end{pspicture}
&
\begin{pspicture}(0,0)(24,41.5)
\putnode{w}{origin}{12}{40}{Forward Flow}
\putnode{w}{origin}{12}{34}{$\ff{k\rightarrow l}\!\circ\!\ff{k}\!\circ\!\ff{i\rightarrow k}$}
\putnode{i}{origin}{4}{25}{\pscirclebox{\white$j$}}
\putnode{j}{i}{16}{0}{\pscirclebox{\white$j$}}
\putnode{k}{i}{8}{-10}{\pscirclebox{\white$j$}}
\putnode{l}{k}{-8}{-10}{\pscirclebox{\white$j$}}
\putnode{m}{k}{8}{-10}{\pscirclebox{\white$j$}}
\putnode{x}{i}{0}{-5}{$i$}
\putnode{x}{j}{0}{-5}{$j$}
\putnode{x}{k}{-4}{0}{$k$}
\putnode{x}{l}{0}{5}{$l$}
\putnode{x}{m}{0}{5}{$m$}
\ncline[nodesep=-.75]{->}{i}{k}
\ncline[nodesep=-.75]{->}{j}{k}
\ncline[nodesep=-.75]{->}{k}{l}
\ncline[nodesep=-.75]{->}{k}{m}
\nccurve[angleA=-10,angleB=90,linewidth=.3mm,offsetB=-.75,offsetA=-1,
	linestyle=dashed,dash=.5mm .5mm]{->}{i}{k}
\nccurve[angleA=270,angleB=10,linewidth=.3mm,offsetA=-.75,offsetB=-1,
	linestyle=dashed,dash=.5mm .5mm]{->}{k}{l}
\nccurve[angleA=270,angleB=90,linewidth=.3mm,offset=-.75,
	nodesepA=-2,nodesepB=-1,linestyle=dashed,dash=.5mm .5mm]{>-}{k}{k}
\end{pspicture}
&
\begin{pspicture}(0,0)(24,41.5)
\putnode{w}{origin}{12}{40}{Backward Flow}
\putnode{w}{origin}{12}{34}{$\fb{i\rightarrow k}\!\circ\!\fb{k}\!\circ\!\fb{k\rightarrow l}$}
\putnode{i}{origin}{4}{25}{\pscirclebox{\white$j$}}
\putnode{j}{i}{16}{0}{\pscirclebox{\white$j$}}
\putnode{k}{i}{8}{-10}{\pscirclebox{\white$j$}}
\putnode{l}{k}{-8}{-10}{\pscirclebox{\white$j$}}
\putnode{m}{k}{8}{-10}{\pscirclebox{\white$j$}}
\putnode{x}{i}{0}{-5}{$i$}
\putnode{x}{j}{0}{-5}{$j$}
\putnode{x}{k}{-4}{0}{$k$}
\putnode{x}{l}{0}{5}{$l$}
\putnode{x}{m}{0}{5}{$m$}
\ncline[nodesep=-.75]{->}{i}{k}
\ncline[nodesep=-.75]{->}{j}{k}
\ncline[nodesep=-.75]{->}{k}{l}
\ncline[nodesep=-.75]{->}{k}{m}
\nccurve[angleA=-10,angleB=90,linewidth=.3mm,offsetB=-.75,offsetA=-1,
	linestyle=dashed,dash=.5mm .5mm]{<-}{i}{k}
\nccurve[angleA=270,angleB=10,linewidth=.3mm,offsetA=-.75,offsetB=-1,
	linestyle=dashed,dash=.5mm .5mm]{<-}{k}{l}
\nccurve[angleA=270,angleB=90,linewidth=.3mm,offset=-.75,
	nodesepA=-1,nodesepB=-2,linestyle=dashed,dash=.5mm .5mm]{-<}{k}{k}
\end{pspicture}
&
\begin{pspicture}(0,0)(24,41.5)
\putnode{w}{origin}{12}{40}{Bidirectional Flow}
\putnode{w}{origin}{12}{34}{$\fb{j\rightarrow k}\!\circ\!\ff{i\rightarrow k}$}
\putnode{i}{origin}{4}{25}{\pscirclebox{\white$j$}}
\putnode{j}{i}{16}{0}{\pscirclebox{\white$j$}}
\putnode{k}{i}{8}{-10}{\pscirclebox{\white$j$}}
\putnode{l}{k}{-8}{-10}{\pscirclebox{\white$j$}}
\putnode{m}{k}{8}{-10}{\pscirclebox{\white$j$}}
\putnode{x}{i}{0}{-5}{$i$}
\putnode{x}{j}{0}{-5}{$j$}
\putnode{x}{k}{-4}{0}{$k$}
\putnode{x}{l}{0}{5}{$l$}
\putnode{x}{m}{0}{5}{$m$}
\ncline[nodesep=-.75]{->}{i}{k}
\ncline[nodesep=-.75]{->}{j}{k}
\ncline[nodesep=-.75]{->}{k}{l}
\ncline[nodesep=-.75]{->}{k}{m}
%
\nccurve[angleA=-10,angleB=90,linewidth=.3mm,offsetB=-.75,offsetA=-1,
	linestyle=dashed,dash=.5mm .5mm]{->}{i}{k}
\nccurve[angleA=90,angleB=190,linewidth=.3mm,offsetA=-.75,offsetB=-1,
	linestyle=dashed,dash=.5mm .5mm]{->}{k}{j}
\end{pspicture}
&
\begin{pspicture}(0,0)(24,41.5)
\putnode{w}{origin}{12}{40}{Bidirectional Flow}
\putnode{w}{origin}{12}{34}{$\ff{k\rightarrow l}\!\circ\!\fb{k\rightarrow m}$}
\putnode{i}{origin}{4}{25}{\pscirclebox{\white$j$}}
\putnode{j}{i}{16}{0}{\pscirclebox{\white$j$}}
\putnode{k}{i}{8}{-10}{\pscirclebox{\white$j$}}
\putnode{l}{k}{-8}{-10}{\pscirclebox{\white$j$}}
\putnode{m}{k}{8}{-10}{\pscirclebox{\white$j$}}
\putnode{x}{i}{0}{-5}{$i$}
\putnode{x}{j}{0}{-5}{$j$}
\putnode{x}{k}{-4}{0}{$k$}
\putnode{x}{l}{0}{5}{$l$}
\putnode{x}{m}{0}{5}{$m$}
\ncline[nodesep=-.75]{->}{i}{k}
\ncline[nodesep=-.75]{->}{j}{k}
\ncline[nodesep=-.75]{->}{k}{l}
\ncline[nodesep=-.75]{->}{k}{m}
%
\nccurve[angleA=270,angleB=10,linewidth=.3mm,offsetA=-.75,offsetB=-1,
	linestyle=dashed,dash=.5mm .5mm]{->}{k}{l}
\nccurve[angleA=270,angleB=170,linewidth=.3mm,offsetA=.75,offsetB=1,
	linestyle=dashed,dash=.5mm .5mm]{<-}{k}{m}
\end{pspicture}
\end{tabular}
\caption{Modelling General Flow Functions and General Flows Using Information Flow
		Paths~\protect\cite{Khedker:1994:GTB:186025.186043}}
\label{fig:genth.modelling}
\end{figure}
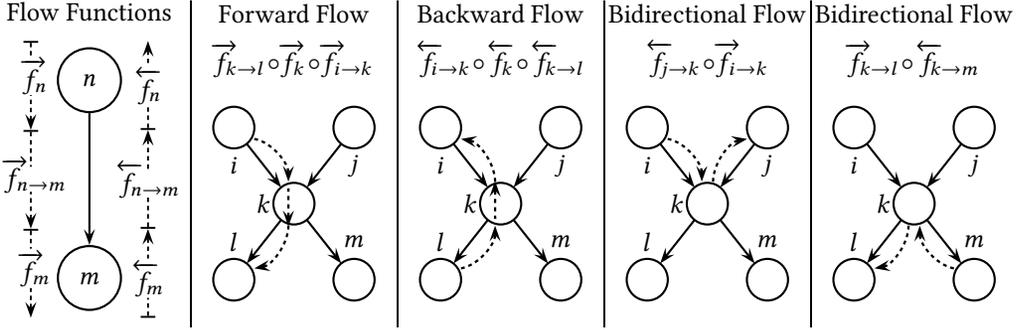

Bidirectional dependencies can be explained by making a distinction between a
\emph{node flow function} and an \emph{edge flow function}~\cite{Khedker:1994:GTB:186025.186043, 
		Khedker:1999:BDF:606666.606676, Khedker:2009:DFA:1592955}. Besides,
a distinction is also made between a \emph{forward flow function} and a \emph{backward flow function} as shown 
in Figure~\ref{fig:genth.modelling}. This leads to the following basic flow functions:
\begin{itemize}
\item a forward flow function $\ff{n}$ for node $n$, is used to compute \Out{n} from \In{n},
\item a backward flow function $\fb{n}$ for node $n$, is used to compute \In{n} from \Out{n},
\item a forward flow function $\ff{n\rightarrow m}$ for edge \text{$n\rightarrow m$} is used to compute \In{m} from \Out{n}, and 
\item a backward flow function $\fb{n\rightarrow m}$ for edge \text{$n\rightarrow m$} is used to compute \Out{n} from
\In{m}.
\end{itemize}
If a particular dependency does not exist in an analysis, the corresponding flow function is defined to compute the
$\top$ value of the lattice. Since \text{$\dfv \sqcap \top = \dfv$} for every \dfv, this correctly models the absence of
the dependency.
As illustrated in Figure~\ref{fig:genth.modelling}, this facilitates modelling all kinds of flows.

\begin{figure}[t]
\begin{pspicture}(0,0)(138,138)
\putnode{p1}{origin}{90}{135}{\begin{tabular}{@{}c@{}}
			All flows 
			\end{tabular}}
\putnode{p2}{p1}{-27}{-13}{\begin{tabular}{@{}c@{}}
			Bidirectional Flows
			\end{tabular}}
\putnode{p3}{p1}{27}{-13}{\begin{tabular}{@{}c@{}}
			Unidirectional Flows
			\end{tabular}}

\ncline{->}{p1}{p2}
\ncline{->}{p1}{p3}
			
\putnode{q1}{p2}{-41}{-18}{\begin{tabular}{@{}l@{}}
				Only 	
				unidirectional 
				\\
				interacting 
				components
			\end{tabular}}

\putnode{q2}{p2}{0}{-20}{\begin{tabular}{@{}c@{}}
				Only \\
				bidirectional \\ components 
			\end{tabular}}

\putnode{q3}{p2}{41}{-18}{\begin{tabular}{@{}c@{}}
				Both unidirectional and \\ bidirectional components 
			\end{tabular}}

\ncline{->}{p2}{q1}
\ncline{->}{p2}{q2}
\ncline{->}{p2}{q3}
\putnode{p4}{q1}{0}{-17}{\begin{tabular}{@{}l@{}}
				$\dfv = \langle \fdfv, \bdfv \rangle$
			\end{tabular}}
\putnode{p5}{q3}{0}{-17}{\begin{tabular}{@{}l@{}}
				$\dfv = \langle \fdfv, \bdfv, \odfv \rangle$
			\end{tabular}}
\putnode{p6}{q2}{0}{-17}{\begin{tabular}{@{}c@{}}
				$\dfv = \fbdfv$
			\end{tabular}}
\ncline{->}{q1}{p4}
\ncline{->}{q3}{p5}
\ncline{->}{q2}{p6}

\putnode{p8}{p6}{22}{-17}{%
		\begin{tabular}{@{}c@{}}
		Bounded number 
		\\ 
		of segments 
			\\
		\end{tabular}}
\putnode{p9}{p6}{-22}{-17}{\begin{tabular}{@{}c@{}}
		Unbounded number \\
		of segments 
			\\
		\end{tabular}}
\ncline[nodesepA=1.25,nodesepB=1]{->}{p6}{p8}
\ncline[nodesepA=1.25,nodesepB=1]{->}{p6}{p9}

\putnode{w}{p5}{0}{-6}{\psframebox[linestyle=none, fillstyle=solid, fillcolor=pink]{\tia}}
\putnode{w}{p4}{0}{-6}{\psframebox[linestyle=none, fillstyle=solid, fillcolor=pink]{\lb, \id, \ddnew}}
\putnode{w}{p9}{0}{-8}{\psframebox[linestyle=none, fillstyle=solid, fillcolor=pink]{\tib, \prea, \co, \su}}
\putnode{w}{p8}{0}{-8}{\psframebox[linestyle=none, fillstyle=solid, fillcolor=pink]{\preb}}

\putnode{l1}{p4}{50}{-60}{\begin{tabular}{@{}rl@{}}
\multicolumn{2}{l}{\em List of bidirectional analyses}\\
\lb & Liveness based points-to analysis~\cite{lfcpa}.\\
\id & Our proposed demand-driven method. \\
\ddnew & Other demand-driven points-to/alias analyses~\cite{Sridharan:2006:RCP:1133981.1134027,%
Guyer:2003:CPA:1760267.1760284,%
Sridharan:2005:DPA:1094811.1094817,%
Yan:2011:DCA:2001420.2001440,%
spth_et_al:LIPIcs:2016:6116,%
Saha:2005:IDP:1069774.1069785}. There are many \\ & other demand-driven analysis that we do not cite here.\\
\tia & Type inferencing with multiple components~\cite{Khedker:2003:BDF:2295357.2295383, Singer04sparsebidirectional}. \\
\tib & Type inferencing with a single component~\cite{Frade:2009:BDA:1480945.1480965}.\\
\prea & Variants of partial redundancy elimination without edge splitting~\cite{Morel:1979:GOS:359060.359069,
Dhamdhere:1993:CBD:158511.158696,Khedker:1994:GTB:186025.186043,Khedker:2009:DFA:1592955}.  
\\
\preb & Variants of partial redundancy elimination with edge splitting or edge placement
\\
& ~\cite{Khedker:1999:BDF:606666.606676,Khedker:2009:DFA:1592955,
Briggs:1994:EPR:178243.178257,%
			Dhamdhere:1992:ALP:143095.143135,%
			Drechsler:1988:SPM:48022.214509,%
			Drechsler:1993:VKR:152819.152823,%
			Hailperin:1998:CCM:295656.295664,%
			vonHanxleden:1994:GBC:773473.178253,%
			Knoop:1992:LCM:143095.143136,%
			Knoop:1994:OCM:183432.183443,%
			Knoop:1994:PDC:178243.178256,%
			Knoop:1995:PAM:223428.207150,%
			krs.strength_red,%
			Dhamdhere:1988:FAC:51607.51621,%
			Dhamdhere:1993:EAB:169701.169684,%
			DBLP:journals/jpl/DhaneshwarD95,%
			Paleri:1998:SAP:307824.307851}.
\\
\co & Coalescing used in bottom-up points-to analysis~\cite{GPG.pta}. \\
\su & Stack usage optimization~\cite{SAABAS2007103,Frade:2009:BDA:1480945.1480965}.
		\end{tabular}}

\end{pspicture}

\caption{The big picture view of data flow analysis. 
The corresponding analyses are listed in the pink boxes.
}
\label{fig:bd.big.picture}
\end{figure}

We use this model for characterizing flows in terms of edges
traversed in the control flow graph by explicating the edge flow functions but
leaving the node flow functions implicit; given two adjacent edges, the common
node between them is known and the associated node flow function can be inferred as
is evident from Figure~\ref{fig:genth.modelling}. 
Hence we use the following notation:

\begin{center}
\begin{tabular}{|rl|}
\hline
\fe & Traversal of edge $e$ along the control flow for applying the edge flow function \ff{e}.
	\rule{0em}{1.3em}%
		\\
\be & Traversal of edge $e$ against the control flow for applying the edge flow function \fb{e}.
	\\ \hline
\end{tabular}
\end{center}

A contiguous sequence of 
flows in the same direction is called a \emph{segment}~\cite{Dhamdhere:1993:CBD:158511.158696,Khedker:1994:GTB:186025.186043}. Thus,
a sequence of \fe is a forward segment whereas 
a sequence of \be is a backward segment. By definition, a segment does not allow mixing \fe and \be.
The flows in unidirectional analyses consist of a single unbounded segment
described by the regular expressions \text{$\left(\fe\right)^+$} or \text{$\left(\be\right)^+$}. 
The flows in bidirectional analysis consist of an unbounded number of unbounded segments
described by the regular expression \text{$\left(\fe \middle\vert \be\right)^+$}.

In order to explain all kinds of flows precisely we define 
a data flow value of node $n$, $\dfv_n$, to possibly contain multiple components using the following notation:

\begin{center}
\setlength{\tabcolsep}{4pt}
\begin{tabular}{|rl|}
\hline
$\fdfv_n$ & Forward component of $\dfv_n$ representing the information generated in ancestors of $n$.
	\rule{0em}{1.3em}%
		\\
$\bdfv_n$ & Backward component of $\dfv_n$ representing the information generated in descendants of $n$.
		\\
$\odfv_n$ & Cross-over component of $\dfv_n$ representing the information generated in other nodes.
		\\
$\fbdfv_n$&  Universal component of $\dfv_n$ representing the information generated in 
		ancestors, \\ & descendants, or other nodes.
	\\ \hline
\end{tabular}
\end{center}

The values of the components \fdfv, \bdfv, and \odfv are interdependent. For example, demands (the value of \bdfv component)
govern the aliases (the value of \fdfv component) which in turn decide what further demands should be raised. 
The cross-over component \odfv represents a flow from one segment into another.
The universal component \fbdfv models the analyses that make no
distinction between the forward, backward, or cross-over components, and the same data flow value is propagated in
all directions. This component subsumes the other three components and hence appears exclusively rather than in combination
with other components.  
These components cover all data flows reported in the literature
as illustrated in Figure~\ref{fig:bd.big.picture}. 
The flow of these components is illustrated in Figure~\ref{fig:genth.modelling}.
For convenience, we name edges 
\text{$i \rightarrow k\equiv e_{ik}$},
\text{$j \rightarrow k\equiv e_{jk}$},
\text{$k \rightarrow l\equiv e_{kl}$}, and 
\text{$k \rightarrow m\equiv e_{km}$}.
\begin{itemize}
\item The flow of unidirectional components is restricted to the designated directions along the control flow paths.

In the forward component, the information flows from node $i$ to $k$ and then to $l$. This 
is described by the sequence \text{$\fee{ik}, \fee{kl}$} which forms a subpath of a control flow path.
Similarly in the backward component, the information flows from 
$l$ to $k$ and then to $i$. It is described by the sequence
\text{$\bee{kl}, \bee{ik}$} which also coincides with a control flow path (although against the control flow).

\item The flow of cross-over and universal components is \emph{not} restricted to the control flow paths.

The cross-over component is highlighted under bidirectional flow where the information flows from 
$i$ to $k$ and then to $j$ or from $m$ to $k$ and then to $l$.
In the first case, information flow is described by the sequence
\text{$\fee{ik}, \bee{jk}$} which does not coincide with any control flow path. 
In the second case, information flow is described by the sequence
\text{$\bee{km}, \fee{kl}$} which also does not coincide with any control flow path. 

\end{itemize}


The $k$-tuple framework~\cite{Masticola:1995:LFM:213978.213989}
classifies information reaching along different sources. This generalization allows modelling 
information flows along edges other than control flow edges (such as dependence edges or synchronization edges for concurrency).
Although this framework can also be used for bidirectional flows, their notion of separate components in the tuples
is not necessarily related to the direction of flows. As a consequence, their formulation for PRE is ad hoc (explained later in
Section~\ref{sec:bd.description}).
Further, there is no definition of \mop. We distinguish between
forward and backward components explicitly and generalize the notion of control flow paths to define \mop for
bidirectional flows.

\subsubsection{Illustrating Bidirectional Dependencies}

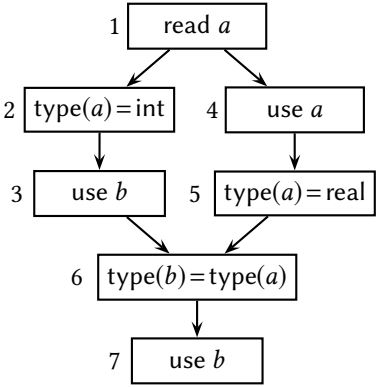
\begin{figure}[t]
\begin{tabular}{c|c}
\begin{tabular}{c}
\begin{pspicture}(0,0)(40,54)
\putnode{n1}{origin}{18}{50}{\psframebox{\white$a*b$}}
\putnode{n2}{n1}{-7}{-12}{\psframebox{$a*b$}}
\putnode{n3}{n2}{0}{-13}{\psframebox[framesep=1.1]{$b\!=\!5$}}
\putnode{m1}{n1}{7}{-9}{\psframebox{\white$a*b$}}
\putnode{n4}{m1}{0}{-10}{\psframebox{$a*b$}}
\putnode{n5}{n4}{0}{-18}{\psframebox{\white$a*b$}}
\putnode{n6}{n5}{-7}{-9}{\psframebox{$a*b$}}
\putnode{n7}{n4}{10}{-9}{\psframebox{$a\!=\!8$}}
\ncline{->}{n1}{n2}
\ncline{->}{n1}{m1}
\ncline{->}{m1}{n4}
\ncline{->}{n2}{n3}
\ncline{->}{n4}{n5}
\ncline{->}{n5}{n6}
\ncline{->}{n3}{n6}
\ncline{->}{n4}{n7}
\ncline{->}{n7}{n5}	
\ncloop[angleA=270,angleB=90,offset=2,loopsize=-14,armA=2.5,armB=3,linearc=.5]{->}{n5}{n4}
\ncloop[angleA=270,angleB=90,loopsize=9,armA=2.5,armB=3,linearc=.5]{->}{n3}{n2}
\putnode{w}{n1}{-6}{0}{1}
\putnode{w}{n2}{-6}{0}{2}
\putnode{w}{m1}{-6}{0}{4}
\putnode{w}{n3}{-6}{0}{3}
\putnode{w}{n4}{-6}{0}{5}
\putnode{w}{n7}{-6}{0}{6}
\putnode{w}{n5}{-6}{0}{7}
\putnode{w}{n6}{-6}{0}{8}
\end{pspicture}
\end{tabular}
&
\begin{tabular}{c}
\begin{minipage}{87mm}
No occurrence of $a*b$ can be hoisted. 
The occurrences of $a*b$ in nodes 2 and 8 cannot be hoisted into node 3. Hence the occurrence of node 8 cannot
be hoisted into node 7. This suppresses the hoisting of $a*b$ from node 5.

 The effect of
$\In{2}=0$ reaches \In{5} along the path 
\text{$I_2, O_3
	  , I_8
	  , O_7
	  , I_5
	$}
which can be described by the sequence \be, \fe, \be, \fe. It does not correspond to any control flow path.
\end{minipage}
\end{tabular}
\end{tabular}
\caption{Information flows in PRE. We use the data flow equations with only Pavin as the CONST
	term~\protect\cite{Dhamdhere:1993:CBD:158511.158696,Khedker:1994:GTB:186025.186043}.
}
\label{fig:bd.pre.exmp}
\end{figure}

\begin{figure}[t]
\newcommand{\Type}{\text{\sf type}\xspace}
\newcommand{\Read}{\text{\sf read}\xspace}
\newcommand{\Use}{\text{\sf use}\xspace}
\newcommand{\Real}{\text{\sf real}\xspace}
\newcommand{\Int}{\text{\sf int}\xspace}
\newcommand{\Tr}{\text{r}\xspace}
\newcommand{\Ti}{\text{i}\xspace}
\newcommand{\Ts}{\text{s}\xspace}
\begin{tabular}{c|c}
\begin{tabular}{c}
\begin{pspicture}(0,0)(49,51)
\putnode{n1}{origin}{24}{48}{\psframebox{\makebox[16mm]{$\Read \; a$\rule[-.25em]{0em}{1em}}}}
\putnode{n2}{n1}{-13}{-11}{\psframebox{{$\Type(a) \!=\! \Int$}}}
\putnode{n3}{n2}{0}{-11}{\psframebox{\makebox[15mm]{$\Use \; b$\rule[-.25em]{0em}{1em}}}}
\putnode{n4}{n1}{13}{-11}{\psframebox{\makebox[16mm]{$\Use \; a$\rule[-.25em]{0em}{1em}}}}
\putnode{n5}{n4}{0}{-11}{\psframebox{$\Type(a) \!=\! \Real$}}
\putnode{n6}{n5}{-13}{-11}{\psframebox{{$\Type(b) \!=\! \Type(a)$\rule[-.25em]{0em}{1em}}}}
\putnode{n7}{n6}{0}{-11}{\psframebox{\makebox[15mm]{$\Use \; b$\rule[-.25em]{0em}{1em}}}}
\ncline{->}{n1}{n2}
\ncline{->}{n1}{n4}
\ncline{->}{n2}{n3}
\ncline{->}{n4}{n5}
\ncline{->}{n5}{n6}
\ncline{->}{n3}{n6}
\ncline{->}{n6}{n7}
\putnode{w}{n1}{-11}{0}{1}
\putnode{w}{n2}{-12}{0}{2}
\putnode{w}{n3}{-11}{0}{3}
\putnode{w}{n4}{-11}{0}{4}
\putnode{w}{n5}{-13}{0}{5\,}
\putnode{w}{n6}{-16}{0}{6}
\putnode{w}{n7}{-11}{0}{7}
\end{pspicture}
\end{tabular}
&
\begin{tabular}{c}
\begin{minipage}{78mm}
The statements in nodes 2, 5, and 6 model type constraints. Other statements do not influence types.
We show the types for \Use and \Read statements in the form \text{$\langle \fdfv,\bdfv,\odfv\rangle$} by abbreviating \Int by \Ti and
\Real by \Tr.

\bigskip

\begin{center}
\begin{tabular}{|c|l|l|}
\hline
Node 	& Types of variable $a$ 
	& Types of variable $b$ 
\\ \hline\hline
1 
	&  $\langle  \emptyset, \left\{\Ti,\Tr\right\}, \left\{\Ti,\Tr\right\} \rangle$
	&  $\langle  \emptyset, \left\{\Ti,\Tr\right\}, \left\{\Ti,\Tr\right\} \rangle$
\\ \hline
3 
	&  $\langle  \left\{\Ti\right\}, \emptyset, \left\{\Tr\right\} \rangle$
	&  $\langle  \emptyset, \left\{\Ti,\Tr\right\}, \left\{\Ti,\Tr\right\} \rangle$
\\ \hline
4 
	&  $\langle  \emptyset, \left\{\Tr\right\} ,\left\{\Ti\right\} \rangle$
	&  $\langle  \emptyset, \left\{\Ti,\Tr\right\}, \left\{\Ti,\Tr\right\} \rangle$
\\ \hline
7 
	&  $\langle  \left\{\Ti,\Tr\right\},\emptyset, \emptyset \rangle$
	&  $\langle  \left\{\Ti,\Tr\right\},\emptyset, \emptyset \rangle$
\\ \hline
\end{tabular}
\end{center}
\bigskip
For variable $a$, type \Tr reaches node 3 as part of cross-over component along the path 
\text{$O_5, I_6, O_3$}.
Similarly, type \Ti reaches node 4 along the path \text{$I_2, O_1, I_4$}.
\end{minipage}
\end{tabular}
\end{tabular}
\caption{An example of flow-sensitive type inferencing. We present a slightly spruced up modelling of the 
analysis~\protect\cite{Khedker:2003:BDF:2295357.2295383}.
At the end of the analysis, an empty set indicates no information and can be interpreted suitably 
depending upon the application of the analysis.
}
\label{fig:bd.type.inferencing}
\end{figure}
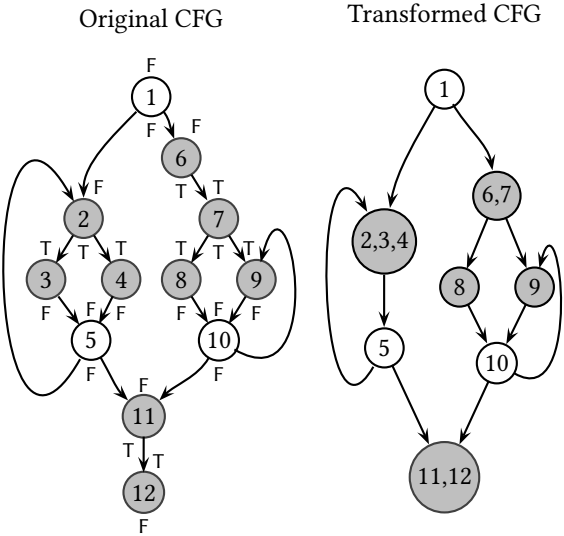

\begin{figure}[t]
\newcommand{\cc}{\text{\sf canCoalesce}\xspace}
\newcommand{\bT}{\text{\footnotesize\sf \raisebox{1.25mm}{T}}\xspace}
\newcommand{\aT}{\text{\footnotesize\sf \raisebox{-3.5mm}{T}}\xspace}
\newcommand{\bF}{\text{\footnotesize\sf \raisebox{1.25mm}{F}}\xspace}
\newcommand{\aF}{\text{\footnotesize\sf \raisebox{-3.5mm}{F}}\xspace}

\begin{tabular}{@{}ccc@{}}
\begin{tabular}{@{}c}
\begin{pspicture}(6,-3)(44,70)
\putnode{n1}{origin}{23}{56}{\pscirclebox{1}}
\putnode{w}{n1}{0}{10}{Original CFG}
\putnode{n6}{n1}{4}{-8}{\pscirclebox[fillstyle=solid,fillcolor=lightgray,linecolor=darkgray]{6}}
\putnode{n7}{n6}{5}{-8}{\pscirclebox[fillstyle=solid,fillcolor=lightgray,linecolor=darkgray]{7}}
\putnode{n8}{n7}{-5}{-8}{\pscirclebox[fillstyle=solid,fillcolor=lightgray,linecolor=darkgray]{8}}
\putnode{n9}{n7}{5}{-8}{\pscirclebox[fillstyle=solid,fillcolor=lightgray,linecolor=darkgray]{9}}
\putnode{n10}{n8}{5}{-8}{\pscirclebox[framesep=.55]{10}}
\putnode{n11}{n10}{-10}{-10}{\pscirclebox[framesep=.7,fillstyle=solid,fillcolor=lightgray,linecolor=darkgray]{11}}
\putnode{n12}{n11}{0}{-10}{\pscirclebox[fillstyle=solid,fillcolor=lightgray,linecolor=darkgray,framesep=.6]{12}}
\putnode{n2}{n1}{-9}{-16}{\pscirclebox[fillstyle=solid,fillcolor=lightgray,linecolor=darkgray]{2}}
\putnode{n3}{n2}{-5}{-8}{\pscirclebox[fillstyle=solid,fillcolor=lightgray,linecolor=darkgray]{3}}
\putnode{n4}{n2}{5}{-8}{\pscirclebox[fillstyle=solid,fillcolor=lightgray,linecolor=darkgray]{4}}
\putnode{n5}{n3}{6}{-8}{\pscirclebox{5}}
\nccurve[angleA=330,angleB=75,ncurv=2.5,nodesepB=-.5,nodesepA=-.75]{->}{n10}{n9}
\nccurve[angleA=310,angleB=110,nodesepA=-.75,nodesepB=-.25]{->}{n1}{n6}
\nccurve[angleA=225,angleB=90,nodesepA=-.75,nodesepB=0]{->}{n1}{n2}
\nccurve[angleA=240,angleB=120,ncurv=2.5,nodesepB=-.5,nodesepA=-.25]{->}{n5}{n2}
\nccurve[angleA=300,angleB=135,nodesepA=-.5,nodesepB=-.9]{->}{n5}{n11}
\nccurve[angleA=240,angleB=45,nodesepA=-.5,nodesepB=-.9]{->}{n10}{n11}
\ncline[nodesepA=-.5,nodesepB=-.5]{->}{n7}{n9}
\ncline[nodesepA=-.5,nodesepB=-.6]{->}{n7}{n8}
\ncline[nodesepA=-.5,nodesepB=-.6]{->}{n2}{n3}
\ncline[nodesepA=-.5,nodesepB=-.6]{->}{n2}{n4}
\ncline[nodesepA=-.5,nodesepB=-.6]{->}{n4}{n5}
\ncline[nodesepA=-.5,nodesepB=-.7]{->}{n3}{n5}
\ncline[nodesepA=-.5,nodesepB=-.6]{->}{n8}{n10}
\ncline[nodesepA=-.5,nodesepB=-.7]{->}{n9}{n10}
\ncline[nodesep=-.3]{->}{n11}{n12}
\ncline[nodesepA=-.5,nodesepB=-.35]{->}{n6}{n7}
\putnode{w}{n1}{0}{5}{\aF}
\putnode{w}{n1}{0}{-5}{\bF}
\putnode{w}{n2}{2}{5}{\aF}
\putnode{w}{n2}{0}{-5}{\bT}
\putnode{w}{n3}{0}{5}{\aT}
\putnode{w}{n3}{0}{-5}{\bF}
\putnode{w}{n4}{0}{5}{\aT}
\putnode{w}{n4}{0}{-5}{\bF}
\putnode{w}{n5}{0}{5}{\aF}
\putnode{w}{n5}{0}{-5}{\bF}
\putnode{w}{n6}{2}{5}{\aF}
\putnode{w}{n6}{0}{-5}{\bT}
\putnode{w}{n7}{0}{5}{\aT}
\putnode{w}{n7}{0}{-5}{\bT}
\putnode{w}{n8}{0}{5}{\aT}
\putnode{w}{n8}{0}{-5}{\bF}
\putnode{w}{n9}{-1}{5}{\aT}
\putnode{w}{n9}{0}{-5}{\bF}
\putnode{w}{n10}{0}{5}{\aF}
\putnode{w}{n10}{0}{-5}{\bF}
\putnode{w}{n11}{0}{5}{\aF}
\putnode{w}{n11}{-2}{-5}{\bT}
\putnode{w}{n12}{2}{5}{\aT}
\putnode{w}{n12}{0}{-5}{\bF}
\end{pspicture}
\end{tabular}
&
\begin{tabular}{c@{}}
\begin{pspicture}(9,-3)(33,70)
\putnode{n1}{origin}{20}{57}{\pscirclebox{1}}
\putnode{w}{n1}{0}{10}{Transformed CFG}
\putnode{n6}{n1}{7}{-14}{\pscirclebox[fillstyle=solid,fillcolor=lightgray,framesep=.6]{6,7}}
\putnode{n8}{n6}{-5}{-12}{\pscirclebox[fillstyle=solid,fillcolor=lightgray]{8}}
\putnode{n9}{n6}{5}{-12}{\pscirclebox[fillstyle=solid,fillcolor=lightgray]{9}}
\putnode{n10}{n8}{5}{-10}{\pscirclebox[framesep=.55]{10}}
\putnode{n12}{n10}{-7}{-15}{\pscirclebox[fillstyle=solid,fillcolor=lightgray,linecolor=darkgray,framesep=.6]{11,12}}
\putnode{n2}{n1}{-8}{-20}{\pscirclebox[fillstyle=solid,fillcolor=lightgray,framesep=.6]{2,3,4}}
\putnode{n5}{n2}{0}{-14}{\pscirclebox{5}}
\nccurve[angleA=330,angleB=75,ncurv=1.5,nodesepB=-.25,nodesepA=-.5]{->}{n10}{n9}
\nccurve[angleA=300,angleB=100,nodesepA=-.75]{->}{n1}{n6}
\ncline[nodesepA=-.5,nodesepB=-.25]{->}{n6}{n9}
\ncline[nodesepA=-.5,nodesepB=-.35]{->}{n6}{n8}
\nccurve[angleA=240,angleB=80,nodesepA=-.5,nodesepB=0]{->}{n1}{n2}
\ncline[nodesepA=-.5,nodesepB=-.5]{->}{n10}{n12}
\nccurve[angleA=270,angleB=90]{->}{n2}{n5}
\ncline[nodesepA=-.5,nodesepB=-.5]{->}{n5}{n12}
\nccurve[angleA=240,angleB=120,ncurv=1.5,nodesepB=-.5,nodesepA=-.25]{->}{n5}{n2}
\ncline[nodesepA=-.5,nodesepB=-.25]{->}{n8}{n10}
\ncline[nodesepA=-.5,nodesepB=-.35]{->}{n9}{n10}
\end{pspicture}
\end{tabular}
&
\begin{tabular}{@{}c}
\begin{minipage}{66mm}
\begin{itemize}
\item The shaded nodes in the original graph are candidate nodes for coalescing with adjacent nodes.
\item A node can be coalesced with a predecessor (successor) \emph{only} if it can be coalesced with \emph{all} predecessors
(successors).
The result of the analysis is marked as true or false for the entry and exit of each node.
\item The shaded nodes in the transformed graph are the coalesced nodes. 

      Node 9 cannot be coalesced with 10, so it cannot coalesced with the other predecessor 7 either. Hence,
      no successor of 7 can be coalesced with it, leaving 8 out of coalescing.
  
      This effect of 10 reaches node 8 along the path $O_{10}, I_9, O_7, I_8$.
\end{itemize}
\end{minipage}
\end{tabular}
\end{tabular}

\caption{An example of coalescing analysis~\cite{GPG.pta}. We present a simplified version by abstracting out the exact need of coalescing and
the features that allow a node to be considered for coalescing.
}
\label{fig:bd.coalsecing}
\end{figure}

\begin{figure}[t]
\newcommand{\load}{\text{\sf load}\xspace}
\newcommand{\store}{\text{\sf store}\xspace}
\newcommand{\pop}{\text{\sf pop}\xspace}
\newcommand{\ls}{\text{\sf ls}\xspace}

\begin{tabular}{c|c}
\begin{tabular}{c}
\begin{pspicture}(0,0)(40,54)
\putnode{n1}{origin}{18}{50}{\psframebox{\makebox[9mm]{\white$a*b$\rule[-.25em]{0em}{.9em}}}}
\putnode{n2}{n1}{-10}{-12}{\psframebox{\makebox[9mm]{\load x\rule[-.25em]{0em}{.9em}}}}
\putnode{n3}{n2}{0}{-12}{\psframebox{\makebox[9mm]{\ls\rule[-.25em]{0em}{.9em}}}}
\putnode{m1}{n1}{10}{-10}{\psframebox{\makebox[9mm]{\load y\rule[-.25em]{0em}{.9em}}}}
\putnode{n4}{m1}{0}{-11}{\psframebox{\makebox[9mm]{\white$a*b$\rule[-.25em]{0em}{.9em}}}}
\putnode{n5}{n4}{-8}{-11}{\psframebox{\makebox[9mm]{\ls\rule[-.25em]{0em}{.9em}}}}
\putnode{n6}{n3}{0}{-19}{\psframebox{\makebox[9mm]{\pop\rule[-.25em]{0em}{.9em}}}}
\putnode{n7}{n4}{8}{-11}{\psframebox{\makebox[9mm]{\store z\rule[-.25em]{0em}{.9em}}}}
\ncline{->}{n1}{n2}
\ncline{->}{n1}{m1}
\ncline{->}{m1}{n4}
\ncline{->}{n2}{n3}
\ncline{->}{n4}{n5}
\ncline{->}{n5}{n6}
\ncline{->}{n3}{n6}
\ncline{->}{n4}{n7}
\putnode{w}{n1}{-7}{0}{1}
\putnode{w}{n2}{-7}{0}{2}
\putnode{w}{m1}{-7}{0}{4}
\putnode{w}{n3}{-7}{0}{3}
\putnode{w}{n4}{-7}{0}{5}
\putnode{w}{n7}{-7}{0}{7}
\putnode{w}{n5}{-7}{0}{6}
\putnode{w}{n6}{-7}{0}{8}
\end{pspicture}
\end{tabular}
&
\begin{tabular}{c}
\begin{minipage}{87mm}
The \ls (level sequence) nodes represent a sequence of instructions 
such that these instructions do not affect the stack height and do not consume the values already 
present in the stack.

The \load-\pop pair along the control flow path $(1,2,3,8)$ and
$(1,4,5,6,8)$ can be considered for elimination.
However, the \store in node 7 requires a \load in node 4, which in turn makes it mandatory to retain
the \pop in node 8. Thus, the \load in node 2 cannot be eliminated.

The effect that the \load in node 2 is mandatory because of 
node 7, reaches node 2 along the path
\text{$I_7, O_5
	  , I_6
	  , O_6
	  , I_8
	  , O_3
	  , I_3
	  , O_2
	$}
which can be described by the sequence \be, \fe, \fe, \fe, \be, \be, \be. It does not correspond to any control flow path.

\end{minipage}
\end{tabular}
\end{tabular}
\caption{An example of \load-\pop stack optimization taken from~\cite{SAABAS2007103,Frade:2009:BDA:1480945.1480965}.
}
\label{fig:bd.su.load.pop}
\end{figure}

We illustrate bidirectional dependencies for 
partial redundancy elimination (aka PRE)~\cite{Dhamdhere:1993:CBD:158511.158696,Khedker:1994:GTB:186025.186043}
in Figure~\ref{fig:bd.pre.exmp},
for type inferencing of flow-sensitive types in Figure~\ref{fig:bd.type.inferencing},
for coalescing analysis in Figure~\ref{fig:bd.coalsecing}, and for
stack based code optimization for load-pop pairs in Figure~\ref{fig:bd.su.load.pop}.
Let $I_i$ and $O_i$ denote the entry and exit points of node $i$.
We refer to the examples of bidirectional analyses presented in Figure~\ref{fig:bd.big.picture}.
PRE (\prea), coalescing analysis, stack based 
(load-pop pair or store-load pair) 
optimization
 and type inferencing (\tib) have only the universal component whereas type inferencing (\tia) has forward, backward, and 
the cross-over components. 
It is easy to see that the flows of universal components and cross-over components
do not coincide with the directed paths in a control flow graph. They have been modelled using 
the concept of \emph{information flow path} which was introduced for 
PRE~\cite{Dhamdhere:1993:CBD:158511.158696,Khedker:1994:GTB:186025.186043}.
An information flow path is essentially a path
in the underlying undirected graph of a control flow graph. It represents a path along which data flow information could flow
depending on the nature of dependencies in an analysis. 
The flow function $f_\rho$ for an information
flow path $\rho$ is defined by composing appropriate flow functions for the nodes appearing in $\rho$ as illustrated in 
Figure~\ref{fig:genth.modelling}.

Let \paths{n} denote the set of control flow paths passing through node $n$. Then the paths in \paths{n} are sufficient to 
compute the values of $\fdfv_n$ and $\bdfv_n$. However, they are not sufficient to compute the 
values of $\odfv_n$ and $\fbdfv_n$. 
This is somewhat incompatible with the usual method of
establishing soundness of a static analysis which is based on showing that the information that the analysis computes for a node
is a conservative approximation of all possible execution paths involving the node; thus the paths in \paths{n} are
sufficient for reasoning about the analyses that have only
$\fdfv_n$ and $\bdfv_n$. 
More specifically,
we cannot reason about PRE, type inferencing of flow-sensitive types,
coalescing analysis or stack based (load-pop pair or store-load pair) optimization
 by comparing the analysis results
with the concrete values seen in the execution traces passing through a given node.



\subsubsection{Bidirectional Analyses Reported in the Literature}
\label{sec:bd.description}

We classify bidirectional flows
into the following categories based on whether they have the cross-over components or not or
whether they combine all components into the lone universal component.
\begin{enumerate}
\item Bidirectional analyses with forward and backward components.

        Most demand-driven methods (such as ours) and liveness based flow- and context- sensitive points-to
        analysis~\cite{lfcpa} have distinct forward and backward components but no cross-over component: 
	a demand or liveness information flows backwards while 
        alias information or points-to information flows forwards. 
        The general bidirectional flows are described by the regular expression \text{$(\fe \vert \be)^+$};
	the demand-driven analyses that we know, start with a backward flow and have a 
        restricted regular expression \text{$\be (\fe \vert \be)^+$}.

        Although the regular expression \text{$(\fe \vert \be)^+$} also admits cross-over components where
        one segment ends and the next segment begins, the reversal in the direction in these analyses is for a different component.
        Information in a particular component flows only in a given direction thereby eliminating the need of
        a cross-over component.

\item Bidirectional analyses with forward, backward, and cross-over components.

      Type inferencing of flow sensitive types~\cite{Khedker:2003:BDF:2295357.2295383, Singer04sparsebidirectional} distinguishes between the
      type reaching a node from its ancestors, descendants, or other nodes. 
      These flows are described by the regular expression \text{$(\fe \vert \be)^+$}.
      This analysis is illustrated in Figure~\ref{fig:bd.type.inferencing}.

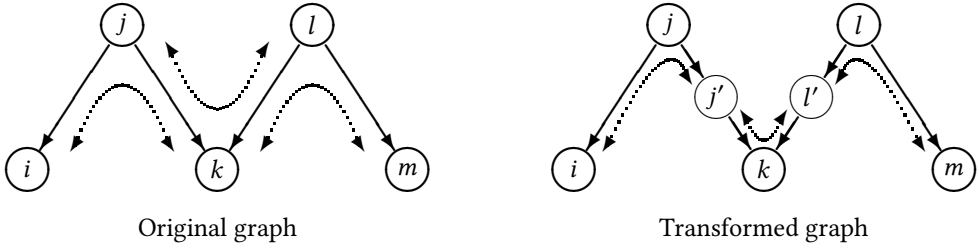
\begin{figure}[t]
\begin{pspicture}(-6,-9)(130,23)
\thicklines
\setlength{\unitlength}{.9pt}
\put(0,0){\circle{20}}
\put(0,0){\makebox(0,0){$i$}}
\put(40,60){\circle{20}}
\put(40,60){\makebox(0,0){$j$}}
\put(80,0){\circle{20}}
\put(80,0){\makebox(0,0){$k$}}
\put(120,60){\circle{20}}
\put(120,60){\makebox(0,0){$l$}}
\put(160,0){\circle{20}}
\put(160,0){\makebox(0,0){$m$}}
\put(80,-25){\makebox(0,0){Original graph}}

\put(34.45,51.68){\vector(-2,-3){29.00}}
\put(45.55,51.68){\vector(2,-3){29.00}}
\put(114.45,51.68){\vector(-2,-3){29.00}}
\put(125.55,51.68){\vector(2,-3){29.00}}
\qbezier[30](20,10)(40,60)(60,10)
\put(20,10){\vector(-2,-3){2}}
\put(60,10){\vector(2,-3){2}}
\qbezier[30](60,50)(80,0)(100,50)
\put(60,50){\vector(-2,3){2}}
\put(100,50){\vector(2,3){2}}
\qbezier[30](100,10)(120,60)(140,10)
\put(100,10){\vector(-2,-3){2}}
\put(140,10){\vector(2,-3){2}}
\put(230,0){\circle{20}}
\put(230,0){\makebox(0,0){$i$}}
\put(270,60){\circle{20}}
\put(270,60){\makebox(0,0){$j$}}
{\thinlines
\put(290,30){\circle{20}}
\put(290,30){\makebox(0,0){$j'$}}
}
\put(310,0){\circle{20}}
\put(310,0){\makebox(0,0){$k$}}
{\thinlines
\put(330,30){\circle{20}}
\put(330,30){\makebox(0,0){$l'$}}
}
\put(350,60){\circle{20}}
\put(350,60){\makebox(0,0){$l$}}
\put(390,0){\circle{20}}
\put(390,0){\makebox(0,0){$m$}}
\put(310,-25){\makebox(0,0){Transformed graph}}
\qbezier[30](244,10)(270,62)(278,38)
\put(244,10){\vector(-2,-3){2}}
\put(278,38){\vector(2,-3){2}}
\qbezier[30](342,38)(350,62)(376,10)
\put(342,38){\vector(-2,-3){2}}
\put(376,10){\vector(2,-3){2}}
\qbezier[15](302,22)(310,2)(318,22)
\put(302,22){\vector(-2,3){2}}
\put(318,22){\vector(2,3){2}}

\put(264.45,51.68){\vector(-2,-3){29.00}}
\put(275.55,51.68){\vector(2,-3){8.90}}
\put(295.55,21.68){\vector(2,-3){8.90}}
\put(324.45,21.68){\vector(-2,-3){8.90}}
\put(344.45,51.68){\vector(-2,-3){8.90}}
\put(355.55,51.68){\vector(2,-3){29.00}}
\end{pspicture}
\caption{Edge-splitting}
\label{fig:split.exmp}
\end{figure}

\item Bidirectional analyses with the lone universal component.

        All variants of PRE have
        both forward and backward flows but the same information flows in both directions. 
	Further, most variants of PRE and the analyses inspired by
              PRE, have some form of restriction as explained below (a more detailed discussion of
              this class of analyses can be found in~\cite{Khedker:1999:BDF:606666.606676}).\footnote{This is admittedly an
              old reference but there has not been much work in this direction of late.}
	     \begin{enumerate}
	     \item The original PRE~\cite{Morel:1979:GOS:359060.359069}
		   has unbounded backward segments but the forward segments contain at most one edge because
                   the forward node flow function \ff{n} is defined to compute $\top$ in PRE (the value on \Out{n}
		   does not depend on that of \In{n}). Thus, the
		   flows in PRE can be described by the regular expression
		   \text{$\left((\be)^+ \!\cdot(\fe)^k
                   \right)^+$} where $k$ is 0 or 1~\cite{Dhamdhere:1993:CBD:158511.158696,Khedker:1994:GTB:186025.186043,Khedker:2009:DFA:1592955}.  
		    Since the length of forward segments is
                   bounded, this analysis has also been called a \emph{mostly backwards}
                   analysis~\cite{Dhamdhere:1992:ALP:143095.143135}. In spite of a bounded forward segment,
                   the number of segments in this analysis remains unbounded. 

		The $k$-tuple framework~\cite{Masticola:1995:LFM:213978.213989}
		distinguishes between the information reaching along forward and backward edges for PRE.
		This classification is ad hoc in the case of PRE
as the information at the entry of a node not only depends on the information reaching along forward edges
but it also depends on the information at the exit of the node. Since the information 
at the entry depends on the information flowing along both, forward and backward edges, distinction between
the information coming along different sources is redundant.

             \item Coalescing analysis~\cite{GPG.pta} tries to find out the nodes in a control flow graph that can be merged
                   together to eliminate redundant control flow for summarizing a procedure for a bottom-up flow- and
                   context-sensitive points-to analysis. It contains only edge flow functions but no node flow functions.
                   Hence each segment has exactly one edge and the flows can be described by the regular expression
		   \text{$\left(\fe \!\cdot \be\right)^+ \mid \left(\be \!\cdot \fe\right)^+$}.
  	     \item 
		   Stack based IR (such as byte code) has instructions such as \text{\sf load $x$} (push the value of $x$ into the
                   top of the stack), \text{\sf store $x$} (pop the top of the stack and store the value into $x$), and
                   \text{\sf pop} (discard the top of the stack). This creates the possibility of optimizing the
                   code to eliminate redundant \text{\sf load-pop} or \text{\sf store-load} 
                   pairs~\cite{SAABAS2007103, Frade:2009:BDA:1480945.1480965}
This decision to eliminate instructions requires the information to flow to all program nodes
and its flows can be described by the regular expression
\text{$(\fe \vert \be)^+$}. 

Similarly, type inferencing~\cite{Frade:2009:BDA:1480945.1480965} has universal component and the
type information flows to all program nodes.
This flow is also
described by the regular expression
\text{$(\fe \vert \be)^+$}. 

	     \item There is a plethora of investigations inspired by PRE that have tried to bound the number of
                   segments in an attempt to perform the analysis using a fixed number of cascaded forward and backward analyses.
                   These investigations were motivated by the belief that most bidirectional analyses can be
                   solved using a bounded sequence of unidirectional flows~\cite{Muchnick:1998:ACD:286076}. These investigations used the
                   following two approaches that are different but achieve the same effect (see~\cite{Khedker:1999:BDF:606666.606676} for more details): 
		   \begin{itemize}
		   \item Modifying the control flow graph using 
                         \emph{edge-spitting}~\cite{Khedker:1999:BDF:606666.606676,Khedker:2009:DFA:1592955}. 
                        Let an edge from a fork node (i.e. a node with multiple successors) to a join node (i.e. a node
			with multiple predecessors) be called a \emph{critical edge}. A large body of work is based on modifying the
                        control flow graph by splitting critical edges by inserting a node thereby replacing the edge by a
	                sequence of two edges as shown in Figure~\ref{fig:split.exmp}.
                        Since the forward segments are restricted to a single edge, this has the effect of prohibiting
                        a backward segment after a forward segment thereby bounding the number of segments.
                        Here we merely cite the papers that use this 
			idea~\cite{Briggs:1994:EPR:178243.178257,%
			Dhamdhere:1992:ALP:143095.143135,%
			Drechsler:1988:SPM:48022.214509,%
			Drechsler:1993:VKR:152819.152823,%
			Hailperin:1998:CCM:295656.295664,%
			vonHanxleden:1994:GBC:773473.178253,%
			Knoop:1992:LCM:143095.143136,%
			Knoop:1994:OCM:183432.183443,%
			Knoop:1994:PDC:178243.178256,%
			Knoop:1995:PAM:223428.207150,%
			krs.strength_red}. 
		\item Modifying the dependencies using \emph{edge-placement}. The idea is to place some computations along edges
                      instead of inserting them in the existing nodes. Effectively, this amounts to splitting the chosen edges. The main
                      difference is that in this approach, the edges are not split a-priori before the analysis but are split on demand after the analysis.
			Again, we merely cite the papers here~\cite{
			Dhamdhere:1988:FAC:51607.51621,%
			Dhamdhere:1993:EAB:169701.169684,%
			DBLP:journals/jpl/DhaneshwarD95,%
			Paleri:1998:SAP:307824.307851} that use edge-placement.
                   \end{itemize}
		Observe that edge-splitting (or edge-placement) reduces PRE to a bounded sequence of cascaded unidirectional analyses 
		because it does not have all the four flow functions. Hence the number of segments in the flows is bounded.
                It is easy to see that this transformation will not reduce other reported bidirectional flows to a bounded sequence of
                unidirectional analyses unless some segment is bounded.
	     \end{enumerate}
	\end{enumerate}


\subsection{Defining Control Flow Paths for Bidirectional Data Flow Analysis}
\label{sec:cfp.bidirectional}

The classical definition of the meet-over-paths solution (\mop) is as 
	follows~\cite{Kildall:1973:UAG:512927.512945, Khedker:2009:DFA:1592955, Nielson:1999:PPA:555142}:
\begin{align*}
\forall n \in \nodes: \;\;\;\;\; \mop(n) \;\;\; = \displaystyle\bigsqcap\limits_{\scriptsize\rho \,\in\, \mathit{Paths}(n)} 
	f_\rho (\boundary)
\end{align*}
where 
\nodes represent the set of all nodes in program \prog,
\boundary{} (short for Boundary Information) denotes the external data flow information reaching the procedure from its calls,
\text{$\mathit{Paths}(n)$} denotes the set of control flow paths reaching $n$, and 
$f_\rho$ represents the composition of the flow functions of the nodes appearing on
path $\rho$.

Intuitively, $\mop(n)$ represents the greatest lower bound of the data flow information that can reach node $n$
along all paths that involve $n$.  We use the following notation for denoting sets of paths involving node $n$
\begin{center}
\begin{tabular}{|rl|}
\hline
\fpaths{n} & The set of control flow paths from \Start{} to the entry of node $n$.
	\\
\bpaths{n} & The set of control flow paths from the exit of node $n$ to \End{}. 
\\ \hline
\end{tabular}
\end{center}
In the presence of loops, nodes could appear multiple times in a control 
flow paths. Hence, node $n$ could occur in paths in \fpaths{n} and \bpaths{n}.
Then, $\mop(n)$ for forward flows is defined using the paths in \fpaths{n} whereas $\mop(n)$ for
backward flows is defined using the paths in \bpaths{n}.
For bidirectional flows, $\mop(n)$ could be defined in a similar manner provided 
\text{$\mathit{Paths}(n)$} can be defined unambiguously because
bidirectional flows require both kinds of paths (i.e. paths in \fpaths{n} and \bpaths{n}).

Since most bidirectional methods (and in particular, demand-driven methods) do not have a cross-over component, we 
do not use information flow paths which are defined in terms of the paths in the underlying undirected graph of the control flow graph.
Instead, we model the computation of forward and backward components of a data flow value by
defining the concept of a \emph{qualified control flow path} involving node $n$ in terms of paths from  
\Start{} to the entry of $n$ and from the exit of $n$ to \End{}. 
Formally,
\begin{align}
\paths{n} & = \left\{ \left(\overrightarrow{\rho},n, \overleftarrow{\rho}\right) \; \middle\vert \; 
		\overrightarrow{\rho} \in \fpaths{n}, \; \overleftarrow{\rho} \in \bpaths{n} 
		\right\}
\\
\Pred\left(\left(
		\overrightarrow{\rho},
		n,
		\overleftarrow{\rho},
		\right), n \right) 
		& = \text{ last node in } \overrightarrow{\rho}
	\nonumber
	\\
\Succ\left(\left(
		\overrightarrow{\rho},
		n,
		\overleftarrow{\rho},
		\right), n \right) 
		& = \text{ first node in } \overleftarrow{\rho}
	\nonumber
\end{align}
With this view, a qualified path \text{$\rho \in \paths{n}$} corresponds to a particular occurrence of node 
in the underlying unqualified control flow path. Consider two qualified control flow paths
$\rho_1$ and $\rho_2$ such that
\text{$\rho_1 = (\overrightarrow{1,2,3},5,\overleftarrow{2,4,5,6})$} 
and \text{$\rho_2 = (\overrightarrow{1,2,3,5,2,4},5,\overleftarrow{6})$}. Both $\rho_1$ and $\rho_2$
correspond to the same underlying control flow path
\text{$(1,2,3,5,2,4,5,6)$}. However, 
$\rho_1$ is used for computing the data flow information reaching the
first occurrence of node 5 in the underlying control flow path whereas 
$\rho_2$ is used for computing the data flow information reaching the
second occurrence of node 5. Also,
\text{$\Pred(\rho_1,5) = 3$},
\text{$\Succ(\rho_1,5) = 2$},
\text{$\Pred(\rho_2,5) = 4$}, and
\text{$\Succ(\rho_2,5) = 6$}.
Note that we overload the \Pred and \Succ functions such that 
$\Pred(n)$ returns a set of predecessors for node $n$ in a program \prog whereas,
$\Pred(\rho, n)$ returns a predecessor of node $n$ along a path $\rho \in \paths{n}$.

If we model type inferencing without the cross-over component as
$\dfv = \langle \fdfv, \bdfv \rangle$,
we could use our definition of qualified control flow paths instead of information flow paths.
However, in case of PRE and coalescing we need to use information flow paths.

We uniformly use the qualified control flow paths with the notation \text{$\rho \in \paths{n}$} leaving it implicit that
the forward flows would use the forward component $\overrightarrow{\rho}$ of $\rho$ whereas the
backward flows would use the backward component $\overleftarrow{\rho}$. 

For our analysis, we assume that only \End{} statement contains a virtual call statement and other statements in $\rho$ are
only pointer assignment statements. The statements other than pointer assignments can be modelled by a \emph{skip} statement for alias analysis and hence
can be effectively removed from $\rho$.

\subsection{Defining \protect\mop and \protect\mfp for Bidirectional Data Flow Analysis}
\label{sec:mfp.mop.bidirectional}

This section defines the \mop for bidirectional flows.  This definition is a generalization of the original definition
for unidirectional flows~\cite{Kildall:1973:UAG:512927.512945,Khedker:2009:DFA:1592955} but is 
different from the definition based on information flow paths~\cite{Khedker:1994:GTB:186025.186043}
in that it excludes the cross-over components.

Let the initialization for a program \prog be described by
vector $\val{} = (\boundary_{\Start{}}, \top, \top, \ldots \boundary_{\End{}})$. The initialization for node $n$ is denoted by
\initval{n}\footnote{
Note that we represent vectors in boldface (e.g. \val{}, \gfunname) and 
an element in the vector using sans-serif and italicized font (e.g. \initval{n}, \gfunnode{n}).
}. 
The vector of data flow values is denoted by \gfunname and the data flow value for node $n$ is denoted by
\gfunnode{n}. 
\gfunnode{n} can contain multiple components depending on the analysis it is defined for.
In general, it represents the values at the entry and exit of a node,
modelled as $\gfunnode{n} = (\gfunnodein{n}{}, \gfunnodeout{n}{})$.
For a bidirectional analysis such as demand-driven alias analysis, we define
$\gfunnodein{n}{} = (\ain_n, \din_n)$ and 
$\gfunnodeout{n}{} = (\aout_n, \dout_n)$. 
We use such a modelling in the further sections.

Let $\gfun{n}{\prog}$ represent the final
value at program point $n$ across all paths involving $n$ in program \prog.
It is the \mfp solution for program \prog.  
Let $\gfun{n}{\rho}$ represent the final value at program point $n$ along
the qualified control flow path \text{$\rho \in \paths{n}$}; it is used for computing \mop solution by merging
values along all paths in \paths{n}. 
The corresponding intermediate values computed in step 
$i>0$ are denoted by $\gfun{n}{\prog,\,i}$ and $\gfun{n}{\rho,\,i}$.
Step $i$ computes values for all the nodes $n\in \nodes$ i.e. 
the entire vector
$\gfunname^{\prog, i}$ and $\gfunname^{\rho, i}$
using the values of 
$\gfunname^{\prog, i\!-\!1}$ and $\gfunname^{\rho, i\!-\!1}$, $i>0$. 

$\gfun{n}{\prog,\,i}$ and $\gfun{n}{\rho,\,i}$
are computed by applying flow function $f_n$ as follows:
\begin{align*}
\gfun{n}{\prog,\,0} 	&= \initval{n}
	\\
\forall i > 0, \gfun{n}{\prog,\,i}  &=
 \displaystyle\bigsqcap_{\renewcommand{\arraystretch}{.8}%
			\scriptsize%
			\begin{array}{@{}c@{}}
			\scriptsize p \in \Pred(n) 
			\\
			\scriptsize s \in \Succ(n)
			\end{array}} 
	f_n \left(
	\gfun{p}{\prog,\,i-1} ,
	\gfun{s}{\prog,\,i-1} 
		\right)
\end{align*}
Similarly, 
\begin{align*}
\gfun{n}{\rho,\,0} 	&= \initval{n}
		\\
\forall i > 0, \gfun{n}{\rho,\,i}  &=
	f_n \left(
	\gfun{p}{\rho,\,i-1} ,
	\gfun{s}{\rho,\,i-1} 
		\right)
	& \text{ where } p = \Pred(\rho,n) \text{ and }
			s = \Succ(\rho,n)
\end{align*}

Intuitively, the main difference between \mfp and \mop is that
\mop does not merge data flow values at the intermediate program
points whereas \mfp does.  We use the superscript ``\text{\op}'' for
the former and ``\text{\fp}'' for the latter.

\begin{align}
\gfun{n}{\fp}  &= \gfun{n}{\prog}
	\label{eq:mfp.def}
	\\
\gfun{n}{\op}  &=
 \displaystyle\bigsqcap_{\renewcommand{\arraystretch}{.8}%
			\scriptsize%
			\begin{array}{@{}c@{}}
			\scriptsize \rho \in \paths{n}
			\end{array}} 
		\gfun{n}{\rho}
	\label{eq:mop.def}
\end{align}

Since $\gfunnode{n} = (\gfunnodein{n}{}, \gfunnodeout{n}{})$, 
meet for \gfun{n}{\op} is defined over the components as described below.
\begin{align*}
\gfunnodein{n}{\op}  &=
 \displaystyle\bigsqcap_{\renewcommand{\arraystretch}{.8}%
			\scriptsize%
			\begin{array}{@{}c@{}}
			\scriptsize \rho \in \paths{n}
			\end{array}} 
		\gfunnodein{n}{\rho}
\\
\gfunnodeout{n}{\op}  &=
 \displaystyle\bigsqcap_{\renewcommand{\arraystretch}{.8}%
			\scriptsize%
			\begin{array}{@{}c@{}}
			\scriptsize \rho \in \paths{n}
			\end{array}} 
		\gfunnodeout{n}{\rho}
\end{align*}

Note that 
\gfunnodein{n}{\op}  and \gfunnodeout{n}{\op} are pairs of forward and backward components.

\begin{lemma}
   $\forall n \in \nodes, \forall \rho \in \paths{n}, \forall i\geq0, \gfun{n}{\prog,\,i} \sqsubseteq \gfun{n}{\rho,\,i}$
   \label{lemma:mfp:paths}
  \end{lemma}
\begin{proof}
We prove the lemma by inducting on steps $i$ required for computing the solution. Termination in a finite number of
steps is ensured by the fact that our flow functions are monotonic and all strictly descending chains in the lattice are finite.

  The base case for induction is step $0$ where value for all nodes are initialized to $\initval{n}$ for both
   \text{$\gfun{n}{\prog,\,0}$} and \text{$\gfun{n}{\rho,\, 0}$} for all $n$ and all $\rho$.  Thus the base case trivially holds.

   For the inductive hypothesis, assume that the lemma holds for some step $k$.
   Then, 
	\[
	\forall n \in \nodes,\forall \rho \in \paths{n}, \gfun{n}{\prog,\,k} \sqsubseteq \gfun{n}{\rho,\,k}
	\]
   For the inductive step of the proof, consider an arbitrary node $n$ in an arbitrary qualified control flow path $\rho$ such that
   $p$ is a predecessor of $n$ in $\rho$ and $s$ is a successor of $n$ in $\rho$.
   \begin{align*}
     \gfun{n}{\prog,\,k+1}  &=
 \displaystyle\bigsqcap_{\renewcommand{\arraystretch}{.8}%
			\scriptsize%
			\begin{array}{@{}c@{}}
			\scriptsize q \in \Pred(n) 
			\\
			\scriptsize r \in \Succ(n)
			\end{array}} 
	f_n \left(
	\gfun{q}{\prog,\,k} ,
	\gfun{r}{\prog,\,k} 
		\right) &
	\\
  \gfun{n}{\rho,\,k+1}  &=
	f_n \left(
	\gfun{p}{\rho,\,k} ,
	\gfun{s}{\rho,\,k} 
		\right)
	\\
\gfun{p}{\prog,\,k} & \sqsubseteq \gfun{p}{\rho,\,k} & \text{(inductive hypothesis for $q=p$)}
	\\
\gfun{s}{\prog,\,k} & \sqsubseteq \gfun{s}{\rho,\,k} & \text{(inductive hypothesis for $r=s$)}
	\\
	f_n \left(
	\gfun{p}{\prog,\,k} ,
	\gfun{s}{\prog,\,k} 
		\right) 
	& \sqsubseteq
	f_n \left(
	\gfun{p}{\rho,\,k} ,
	\gfun{s}{\rho,\,k} 
		\right) 
		& \text{(applying $f_n$)}
	\\
  \gfun{n}{\prog,\,k+1} & \sqsubseteq \gfun{n}{\rho,\,k+1}
   \end{align*}
  Hence the lemma.
\end{proof}

\begin{lemma}
   $\forall n\in \nodes, \gfun{n}{\text{fp}} \sqsubseteq \gfun{n}{\text{op}}$
   \label{lemma:mfp:mop}
  \end{lemma}

  \begin{proof}
	Follows from Definitions~(\ref{eq:mfp.def}),~(\ref{eq:mop.def}) and Lemma~\ref{lemma:mfp:paths}.
  \end{proof}

    \section{Defining \mfp and \mop for Our Method}
    \label{sec:mop:demand}

In this section, we instantiate the \mfp and \mop defined for general bidirectional problems in
Section~\ref{sec:mfp.mop.bidirectional}, to our demand-driven alias analysis.


Let \aval and \dval denote the vectors of demand and alias data flow
values associated with all nodes (program points) in the control flow
graph.  In general, $\aval$ includes both $\ain$ and $\aout$ values,
and $\dval$ includes both $\din$ and $\dout$ values.  Our algorithm
computes a pair \text{$(\aval,\dval)$} as the maximum fixed point of
function \text{$\vecf: \aval\times \dval \to \aval\times \dval$} which
is defined in terms of mutually recursive functions \text{$\vecd:
  \aval\times \dval \to \dval$} and 
\text{$\veca: \aval \times \dval \to \aval$}
that compute \dval and \aval, respectively.  
\begin{align*}
(\aval,\dval)  & = \vecf\, (\aval,\dval)
	\\
\aval          & = \veca\, (\aval,\dval) 
	\\
\dval          & = \vecd\, (\aval,\dval) 
\end{align*}

We require
\vecd and \veca to be monotonic. 
As is usual in data flow analysis, we distinguish between the \mfp (maximum fixed point) and \mop (meet over paths) solutions. 
 We use the superscript ``\emph{fp}'' for the former and ``\emph{op}'' for the latter.
\begin{align*}
\mfpaval & = \mfpveca\, (\mfpaval,\mfpdval) 
	\\
\mfpdval & = \mfpvecd\, (\mfpaval,\mfpdval) 
	\\
\mopaval & = \mopveca\, (\mopaval,\mopdval) 
	\\
\mopdval & = \mopvecd\, (\mopaval,\mopdval) 
\end{align*}
Note that \mfp is a fixed point computation of the data flow values over a control flow graph. 
%
Let $\nda_n$ and $\ndd_n$ denote the node level flow functions for computing aliases and demands respectively.
We instantiate the \mfp for general bidirectional problems in
Section~\ref{sec:mfp.mop.bidirectional}.
Then, $\forall n \in \nodes$,
the \mfp solution \text{$(\mfpavalseries, \mfpdvalseries)$} is computed as follows:

\begin{align}
\gfun{n}{\fp} = (\afun{n}{\fp}, \dfun{n}{\fp}) &=  (\afun{n}{\prog}, \dfun{n}{\prog})
\label{eq:mfp.def.ad} 
\intertext{where,}
\afun{n}{\prog,\,0} &= \top
\nonumber
\\
\dfun{n}{\prog,\,0} &= \top
\nonumber
\\
 \forall i > 0, \afun{n}{\prog,\,i} &=
 \displaystyle\bigsqcap_{\renewcommand{\arraystretch}{.8}%
			\scriptsize%
			\begin{array}{@{}c@{}}
			\scriptsize p \in \Pred(n) \\
			\scriptsize s \in \Succ(n) 
			\end{array}} 
	\nda_n \left(
	\afun{p}{\prog,\,i-1} ,
	\dfun{s}{\prog,\,i-1} 
		\right)
	\label{afun:ff}	\\
\forall i > 0, \dfun{n}{\prog,\,i} &=
 \displaystyle\bigsqcap_{\renewcommand{\arraystretch}{.8}%
			\scriptsize%
			\begin{array}{@{}c@{}}
			\scriptsize p \in \Pred(n) \\
			\scriptsize s \in \Succ(n)
			\end{array}} 
	\ndd_n \left(
	\afun{p}{\prog,\,i-1} ,
	\dfun{s}{\prog,\,i-1} 
		\right)
		\label{dfun:ff}	
\end{align}

We instantiate the \mop for general bidirectional problems in
Section~\ref{sec:mfp.mop.bidirectional}.
Then, $\forall n \in \nodes$,
the \mop solution \text{$(\mopavalseries, \mopdvalseries)$} is computed as follows:
\begin{align}
\gfun{n}{\op} = (\afun{n}{\op}, \dfun{n}{\op}) &= 
	(\displaystyle\bigsqcap_{\rho \in \paths{n}} \; \eqmoptaval{\rho}{n} ,
	\displaystyle\bigsqcap_{\rho \in \paths{n}} \; \eqmoptdval{\rho}{n})
\label{eq:mop.def.ad} 
\intertext{where, for a path $\rho \in \paths{n}$,}
\eqmoptaval{\rho}{n} &=
		\begin{cases}
		\boundary & n \text{ is } \Start{}
		\\
			\nda_n\left(\eqmoptaval{\rho}{p}, \eqmoptdval{\rho}{s}\right)
		& p = \Pred(\rho, n), s = \Succ(\rho, n)
		\\
		\end{cases}
		\label{eqmoptaval}
	\\
\eqmoptdval{\rho}{n} &=
		\begin{cases}
		\boundary & n \text{ is } \End{}
		\\
			\ndd_n\left(\eqmoptaval{\rho}{p}, \eqmoptdval{\rho}{s}\right)
		& p = \Pred(\rho, n),s = \Succ(\rho, n)
		\\
		\end{cases}
		\label{eqmoptdval}
\end{align}

Note that aggregated set of demands at node $n$ represented by
$\eqmopdval{n}$ may never be used in computing alias information at
$n$ along any qualified control flow path.  However, every demand in
$\eqmopdval{n}$ is necessarily used to compute alias information at
$n$ along some qualified control flow path.




  \newcommand{\aconc}[2]{\text{$\mathbb{A}_{#1}^{#2}$}\xspace}

  \section{Soundness Proof}\label{sec:soundness-proof}

\newcommand{\vcnode}{\text{$v$}\xspace}

In this section, we prove the soundness of our
demand-driven alias analysis algorithm (\id).

Consider a program \prog with a virtual function call
of the form\footnote{We assume that all calls are normalized to the form  $x\rightarrow\text{\emph{vfun}} ()$ by 
introducing temporary variables, if the need be.} $x\rightarrow\text{\emph{vfun}} ()$ at a node \vcnode. Then,
the objective of our analysis is to compute a sound approximation of all aliases
of $x$ at \vcnode  
considering all possible executions
of the program. We formalize this notion below.


Following the notational convention introduced in the previous section,
let $\afun{n}{\id}$ denote the set of aliases computed by algorithm
$\id$ at an arbitrary node $n$ of the program \prog.  Similarly, let
$\dfun{n}{\id}$ denote the set of access expressions for which demands
are raised at node $n$, when executing algorithm $\id$. 
  By abuse of notation, we also use $\afun{n}{\cm}$ to
denote the alias information computed in the concrete memory
at $n$. Note that this implies $\forall n\in \nodes, \afun{n}{\cm} =
  \displaystyle\bigsqcap_{\rho \in \paths{n}}\afun{n}{\cm,\rho}$.


   We define function \text{\restrict{\aliasing}{\demand}}
 to restrict alias
   information in \aliasing to the demands in \demand.
Since, 
\text{\aliasing = (\ain, \aout)} and \text{\demand = (\din, \dout)}, 
function \Restrict is computed component-wise 
to yield a pair \text{(\ain, \aout)}.
Formally, we define it as
\begin{align} 
\restrict{\aliasing}{\demand} = (\{(a,b)\vert a\in \din, (a,b) \in \ain\}, 
\{(a,b)\vert a\in \dout, (a,b) \in \aout\})
\label{eq:restrict}
\end{align}
  As an example,
   if $\ain = \{(d_1, a_1), (d_2, a_2)\}$, $\aout = \{(d_1,a_3), (d_2, a_4)\}$,
$\din = \{d_1\}$ and $\dout = \{d_2\}$, then
   $\restrict{\aliasing}{\demand} = \left(\left\{(d_1, a_1)\}, \{(d_2, a_4)\right\}\right)$.

   \newcommand{\alg}{\text{{\sc aa}}}
   \newcommand{\dfunhat}[2]{\text{$\widehat{\demand}_{#1}^{\,#2}$}\xspace}

\begin{definition}
  An alias analysis algorithm $\alg$ is \emph{sound for a program
    $\prog$ with respect to a set of demands $\dfun{n}{}$ at node $n$}
    iff $\restrict{\afun{n}{\alg}}{\dfun{n}{}} \sqsubseteq
  \restrict{\afun{n}{\cm}}{\dfun{n}{}}$.
\end{definition}
Since the end goal of our analysis is to compute the set of all
incoming aliases of $x$ at node $\vcnode \in \origin$, where
$x\rightarrow \text{\em vfun}()$ is a virtual function call at
\vcnode, we define the overall soundness of $\alg$ as follows.
\begin{definition}
  \label{def:sound}
An alias analysis algorithm $\alg$ is \emph{sound} iff for every program
$\prog$ and every node $v \in \origin$ with a virtual function call of the form
$x\rightarrow \text{\em vfun}()$, $\alg$ is sound for
$\prog$ with respect to $\dfun{\vcnode}{} = (\{x\}, \emptyset)$.
\end{definition}
Note that $\dfun{\vcnode}{}$ has two components in general, and we
are interested only in the demand $x$ in the component representing
the demands at the entry of $\vcnode$.

 \begin{thm}
Algorithm \id is sound.
  \label{thm:soundness}
 \end{thm}
\begin{proof}

  Let $\prog$ be a program of the form discussed in Section~\ref{sec:formulation}
  and let $\vcnode$ be a node in the set $\origin$, where
  $x \rightarrow {\text {\em vfun}}()$ is the instruction at
  $\vcnode$.  To prove the theorem, we need to show that
  $\restrict{\afun{\vcnode}{\id}}{\dfun{\vcnode}{}} \sqsubseteq
  \restrict{\afun{\vcnode}{\cm}}{\dfun{\vcnode}{}}$, where
  $\dfun{\vcnode}{}$ is as in Definition~\ref{def:sound}.  We actually
  show something stronger.  Specifically, we prove soundness of $\id$
  for $\prog$ with respect to the common set of demands $\dfunhat{n}{}$
  that are raised at $n$, (using equations (\ref{eqmoptaval}) and
  (\ref{eqmoptdval})) along every qualified control flow path through
  $n$.  
  In addition, we show in Lemma~\ref{lemma-initial-demand-sound}
  that $x \in \din_{\vcnode}^{\rho}$ for every $\vcnode \in
  \origin$ and for every qualified control flow path $\rho \in \paths{\vcnode}$.
  Together, these prove Theorem~\ref{thm:soundness}.

In order to establish that $\forall n\in \nodes,
\restrict{\afun{n}{\id}}{\dfunhat{n}{}} \sqsubseteq
  \restrict{\afun{n}{\cm}}{\dfunhat{n}{}}$, we proceed in three
  steps. 
\begin{itemize}
\item First, we show in Lemma~\ref{lemma:mfp} that our algorithm
  computes the \mfp solution for alias information. More
  formally, $\forall n \in \nodes, \afun{n}{\id} = \afun{n}{\fp}$.
\item  Next, we show in Lemma~\ref{lemma:mfp-mop} that $\mfp$ is sound with
  respect to the \mop solution.  In other words, $\forall n\in \nodes,
  \afun{n}{\fp} \sqsubseteq \afun{n}{\op}$.  
\item Finally, we show in
  Lemma~\ref{lemma:mop:ideal} that the $\mop$ solution for aliases is
  sound with respect to $\dfunhat{n}{}$.  In other words, $\forall
  n\in \nodes, \restrict{\afun{n}{\op}}{\dfunhat{n}{}} \sqsubseteq
  \restrict{\afun{n}{\cm}}{\dfunhat{n}{}}$.
\end{itemize}
Putting the parts together proves the theorem.
\end{proof}

We now discuss in detail the various lemmas used in the proof of
Theorem~\ref{thm:soundness}.

\subsection{Algorithm \id Computes \mfp Solution}
\label{sec:id-mfp}
In this section we prove that our flow functions are monotonic and algorithm \id computes \mfp solution.

   \begin{lemma}
   Flow functions $\nda_n$ and $\ndd_n$ which computes
 \afun{}{} and \dfun{}{} are monotonic flow functions.
 \label{lemma:monotonicff}
 \end{lemma}
\begin{proof}
 Flow functions $\nda_n$ and $\ndd_n$ computes
 \afun{}{} and \dfun{}{} as shown in Equations~\ref{afun:ff} and~\ref{dfun:ff}.
 $\nda_n$ is defined in terms of $\ain_n$ and $\aout_n$ (Equations~\ref{eq:ain} and~\ref{eq:aout}) and
 $\ndd_n$ is defined in terms of $\din_n$ and $\dout_n$ (Equations~\ref{eq:din} and~\ref{eq:dout}).
  
  All the set operations are monotone except the set difference. 
It is also monotonic for us because $\dkill_n$(Equation~\ref{eq:dkill}) 
is constant.
 Function \dn, identifies abstract names for access expressions appearing in the program, is monotone. Thus function $\ldgen_n$(Equation~\ref{eq:ldgen})
 and $\rdgen_n$(Equation~\ref{eq:rdgen}) are monotone.
 Function $\dgen_n$(Equation~\ref{eq:dgen}) is monotone as it is computed using $\ldgen_n$ and $\rdgen_n$.
 Thus $\din_n$ and $\dout_n$ which are computed using $\dgen_n$ and $\dkill_n$ are monotonic.
 
 Function $\AGen_n$(Equation~\ref{eq:agen}) is monotone. 
 Function $\AKill_n$(Equation~\ref{eq:akill}) is a constant.
 Thus $\ain_n$ and $\aout_n$ which are computed using $\AGen_n$ and $\AKill_n$ are monotonic.
 
 Thus the flow functions $\nda_n$ and $\ndd_n$ are monotonic.
\end{proof}

\begin{lemma}
Consider the series
\text{$\dval^{0}\rightsquigarrow 
\aval^{0}
\rightsquigarrow \dval^{1} \rightsquigarrow \aval^1
\rightsquigarrow \dval^2 \rightsquigarrow \aval^2
\ldots \dval^k \rightsquigarrow \aval^k
$} computed as alternate fixed points by keeping \aval constant while computing \dval and vice-versa using the following
equations:
 \begin{align*}
 \dval^{i} &= \vecd (\aval^{i-1}, \dval^{i})
 \\
 \aval^i &= \veca (\aval^{i}, \dval^{i})
 \end{align*}
The final fixed points
\text{$\aval^k ,\dval^k$} define the maximum fixed point of \text{$\;\vecf: \aval\times \dval \to \aval\times \dval$}.
\label{lemma:mfp}
\end{lemma}

\begin{proof}
Lattice for aliases and demands are finite and thus complete lattice. 
 \veca and \vecd are monotone flow functions as discussed in Lemma~\ref{lemma:monotonicff} 
and computes maximum fixed point starting from \text{$\dval^0 = \vec{\emptyset}$} and 
 \text{$\aval^0 = \vec{\emptyset}$} which is the $\top$ 
value of the lattice. The lemma follows from
\bekic's theorem~\cite{Winskel:1993:FSP:151145}.
\end{proof}

\subsection{\mfp is a Sound Approximation of \mop}

In this section we prove that our algorithm which computes \mfp solution as discussed in Section~\ref{sec:id-mfp},
is a sound approximation of the \mop solution.
   \begin{lemma}
   $\forall n\in \nodes, (\afun{n}{\fp}, \dfun{n}{\fp}) \sqsubseteq (\afun{n}{\op}, \dfun{n}{\op})$
   \label{lemma:mfp-mop}
  \end{lemma}
\begin{proof}
 Flow function $f_n$ is modelled as monotonic flow functions $\nda_n$ and $\ndd_n$ (Lemma~\ref{lemma:monotonicff})
 to compute
 \afun{}{} and \dfun{}{} as shown in~\ref{afun:ff} and~\ref{dfun:ff}.
 The lemma follows from considering $\dfv = (\afun{}{}, \dfun{}{})$ 
as shown in equations~\ref{eq:mfp.def.ad} and~\ref{eq:mop.def.ad} in 
 lemma~\ref{lemma:mfp:mop}. 
\end{proof}

\subsection{Function {\sf Restrict} Distributes Over Meet for Aliases}

\begin{lemma}
Function {\sf Restrict} distributes over meet for aliases.
\label{lemma:restrict:distributive}
\end{lemma}

\begin{proof}
We show below that function \text{\sf Restrict} distributes over $\sqcap$ (which is $\cup$ for aliases).
We have formally defined {\sf Restrict} in Equation~\ref{eq:restrict}.
For convenience, we prove distributivity considering a single component---it is easy to see that the argument holds for multiple components also.

We first show that
$(a,b) \in \restrict{\displaystyle\bigsqcap_{i=1}^{\infty}\afun{i}{}}{\djoin{}{}} 
\Rightarrow (a,b) \in  \displaystyle\bigsqcap_{i=1}^{\infty} \restrict{\afun{i}{}}{\djoin{}{}}$ 
(Equation~\ref{eq:restrict:distributive1}) and then show that
$(a,b) \in  \displaystyle\bigsqcap_{i=1}^{\infty} \restrict{\afun{i}{}}{\djoin{}{}}
\Rightarrow (a,b) \in \restrict{\displaystyle\bigsqcap_{i=1}^{\infty}\afun{i}{}}{\djoin{}{}}$
(Equation~\ref{eq:restrict:distributive2}).


\begin{align}
\text{Let} (a,b) \in
	& \; \restrict{\displaystyle\bigsqcap_{i=1}^{\infty}\afun{i}{}}{\djoin{}{}} 
		\nonumber\\
\Rightarrow
	& \; a\in \djoin{}{}, (a,b) \in\displaystyle\bigsqcap_{i=1}^{\infty}\afun{i}{} & (\text{by definition})
		\nonumber\\
\Rightarrow
	& 	
		\; a\in\djoin{}{},\; \exists j, 1\leq j < \infty, (a,b) \in \afun{j}{}
		\nonumber\\
\Rightarrow
	& \; \exists j, 1\leq j<\infty, (a,b) \in \restrict{\afun{j}{}}{\djoin{}{}} 
		\nonumber\\
\Rightarrow
	& \; (a,b) \in  \displaystyle\bigsqcap_{i=1}^{\infty} \restrict{\afun{i}{}}{\djoin{}{}}
\label{eq:restrict:distributive1}
\end{align}

\begin{align}
\text{For the converse, let} (a,b) \in
	& \; \displaystyle\bigsqcap_{i=1}^{\infty} \restrict{\afun{i}{}}{\djoin{}{}}
		\nonumber\\
\Rightarrow
	& \; \exists j, 1\leq j<\infty, (a,b) \in \restrict{\afun{j}{}}{\djoin{}{}} 
		\nonumber\\
\Rightarrow
	& \; \exists j, 1\leq j < \infty, a\in\djoin{}{}, (a,b) \in \afun{j}{}
		\nonumber\\
\Rightarrow
	& \; a\in \djoin{}{}, (a,b) \in\displaystyle\bigsqcap_{i=1}^{\infty}\afun{i}{} 
		\nonumber\\
\Rightarrow
	& \; (a,b)\in\restrict{\displaystyle\bigsqcap_{i=1}^{\infty}\afun{i}{}}{\djoin{}{}} & (\text{by definition})
\label{eq:restrict:distributive2}
\end{align}

Hence, {\sf Restrict} distributes over meet for the alias component.
\end{proof}


\subsection{\mop for Aliases is Sound}
\newcommand{\alphacm}{\llbracket\alpha\rrbracket\xspace}
\newcommand{\betacm}{\llbracket\beta\rrbracket\xspace}
\newcommand{\rhscm}{\llbracket\rhs\rrbracket\xspace}

The three lemmas in this section form the major part of our soundness proof.

Let $\alphacm$ represent the runtime location of access expression $\alpha$ in the concrete memory and
$\balpha$ represents the abstract name of $\alpha$ as computed by our heap abstraction.
They are illustrated in Figure~\ref{fig:cm-am}.

 \begin{lemma}
\mop for aliases at node $n$ is sound with respect to $\djoin{n}{\jop}$.
  \label{lemma:mop:ideal}
 \end{lemma}
\begin{proof}

  The lemma is proved by first showing that if we restrict our
  analysis to a single (arbitrary) qualified control flow path $\rho$
  in $\paths{n}$, the alias information computed by 
   Equations~(\ref{eqmoptaval}) and~(\ref{eqmoptdval}) soundly approximates
  the alias information computed in the concrete memory at $n$ when
  considering the same qualified control flow path.  Formally,
  \[\forall n\in \nodes, \forall \rho\in \paths{n} :
  \restrict{\afun{n}{\rho}}{\dfun{n}{\rho}} \sqsubseteq
  \restrict{\afun{n}{\cm, \rho}}{\dfun{n}{\rho}}.\] The above relation
  is established in Lemma~\ref{lemma-soundness-flow-function-proof}.
  Since $\forall n\in \nodes, \forall \rho\in\paths{n} : \dfun{n}{\rho} \sqsubseteq \dfunhat{n}{}$, 
where \dfunhat{n}{} represents the common set of demands along every qualified control flow path through $n$,
it follows
  that
  \[\forall n \in \nodes, \forall \rho\in \paths{n} : 
  \restrict{\afun{n}{\rho}}{\dfunhat{n}{}} \sqsubseteq
  \restrict{\afun{n}{\cm,\rho}}{\dfunhat{n}{}}.\]
  From monotonicity of the meet operation, we have
  \[\forall n \in \nodes : 
\displaystyle\bigsqcap_{\rho \in \paths{n}}\,\restrict{\afun{n}{\rho}}{\dfunhat{n}{}} \sqsubseteq 
\displaystyle\bigsqcap_{\rho \in \paths{n}}\,\restrict{\afun{n}{\cm,\rho}}{\dfunhat{n}{}}.\]
  We also know from Lemma~\ref{lemma:restrict:distributive} that {\sf Restrict}
  distributes over meet in its first argument (i.e. aliases). Therefore,
  we have
  \[\forall n \in \nodes : 
\restrict{\displaystyle\bigsqcap_{\rho \in \paths{n}}\afun{n}{\rho}}{\dfunhat{n}{}} \sqsubseteq 
\restrict{\displaystyle\bigsqcap_{\rho \in \paths{n}}\afun{n}{\cm,\rho}}{\dfunhat{n}{}}.\]

  By definition, $\forall n\in\nodes : \displaystyle\bigsqcap_{\rho \in \paths{n}}\afun{n}{\rho} = \afun{n}{\op}$
    and $\forall n\in\nodes : \displaystyle\bigsqcap_{\rho \in \paths{n}}\afun{n}{\cm,\rho} = \afun{n}{\cm}$.
      Therefore, 
\[\forall n \in \nodes : \restrict{\afun{n}{\op}}{\dfunhat{n}{}} \sqsubseteq \restrict{\afun{n}{\cm}}{\dfunhat{n}{}}.\]  
Hence the lemma is proved.
\end{proof}

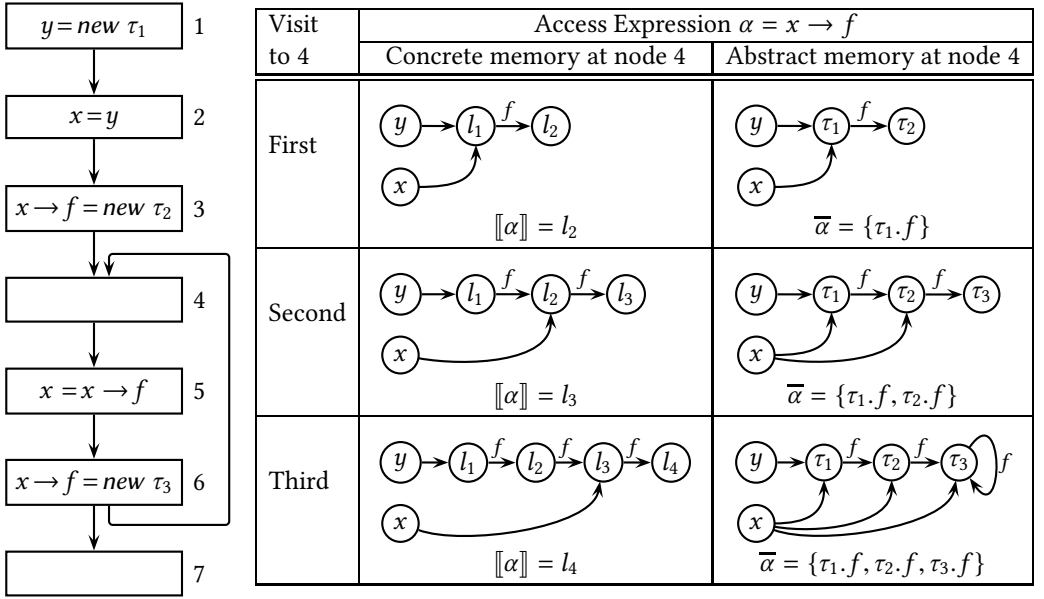
\begin{figure}[t]
 
 \begin{center}
\begin{tabular}{@{}cc@{}}
 \begin{tabular}{@{}c@{}}
  \begin{pspicture}(4,0)(36,80)
\putnode{n1}{origin}{18}{75}{\psframebox{\makebox[21mm]{$y \!=\! \new\ \tau_1$\rule[-.25em]{0em}{1em}}}}
\putnode{n2}{n1}{0}{-12}{\psframebox{\makebox[21mm]{{\white A}$x \!=\! y${\white A}}}}
\putnode{n3}{n2}{0}{-12}{\psframebox{\makebox[21mm]{$x\!\rightarrow\! f \!=\! \new\ \tau_2$\rule[-.25em]{0em}{1em}}}}
\putnode{n4}{n3}{0}{-12}{\psframebox{\makebox[21mm]{\white$Type(a) \!\!=\!\! Real$}}}
\putnode{n5}{n4}{0}{-12}{\psframebox{\makebox[21mm]{$x =\! x\rightarrow\! f$\rule[-.25em]{0em}{1em}}}}
\putnode{n6}{n5}{0}{-12}{\psframebox{\makebox[21mm]{$x\!\rightarrow\! f \!=\! \new\ \tau_3$\rule[-.25em]{0em}{1em}}}}
\putnode{n7}{n6}{0}{-12}{\psframebox{\makebox[21mm]{\white$Type(a) \!=\! Real$}}}
\ncline{->}{n1}{n2}
\ncline{->}{n2}{n3}
\ncline{->}{n3}{n4}
\ncline{->}{n4}{n5}
\ncloop[angleA=270,angleB=90,offset=2,loopsize=-16,armA=2.5,armB=3,linearc=.5]{->}{n6}{n4}
\ncline{->}{n5}{n6}
\ncline{->}{n6}{n7}
\putnode{w}{n1}{14}{0}{1}
\putnode{w}{n2}{14}{0}{2}
\putnode{w}{n3}{14}{0}{3}
\putnode{w}{n4}{14}{0}{4}
\putnode{w}{n5}{14}{0}{5}
\putnode{w}{n6}{14}{0}{6}
\putnode{w}{n7}{14}{0}{7}
\end{pspicture}
\end{tabular} &

\begin{tabular}{|l|c|c|}
\hline
Visit &
\multicolumn{2}{|c|}{Access Expression $\alpha = x\rightarrow f$}
\\ \cline{2-3} 
to 4
&
Concrete memory at node 4 & Abstract memory at node 4\\ \hline\hline

First &

\begin{tabular}{c}
  	\begin{pspicture}(0,1)(40,17)
	\psset{nodesep=-1.5}
	
	\putnode{y}{origin}{2}{11}{\pscirclebox{$y$}}
	\putnode{l1}{y}{10}{0}{\pscirclebox[framesep=0.5]{$l_1$}}
	\putnode{l2}{l1}{10}{0}{\pscirclebox[framesep=0.5]{$l_2$}}
	\putnode{x}{y}{0}{-8}{\pscirclebox{$x$}}
	
	\ncline[nodesep=0]{->}{y}{l1}
	\ncline[nodesep=0]{->}{l1}{l2}
	\aput[1pt](.4){\footnotesize$f$}

	\nccurve[angleA=0,angleB=270,nodesep=-0.2,ncurv=1]{->}{x}{l1}
	\end{pspicture}

\end{tabular}
&
\begin{tabular}{c}
  	\begin{pspicture}(0,1)(35,17)
	\psset{nodesep=-1.5}
	
	\putnode{y}{origin}{2}{11}{\pscirclebox{$y$}}
	\putnode{l1}{y}{10}{0}{\pscirclebox[framesep=0.5]{$\tau_1$}}
	\putnode{l2}{l1}{10}{0}{\pscirclebox[framesep=0.5]{$\tau_2$}}
	\putnode{x}{y}{0}{-8}{\pscirclebox{$x$}}
	
	\ncline[nodesep=0]{->}{y}{l1}
	\ncline[nodesep=0]{->}{l1}{l2}
	\aput[1pt](.4){\footnotesize$f$}
	\nccurve[angleA=0,angleB=270,nodesep=-0.2,ncurv=1]{->}{x}{l1}
	\end{pspicture}

\end{tabular} 
\\
\rule[-.5em]{0em}{1em}%
&
$\alphacm = l_2$ & $\balpha = \{\tau_1.f\}$ 
\\ \hline
Second &

\begin{tabular}{c}
  	\begin{pspicture}(0,1)(40,17)
	\psset{nodesep=-1.5}
	
	\putnode{y}{origin}{2}{11}{\pscirclebox{$y$}}
	\putnode{l1}{y}{10}{0}{\pscirclebox[framesep=0.5]{$l_1$}}
	\putnode{l2}{l1}{10}{0}{\pscirclebox[framesep=0.5]{$l_2$}}
	\putnode{l3}{l2}{10}{0}{\pscirclebox[framesep=0.5]{$l_3$}}
	\putnode{x}{y}{0}{-8}{\pscirclebox{$x$}}
	
	\ncline[nodesep=0]{->}{y}{l1}
	\ncline[nodesep=0]{->}{l1}{l2}
	\aput[1pt](.4){\footnotesize $f$}
	\ncline[nodesep=0]{->}{l2}{l3}
	\aput[1pt](.4){\footnotesize$f$}

	\nccurve[angleA=-10,angleB=270,nodesep=-0.2,ncurv=.7]{->}{x}{l2}
	\end{pspicture}

\end{tabular}
&
\begin{tabular}{c}
  	\begin{pspicture}(0,1)(35,17)
	\psset{nodesep=-1.5}
	
	\putnode{y}{origin}{2}{11}{\pscirclebox{$y$}}
	\putnode{l1}{y}{10}{0}{\pscirclebox[framesep=0.5]{$\tau_1$}}
	\putnode{l2}{l1}{10}{0}{\pscirclebox[framesep=0.5]{$\tau_2$}}
	\putnode{l3}{l2}{10}{0}{\pscirclebox[framesep=0.5]{$\tau_3$}}
	\putnode{x}{y}{0}{-8}{\pscirclebox{$x$}}
	
	\ncline[nodesep=0]{->}{y}{l1}
	\ncline[nodesep=0]{->}{l1}{l2}
	\aput[1pt](.4){\footnotesize$f$}
	\ncline[nodesep=0]{->}{l2}{l3}
	\aput[1pt](.4){\footnotesize$f$}
	\nccurve[angleA=0,angleB=270,nodesep=-0.2,ncurv=1]{->}{x}{l1}
	\nccurve[angleA=-10,angleB=270,nodesep=-0.2,ncurv=.7]{->}{x}{l2}
	\end{pspicture}

\end{tabular} 
\\
\rule[-.5em]{0em}{1em}%
&
$\alphacm = l_3$ & $\balpha = \{\tau_1.f,\tau_2.f\}$ 
\\ \hline
Third &
\begin{tabular}{c}
  	\begin{pspicture}(0,1)(40,17)
	\psset{nodesep=-1.5}
	
	\putnode{y}{origin}{2}{11}{\pscirclebox{$y$}}
	\putnode{l1}{y}{9}{0}{\pscirclebox[framesep=0.5]{$l_1$}}
	\putnode{l2}{l1}{9}{0}{\pscirclebox[framesep=0.5]{$l_2$}}
	\putnode{l3}{l2}{9}{0}{\pscirclebox[framesep=0.5]{$l_3$}}
	\putnode{l4}{l3}{9}{0}{\pscirclebox[framesep=0.5]{$l_4$}}
	\putnode{x}{y}{0}{-8}{\pscirclebox{$x$}}
	
	\ncline[nodesep=0]{->}{y}{l1}
	\ncline[nodesep=0]{->}{l1}{l2}
	\aput[1pt](.4){\footnotesize$f$}
	\ncline[nodesep=0]{->}{l2}{l3}
	\aput[1pt](.4){\footnotesize$f$}
	\ncline[nodesep=0]{->}{l3}{l4}
	\aput[1pt](.4){\footnotesize$f$}
	\nccurve[angleA=-20,angleB=260,nodesep=-0.2,ncurv=.5]{->}{x}{l3}
	\end{pspicture}

\end{tabular}
&
\begin{tabular}{c}
  	\begin{pspicture}(0,1)(35,17)
	\psset{nodesep=-1.5}
	
	\putnode{y}{origin}{2}{11}{\pscirclebox{$y$}}
	\putnode{l1}{y}{9}{0}{\pscirclebox[framesep=0.5]{$\tau_1$}}
	\putnode{l2}{l1}{9}{0}{\pscirclebox[framesep=0.5]{$\tau_2$}}
	\putnode{l3}{l2}{9}{0}{\pscirclebox[framesep=0.5]{$\tau_3$}}
	\putnode{x}{y}{0}{-8}{\pscirclebox{$x$}}
	
	\ncline[nodesep=0]{->}{y}{l1}
	\ncline[nodesep=0]{->}{l1}{l2}
	\aput[1pt](.4){\footnotesize$f$}
	\ncline[nodesep=0]{->}{l2}{l3}
	\aput[1pt](.4){\footnotesize$f$}	
      \nccurve[angleA=55,angleB=-55,ncurv=4,nodesep=-.5]{->}{l3}{l3}
      \aput[1pt](.5){\footnotesize$f$}
      \nccurve[angleA=0,angleB=270,nodesep=-0.2,ncurv=1]{->}{x}{l1}
	\nccurve[angleA=-10,angleB=270,nodesep=-0.2,ncurv=.7]{->}{x}{l2}
	\nccurve[angleA=-20,angleB=260,nodesep=-0.2,ncurv=.5]{->}{x}{l3}
	\end{pspicture}

\end{tabular}
\\
\rule[-.5em]{0em}{1em}%
&
$\alphacm = l_4$ & $\balpha = \{\tau_1.f,\tau_2.f,\tau_3.f\}$
\\ \hline
\end{tabular}

\end{tabular}
\caption{
\protect
Illustrating the 
abstract names and concrete locations of an access expression $\alpha$, denoted by $\balpha$ and
$\alphacm$, respectively.
We show the concrete memory and the abstract memory (under type-based abstraction) for different iterations of the loop in the
program.}
\label{fig:cm-am}

\end{center}
\end{figure}

\begin{lemma}
The alias information computed by  Equations~(\ref{eqmoptaval}) and 
	(\ref{eqmoptdval}) soundly approximates
  the alias information computed in the concrete memory at $n$ when
  considering the same qualified control flow path.
 
 \label{lemma-soundness-flow-function-proof}
\end{lemma}

\begin{proof}
  
  We show below that the alias information computed in $\afun{n}{\rho}$
  is a sound approximation of $\afun{n}{\cm,\rho}$. Formally, 
\[\forall n\in \nodes, \forall \rho\in \paths{n} : \restrict{\afun{n}{\rho}}{\dfun{n}{\rho}} 
\sqsubseteq \restrict{\afun{n}{\cm, \rho}}{\dfun{n}{\rho}}.\]
For convenience,
  we re-write this proof obligation in the following equivalent form.
\[\forall n \in \nodes, \forall \rho \in \paths{n}:\;\; 
(\alphacm, \betacm) \in \afun{n}{\cm,\rho} 
\wedge \balpha \subseteq \dfun{n}{\rho}
\Rightarrow \balpha \times \bbeta \subseteq \afun{n}{\rho} 
\wedge \balpha \subseteq \dfun{n}{\rho}
\]

The proof is organized along the following steps:
\begin{enumerate}[(a)]
\item Devising a notation for inducting on the forward components of qualified control flow paths.
\item Refining the proof obligation to use the data flow values at the exit of a node
      (i.e. using \afunout{n}{\cm,\rho} and \dfunout{n}{\rho}).
\item Proving the basis of the induction.
\item Proving the inductive step. This step uses the case analysis illustrated in Figure~\ref{alpha-beta-l-r-soundness}.
\end{enumerate}

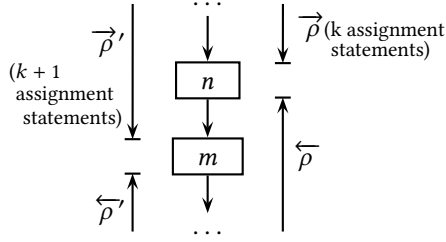
\begin{figure}[t]
 
 \begin{center}
 \begin{pspicture}(0,0)(65,40)

\putnode{a0}{origin}{29}{35}{\psframebox[linestyle=none]{${\white\ }\ldots{\white\ }$}}
\putnode{a1}{a0}{0}{-10}{\psframebox{${\white A}n{\white B}$}}
\putnode{a2}{a1}{0}{-10}{\psframebox{${\white A}m{\white B}$}}
\putnode{a3}{a2}{0}{-10}{\psframebox[linestyle=none]{${\white\ }\ldots{\white\ }$}}

\putnode{b0}{a0}{-10}{3}{\psframebox[linestyle=none]{$\white\alpha, \beta$}}
\putnode{b1}{a2}{-10}{0}{\psframebox[linestyle=none]{$\white\alpha$}}
\putnode{b2}{a3}{-10}{0}{\psframebox[linestyle=none]{$\white\alpha$}}

\putnode{c0}{a1}{-13}{5}{\psframebox[linestyle=none]{$\overrightarrow{\rho}'$}}
\putnode{c00}{c0}{-10}{-4}{\psframebox[linestyle=none]{\footnotesize{($k+1$}}}
\putnode{c00}{c0}{-6}{-7}{\psframebox[linestyle=none]{\footnotesize{assignment}}}
\putnode{c00}{c0}{-5}{-10}{\psframebox[linestyle=none]{\footnotesize{statements)}}}

\putnode{c1}{a3}{-13}{3}{\psframebox[linestyle=none]{$\overleftarrow{\rho}'$}}

\putnode{d0}{a0}{10}{3}{\psframebox[linestyle=none]{$\white\alpha, \beta$}}
\putnode{d1}{a1}{10}{0}{\psframebox[linestyle=none]{$\white\alpha$}}
\putnode{d2}{a3}{10}{0}{\psframebox[linestyle=none]{$\white\alpha$}}

\putnode{e0}{a0}{22}{-3}{\psframebox[linestyle=none]{$\overrightarrow{\rho} \text{\footnotesize{(k assignment}}$}}
\putnode{e00}{e0}{1}{-3}{\psframebox[linestyle=none]{\footnotesize{statements)}}}
\putnode{e1}{a2}{13}{0}{\psframebox[linestyle=none]{$\overleftarrow{\rho}$}}

\ncline{->}{a0}{a1}
\ncline{->}{a1}{a2}
\ncline[nodesepB=1]{->}{a2}{a3}

\ncline{->|}{b0}{b1}
\ncline[nodesepA=-2]{->|}{b2}{b1}
\ncline{->|}{d0}{d1}
\ncline[nodesepA=-2]{->|}{d2}{d1}

 \end{pspicture}
\end{center}

\caption{Illustrating the qualified control flow paths used for induction in Lemma~\ref{lemma-soundness-flow-function-proof}}
\label{fig:describe:paths:lemma}
\end{figure}

Consider a qualified control flow path
\text{$\rho = (\overrightarrow{\rho}, n, \overleftarrow{\rho})$} such that
$\overrightarrow{\rho}$ contains $k$ assignment statements as illustrated in Figure~\ref{fig:describe:paths:lemma}.
Statement $n$ is the $k+1^{th}$ assignment statement in $\rho$.
Consider $\rho'$ obtained by extending \text{$\overrightarrow{\rho}$} to include $n$
and remove $m$ from \text{$\overleftarrow{\rho}$}. 
Thus both $\rho$ and $\rho'$ have the same sequence of statements except that the
forward and backward components are different.
 Formally,
\text{$\rho' = (\overrightarrow{\rho'}, m, \overleftarrow{\rho'})$} such that
\text{$\Succ(\rho,n)$} is $m$. Thus,
\text{$\overleftarrow{\rho} = (m, \overleftarrow{\rho'})$} and
\text{$\overrightarrow{\rho'} = (\overrightarrow{\rho}, n)$} and
\text{$\overrightarrow{\rho'}$} contains $k+1$ assignment statements.
When we refer to node $n$, we use the qualified control flow path $\rho$; when we refer to node $m$, 
we use the path $\rho'$.
For $\rho'$,
\afunin{m}{\rho'} is same as \afunout{n}{\rho}
where $n$ is a predecessor of $m$ as described in 
Figure~\ref{fig:describe:paths:lemma}.
Computation of \ain along a qualified control flow path
is described below.
\begin{align*}
\forall m\in \nodes, \forall \rho'\in \paths{m} : \;
\afunin{m}{\rho'} & =  
		\begin{cases}
		\boundary & m \text{ is } \Start{p}
		\\
		\afunout{n}{\rho} & n = \Pred(\rho', m)
		\end{cases}
\end{align*}

Since \ain is nothing but the \aout at predecessor along a path,
it thus suffices to prove soundness with respect to only $\aout$ as described below.
\[\forall n \in \nodes, \forall \rho \in \paths{n}:\;\; 
(\alphacm, \betacm) \in \afunout{n}{\cm,\rho} \wedge \balpha \subseteq \dfunout{n}{\rho}
\Rightarrow \balpha \times \bbeta \subseteq \afunout{n}{\rho} \]

Note that \afunout{n}{\rho} by definition is computed 
by restricting to \dfunout{n}{\rho}. Thus such a restriction 
on the right-hand side in the above equation is not necessary and thus not 
mentioned.

 We prove the lemma by induction on the number of assignment statements in the forward component $\overrightarrow{\rho}$ of $\rho$.
 The basis of the induction contains 0 assignment statement in the forward component of the qualified control flow path.
 Thus the statement of interest is the first assignment statement and it must be of the form \text{$x = \rhs$} where \rhs{} is either 
 $ \new\ \tau$ or $\&a$ where $a$ is an object of class $\tau$.
For both the cases we have \text{$(x, \rhscm) \in \afunout{n}{\cm,\rho}$}.
For the former case, we have \text{$(x, \&l) \in \afunout{n}{\cm,\rho}$} where $l$ is the location of the allocated object;
for the latter case, we have \text{$(x, \&a) \in \afunout{n}{\cm,\rho}$}.
From Equation~(\ref{eq:agen}),
if $x \in \dfunout{n}{\rho}$
we have \text{$(x, \brhs\,) \in \afunout{n}{\rho}$} where \brhs is the abstract name of \rhs. By definition, \brhs is an
over-approximation of \rhs{} under type-based abstraction (or allocation-site-based abstraction as defined in abstract name) and hence the basis follows.

 For the inductive hypothesis, assume that the lemma holds for a path
$\rho$ reaching node $n$.
Path $\rho'$ is obtained by extending \text{$\overrightarrow{\rho}$} to include $n$.
We prove the inductive step for statement $m$ in the qualified control flow path $\rho'$, which is of the form $\lhs = \rhs$.
Consider \text{$(\alphacm, \betacm) \in \afunout{m}{\cm,\rho'} \wedge \balpha \subseteq \dfunout{m}{\rho'}$}.
Now access expressions $\alpha$ and $\beta$ can take the combinations with respect to \lhs and \rhs{} as described in Figure~\ref{alpha-beta-l-r-soundness}.
We discuss these cases below.

\begin{figure}[t]
 
 \begin{center}
 \begin{pspicture}(0,0)(120,55)

\putnode{a0}{origin}{53}{50}{\psframebox{$\alpha, \beta$}}

\putnode{a2}{a0}{-38}{-10}{\psframebox{$\alpha = r$}}
\putnode{g1}{a2}{19}{-10}{\psframebox[linestyle=none,fillstyle=solid, fillcolor=lpink]{\rule{30mm}{0mm}\rule{0mm}{7mm}}}
\putnode{c1}{a2}{9}{-10}{\psframebox{$\beta = l$}}
\putnode{c2}{a2}{-9}{-10}{\psframebox{$\beta \neq l$}}

\putnode{a1}{a0}{0}{-10}{\psframebox{$\alpha = l$}}
\putnode[l]{g2}{a1}{3}{-10}{\psframebox[linestyle=none,fillstyle=solid, fillcolor=lightblue]{\rule{24.5mm}{0mm}\rule{0mm}{7mm}}}
\putnode{b1}{a1}{-9}{-10}{\psframebox{$\beta = r$}}
\putnode{b2}{a1}{9}{-10}{\psframebox{$\beta \neq r$}}

\putnode{a3}{a0}{38}{-10}{\psframebox{$\alpha \neq l, \alpha \neq r$}}
\putnode{d1}{a3}{-14}{-10}{\psframebox{$\beta = l$}}
\putnode{d2}{a3}{0}{-10}{\psframebox{$\beta = r$}}
\putnode{d3}{a3}{18}{-10}{\psframebox{$\beta \neq l, \beta \neq r$}}


\putnode{e1}{g1}{0}{-9}{\psframebox[linestyle=none]{
\begin{tabular}{l}
 Generation of \\ direct alias 
\end{tabular}
}}

\putnode{e2}{g2}{13}{-9}{\psframebox[linestyle=none]{
\begin{tabular}{l}
 Generation of \\derived alias 
\end{tabular}
}}

\putnode{e3}{d1}{-15}{-23}{\psframebox[linestyle=none]{
\begin{tabular}{l}
No generation of alias.  \\ Only  propagation of  alias. 
\end{tabular}
}}

\ncline{->}{a0}{a1}
\ncline{->}{a0}{a2}
\ncline{->}{a0}{a3}

\ncline{->}{a1}{b1}
\ncline{->}{a1}{b2}

\ncline{->}{a2}{c1}
\ncline{->}{a2}{c2}

\ncline{->}{a3}{d1}
\ncline{->}{a3}{d2}
\ncline{->}{a3}{d3}


\nccurve[angleA=145,angleB=270,linestyle=dashed,dash=.6 .6,offsetA=10,nodesepA=5]{->}{e3}{c2}
\nccurve[angleA=75,angleB=270,linestyle=dashed,dash=.6 .6,nodesepA=1,offsetA=-8]{->}{e3}{d2}
\nccurve[angleA=45,angleB=270,linestyle=dashed,dash=.6 .6,offsetA=-10,nodesepA=4]{->}{e3}{d3}
 \end{pspicture}
\end{center}

{\small
\begin{itemize}
\item Cases $(\alpha = \rhs, \beta = \lhs)$ and $(\alpha = \lhs, \beta = \rhs)$
      correspond to aliases $(\blhs\times\brhs)$ reaching the exit of node $m$. These
     aliases are generated by $m$  directly without using the aliases
      reaching the entry of $m$.
\item Cases $(\alpha = \lhs, \beta \neq \rhs)$ and $(\alpha \neq \lhs, \alpha \neq \rhs, \beta = \lhs$) 
      correspond to aliases $(\blhs\times\bbeta)$ or $(\blhs\times \balpha)$ reaching the
      exit of node $m$. They are derived using the aliases reaching the entry of $m$.
\item The remaining cases correspond to the aliases that are propagated from the entry of $m$ to  the exit of $m$ 
      without any change. They are 
 \text{$(\brhs\times\bbeta)$}, \text{$(\brhs\times \balpha)$} or \text{$(\balpha\times\bbeta)$}.

\end{itemize}
}

\caption{Given the aliases 
\text{$(\alphacm, \betacm)$} reaching the exit of node $m$ in the concrete memory, the picture shows various
cases of the aliases reaching the exit of node $m$ in the abstract memory in terms of 
$\alpha$ and $\beta$ (Lemma~\ref{lemma-soundness-flow-function-proof}).}
\label{alpha-beta-l-r-soundness}
\end{figure}

\begin{itemize}
 
 \item Generation of direct alias :
 \begin{itemize}
 \item $\alpha = \lhs, \beta = \rhs$
 
 Since
$\balpha \subseteq \dfunout{m}{\rho'}$ demand for \rhs{} will be raised appropriately. 
 From  Equations~(\ref{eq:dgen}) and~(\ref{eq:ldgen}), 
 if $\rhs \in \{*x, x\rightarrow f\}$ then $\{x, \&x\} \subseteq \dfunout{n}{\rho}$.
 If $\rhs \not\in \{*x, x\rightarrow f\}$, \brhs can be computed directly. Otherwise,
 the inductive hypothesis ensures that we have the alias information of $x$ and hence we can compute \brhs.
 From Equation~(\ref{eq:agen}), $\balpha \subseteq \dfunout{m}{\rho'} \Rightarrow 
\balpha \times \bbeta \subseteq \afunout{m}{\rho'}$.

  \item $\alpha = \rhs, \beta = \lhs$
 
This is a dual of the previous case with \rhs{} and \lhs interchanged. It holds because aliasing is symmetric.
 \end{itemize}

 \item Generation of derived alias :
 \begin{itemize}
\item $\alpha = \lhs, \beta \neq \rhs$ 

 Since
\text{$\balpha \subseteq \dfunout{m}{\rho'}$} appropriate demand for \rhs{} will be raised. 
Also, since \text{$\alpha = \lhs, \beta \neq \rhs$} and
\text{$(\alphacm, \betacm) \in \afunout{m}{\cm,\rho'}$},
such an alias can be derived from the alias
\text{$(\rhscm, \betacm) \in \afunout{n}{\cm,\rho}$}.
 Based on the inductive hypothesis, we have the alias information \text{$\brhs \times \bbeta \subseteq \afunout{n}{\rho}$} and
 we compute derived alias between \balpha (which is \blhs, by raising appropriate demand for \lhs) and \bbeta at node $m$.
 Thus, \text{$\balpha \subseteq \dfunout{m}{\rho'} \Rightarrow \balpha \times \bbeta \subseteq \afunout{m}{\rho'}$}.

\item $\alpha \neq \lhs, \alpha \neq \rhs, \beta = \lhs$.  

 Since
\text{$\balpha \subseteq \dfunout{m}{\rho'}$} and 
\text{$\alpha \neq \lhs, \alpha \neq \rhs
\Rightarrow
\balpha \subseteq \dfunout{n}{\rho}$}.
Also, since 
\text{$(\alphacm, \betacm) \in \afunout{m}{\cm,\rho'}$} and
\text{$\alpha \neq \lhs, \alpha \neq \rhs, \beta = \lhs$},
such an alias can be derived from the alias
\text{$(\alphacm, \rhscm) \in \afunout{n}{\cm,\rho}$}.
This is a dual of the previous case with $\alpha$ and $\beta$ interchanged. 
\end{itemize}

 \item Propagation of alias from predecessor : 
 \begin{inparaenum}[\em (a)]
 \item $\alpha = \rhs, \beta \neq \lhs$, 
 \item $\alpha \neq \lhs, \alpha \neq \rhs, \beta = \rhs$, and 
 \item $\alpha \neq \lhs, \alpha \neq \rhs, \beta \neq \lhs, \beta \neq \rhs$ 
 \end{inparaenum}

 From  Equations~(\ref{eq:din}) and~(\ref{eq:dout}), 
 \text{$\balpha \subseteq \dfunout{m}{\rho'} \Rightarrow \balpha \subseteq \dfunout{n}{\rho}$}.
 Since we have \text{$(\alphacm, \betacm) \in \afunout{m}{\cm,\rho'}$}, 
such an alias
 relationship is not computed by statement $m$, instead it is propagated from the predecessors of $m$ unmodified.
 Thus \text{$(\alphacm, \betacm) \in \afunout{n}{\cm,\rho}$}.
 From the inductive hypothesis, 
\text{$(\alphacm, \betacm) \in \afunout{n}{\cm,\rho} \wedge \balpha \subseteq \dfunout{n}{\rho}
 \Rightarrow 
\balpha \times \bbeta \subseteq \afunout{n}{\rho}$}.
 This alias relationship will be propagated to statement $m$.
 From  Equations~(\ref{eq:ain}) and~(\ref{eq:aout}), 
 \text{$\balpha \subseteq \dfunout{m}{\rho'} \Rightarrow 
\balpha \times \bbeta \subseteq \afunout{m}{\rho'}$}.

\end{itemize}

    Thus the lemma holds.
 \end{proof}

\begin{lemma}
  Let a program \prog contain the virtual function call $x \rightarrow {\text{\em vfun}}()$ at
  node $\vcnode \in \origin$.  Then $x$ is in $\din_{\vcnode}^{\rho}$ for all qualified control flow paths
  $\rho \in \paths{\vcnode}$.
 \label{lemma-initial-demand-sound}
\end{lemma}

\begin{proof}
  Let $\rho$ be an arbitrary qualified control flow path in
  $\paths{\vcnode}$.  The demand corresponding to the virtual call
  statement in node $\vcnode$ is computed using \ldgen
  (Equation~\ref{eq:dgen}).  Specifically, the third case
  ($\vname(x)\neq \emptyset$) in Equation~\ref{eq:dgen} applies and a
  demand for the abstract name of $x$, i.e. $x$ itself
  (Table~\ref{table-dn}), is generated.  Thus a demand for $x$ is
  generated by Equations~\ref{eq:din},~\ref{eq:dgen}
  and~\ref{eq:ldgen}.  Hence $x \in \din_{\vcnode}^{\rho}$.
\end{proof}

\section{Precision Proof}\label{sec:precision-proof}
In this section, we compare the precision of alias information
computed by algorithm $\id$ with that of two other algorithms.
These algorithms are  
\begin{enumerate}[(i)]
\item a demand-driven alias analysis using
conventional speculation, referred to earlier in
Section~\ref{sec:speculation} and henceforth denoted as $\cd$, and
\item an exhaustive algorithm, referred to in Section~\ref{sec:intro}
and henceforth denoted as $\ex$.  
\end{enumerate}
Specifically, we show that $\id$ is
at least as precise as, and sometimes more precise than
$\ex$, while $\cd$ and $\ex$ enjoy exactly the same degree of
precision.
To formalize this notion, we introduce some terminologies.
\newcommand{\alga}{\text{\sc aa1}\xspace}
\newcommand{\algb}{\text{\sc aa2}\xspace}

\begin{definition}
An alias analysis algorithm \alga is \emph{at least as precise} as
another alias analysis algorithm \algb for a program $\prog$ with
respect to a set of demands $\dfun{n}{}$ at node $n$ iff
\[\restrict {\afun{n}{\algb}}{\dfun{n}{}} \sqsubseteq \restrict {\afun{n}{\alga}}{\dfun{n}{}} 
.\]
\end{definition}

As explained in the soundness proofs,
since the end goal of our analysis is to compute
the set of all incoming aliases of $x$ at node $\vcnode \in \origin$,
where 
the virtual function call at
\vcnode is of the form
$x\rightarrow \text{\em vfun}()$,
we compare the overall precision of two alias analysis
algorithms as follows.

\begin{definition}
(Precision comparison of alias analysis algorithms \alga and \algb across all programs for virtual call resolution.)
  \begin{enumerate}
  \item Algorithm \alga is \emph{at least as precise} as algorithm \algb
    iff for every program $\prog$ and every node $\vcnode \in \origin$
    with a virtual function call of the form 
\text{$x \rightarrow {\text{\em vfun}}()$}
 in $\prog$, $\alga$ is at least as precise as $\algb$
    for $\prog$ with respect to 
\text{$\dfun{\vcnode}{} = (\{x\},\emptyset)$}.
  \item Algorithms $\alga$ and $\algb$ are \emph{equiprecise} iff
    $\alga$ is at least as precise as $\algb$ and vice versa.
  \item Algorithm $\alga$ is \emph{more precise} than algorithm
    $\algb$ iff $\alga$ is at least as precise as $\algb$ but 
$\alga$ and $\algb$ are not equiprecise.
  \end{enumerate}
\end{definition}

Note that if $\alga$ is more precise than $\algb$, then it
does not imply that the alias information computed by $\algb$ is a
strict over-approximation of that computed by $\algb$ for all programs.
However, there is at least one program $\prog$ and one 
\text{$\vcnode \in \origin$} where this holds.

 \begin{thm}
   Algorithm \id is more precise than algorithm \ex, while algorithms \cd and \ex are equiprecise.
  \label{thm:precision}
  \end{thm}
\begin{proof}
To show that \id is more precise than \ex, we first show in
Lemma~\ref{lemma-id-precise} that \id is at least as precise as \ex.
Subsequently, we show by means of an example in
Lemma~\ref{lemma-id-precise-eg} that 
there exists a case when \id is more precise than \ex.

The proof of equiprecision of \cd and \ex follows from Lemmas~\ref{lemma-cd-precise} and ~\ref{lemma-cd-ex-equiv}, which show that \cd is at least as precise
as \ex and vice versa.
\end{proof}

        Section~\ref{sec:java.precision}
       modifies algorithm \id to define algorithm \jd for demand-driven alias analysis of Java programs
       by eliminating the speculation of demands required for C/C++. 
      Then the section shows that unlike \cd and \ex which are equiprecise, \jd is more precise than \ex.

\subsection{Algorithm \id is More Precise than Algorithm \ex}
The following two lemmas establish the desired result.
\begin{lemma}
Algorithm \id is at least as precise as algorithm \ex.
  \label{lemma-id-precise}
 \end{lemma}

\begin{proof}
  Consider an arbitrary program $\prog$ and 
  let $\aid$ and $\aex$ represent the vectors of alias information
  at all nodes in $\prog$, as computed by algorithms \id and \ex
  respectively.  
  We prove the
  lemma by establishing a significantly stronger result: $\forall n \in \nodes$,
  $\afun{n}{\id} \sqsupseteq \afun{n}{\ex}$  
  instead of showing the minimal requirement:
  $\forall \vcnode \in \origin$, $\restrict{\afun{\vcnode}{\id}}{\dfun{\vcnode}{\id}} \sqsupseteq \restrict{\afun{\vcnode}{\ex}}{\dfun{\vcnode}{\id}}$.

  Let Function $\vecfid = (\veca^{\id}, \vecd^{\id})$ denote the vector of
  data flow functions ($\veca^{\id}$ for aliases and $\vecd^{\id}$ for
  demands) corresponding to algorithm \id.  
  The initial vector of data flow values used in algorithm \id is
  $\val{}^{\id} = (\aval^{\top}, \dval^{\top})$, where $\aval^{\top}$
  (resp. $\dval^{\top}$) denotes the vector of aliases (resp. demands)
  at all nodes set to $\top$. The alias and demand information
  eventually computed by \id, denoted $(\aval^{\id}, \dval^{\id})$, is
  the fixed-point of $\vecfid$ starting from $\val{}^{\id}$.  We
  denote the alias and demand information computed by $\id$ in the
  $i^{th}$ fixed-point iteration by $\aval^{\id,i}$ and
  $\dval^{\id,i}$ respectively.
  
  Algorithm \ex, on the other hand, computes all possible aliases at
  each node, as if the demand raised at each node was the set of all
  possible access expressions.  Therefore, an alternate view of
  algorithm \ex is that it is the same as \vecfid but with demands at
  all nodes set to $\bot$ at each step.  Formally, let
  $\vecgex(\cdot)$ denote the data flow function corresponding to
  algorithm \ex.  Then, for every alias vector $\aval$,
  $\vecgex(\aval) = \veca^{\id}(\aval, \dval^{\bot})$, where
  $\dval^{\bot} = \langle \bot, \ldots \bot\rangle$, i.e. demand at
  every node is the set of all possible access expressions.  The alias
  information eventually computed by \ex, denoted $\aval^{\ex}$, is
  the fixed-point of $\vecgex$, starting from the initial vector of
  aliases $\val{}^{\ex} = \aval^{\top}$.  Let $\aval^{\ex,i}$ denote
  the alias information computed by $\ex$ in the $i^{th}$ fixed point
  iteration.  
We use induction on the fixed point iterations \text{$i\geq 0$} to show that
\text{$\aval^{\ex,i} \sqsubseteq \aval^{\id,i}$}.

Consider \text{$i=0$} to be the base case.
Then,
\text{$\aval^{\ex,0} = \val{}^{\ex} = \aval^{\top} = \aval^{\id,0}$}.
Thus the base case trivially holds.
For the inductive hypothesis, assume that the lemma holds for some 
\text{$k>0$}. Then,
\begin{align}
\aval^{\ex,k} \sqsubseteq \aval^{\id,k}
\label{proof.step.a}
\end{align}
For the inductive step, we show that the lemma holds for $k+1$.

\[
\setlength{\arraycolsep}{2pt}
  \begin{array}{lcrll}
& & \veca^{\id} (\aval^{\ex,k}, \dval^{\bot}) &\sqsubseteq \veca^{\id} (\aval^{\id,k}, \dval^{\id,k}) &
 \text{(applying $\veca^{\id}$ to (\ref{proof.step.a}), $\dval^{\bot} \sqsubseteq \dval^{\id,k},$}
\\ & & & & \text{$\;\;\veca^{\id}$ is monotonic)}
\\
& \Rightarrow 
& \vecgex (\aval^{\ex,k}) = \veca^{\id}(\aval^{\ex,k}, \dval^{\bot}) 
&\sqsubseteq \veca^{\id} (\aval^{\id,k}, \dval^{\id,k}) 
&
   \\
& \Rightarrow 
& \aval^{\ex,k+1} 
&\sqsubseteq \aval^{\id,k+1} 
&
\end{array}
\]

Thus it follows that
   $ \forall n \in \nodes : \afun{n}{\ex} \sqsubseteq \afun{n}{\id}$.
 \end{proof}

\begin{corollary}
Let $\vecfdd = (\veca^{\dd}, \vecd^{\dd})$ be a vector of monotonic
flow functions that computes aliases and demands in an alias
analysis algorithm $\dd$.  Let $(\add, \ddd)$ denote the fixed-point
of $\vecfdd$ obtained by starting from the initial vector
$\val{}^{\dd} = (\aval^{\top}, \dval^{\top})$.  If $\vecgex(\aval) =
\veca^{\dd}(\aval, \dval^{\bot})$ for all alias vectors $\aval$, then $\aex \sqsubseteq \add$.
\label{corr:dd-ex}
\end{corollary}

\begin{proof}
Follows from lemma~\ref{lemma-id-precise} where the only requirement
is that the flow functions are monotonic, and that $\vecgex(\aval) =
\veca^{\dd}(\aval, \dval^{\bot})$.
\end{proof}

\begin{lemma}
Algorithm $\id$ and 
algorithm $\ex$ are not equiprecise.
  \label{lemma-id-precise-eg}
 \end{lemma}
 \begin{proof}

   We prove the lemma by providing an example to show that
algorithm \id is strictly more precise than algorithm \ex.
Consider the
   example discussed in Figure~\ref{m-eg-type}.  We refer to the alias
   and demand computed at line 28, where the instruction is 
\text{$t\rightarrow {\text{\em vfun}}()$}.  We consider the demand
   \text{$\dfun{28}{\id} = (\{t\}, \emptyset)$}, 
and note that the alias pair
   \text{$(t, \&Z)\in \ain^{\ex}_{28}$}. 
However, \text{$(t, \&Z)\not\in \ain^{\id}_{28}$}.  
Thus, algorithm \id and algorithm \ex are not equiprecise.
 \end{proof}

\subsection{Algorithms \cd and \ex are Equiprecise}
\label{sec:compare:cd:ex}
The existing demand-driven methods~\cite{Heintze:2001:DPA:378795.378802,lfcpa}
are defined for C and do not deal with objects on the heap.
We extend their demand-driven strategy to include objects on the heap and call it as
the conventional demand-driven method (\cd). With this extension we can use it for static
resolution of virtual function calls in C++.

 We define \vecfcd to compute alias and demand using the conventional
 demand-driven method. Note that $\vecfcd$ can be obtained from
 \vecfid with minor changes.  In the conventional demand-driven
 method, we do not generate any demand of the form of an address-of a
 variable.  Also, if \lhs is of the form $*x$ of $x\rightarrow f$,
 then we generate demand for \base(\lhs) irrespective of the demand
 being raised.  With these modifications to our data flow equations
 presented in Section~\ref{sec:formulation}, we get the data flow
 equations that compute demands and aliases using the conventional
 speculation method. 

 \begin{lemma}
   Algorithm \cd is at least as precise as algorithm \ex.
\label{lemma-cd-precise}
 \end{lemma}

 \begin{proof}
Follows from corollary~\ref{corr:dd-ex} of
lemma~\ref{lemma-id-precise} since \vecfcd is a monotonic flow
function, and $\vecgex(\aval) = \vecfcd{(\aval, \dval^{\bot})}$.

 \end{proof}

\begin{lemma}
  Algorithm \ex is at least as precise as algorithm \cd.
 \label{lemma-cd-ex-equiv}
\end{lemma}

\begin{proof}

  To prove this lemma, we establish a stronger result.  Specifically,
  we show that for a demand computed using \cd at a node $n$, if there
  exists an alias pair computed by \ex then such an alias pair will
  also be computed by \cd.  Formally, if \dcd is the vector of
  demands computed by algorithm $\cd$ after 
the convergence of the analysis on program
  $\prog$, then 
\[\forall n \in \nodes,
  \restrict{\afun{n}{\ex}}{\dfun{n}{\cd}} \sqsupseteq
  \restrict{\afun{n}{\cd}}{\dfun{n}{\cd}}.\]  
For convenience,
  we re-write this proof obligation in the following equivalent form.
 \[\forall m>0, \forall n\in \nodes, \forall \alpha, \forall \beta, \balpha\subseteq\dout^{\cd}_n \wedge 
 \balpha\times \bbeta\subseteq \aout^{\ex, m}_n \Rightarrow \balpha\times \bbeta\subseteq \aout^{\cd}_n
 \]
  Here, $\aval^{\ex}$ represents the vector of aliases obtained by
  performing fixed point computation of function $\vecfcd
  (\aval^{\ex}, \bot)$.  An entry in the vector denoted $\aout^{\ex,
    m}_n$ represents an alias pair that appeared for the first time in
  the set of aliases computed by algorithm \ex for node $n$ after $m$
  steps (iterations of fixed-point computation). Thus,
  \[\aout^{\ex}_n = \bigcup_{m=1}^{\infty} \aout^{\ex,m}_n\]
  
 Our proof is by induction on the number of steps $m$ it takes for an
 alias pair to appear at a node $n$ in the program.  Note that this
 number may not have any correlation with the number of instructions
 in a path from the start of the program to the node $n$.  This is
 especially true in the setting of bidirectional data flow analysis, as
 is the case here.  We prove for an arbitrarily chosen node $n$ and an
 arbitrary alias pair at $n$; hence the argument holds for all alias
 pairs at all nodes in the program.
 
 The base case for our induction is $m=1$. It corresponds to the 
situation when an
 alias pair appears in a single step. Since the alias pair appears in 
a single
 step, it means that the alias pair is added because of the statement
 in node $n$ and it could not have been added due to aliases from
 predecessor.  Thus statement in node $n$ ought to be an assignment
 statement.  Consider the statement $n$ to be of the form $\lhs =
 \rhs$.  In order for statement $n$ to generate alias of the form
 $\blhs\times\brhs${} 
without using alias information reaching $n$,
  $\lhs \in \{x, x.f\}$ and $\rhs \in \{y, \&a,
 y.f, \new\ \tau\}$. For other cases of \lhs and \text{\rhs,} abstract names
 \blhs and \brhs cannot be computed without the aliases from the
predecessors. Thus such an alias pair cannot
 appear in a single step.  Since $\balpha\subseteq\dout^{\cd}_n$ and
 $\blhs\times\brhs \subseteq \aout^{\ex,1}_n$ then $\alpha = \lhs$ or
 $\alpha = \rhs$.  In both the cases, $\blhs\times \brhs \subseteq
 \aout_n^{\cd}$.  Thus the lemma holds for the base case.
 
 For the inductive hypothesis, assume that the lemma holds for the aliases 
that appear in $k$ steps.
 We prove the inductive step by arguing the lemma for the aliases 
that appear in $k+1$ step.
Consider the case where for demand $\balpha\subseteq\dout^{\cd}_n$
the alias $\balpha\times \bbeta\subseteq \aout^{\ex, k+1}_n$ takes $k+1$ steps to appear 
at node $n$.
Statement $n$ could be of various kinds like an assignment statement, conditional statement, nop, print statement.
We classify different kinds of statements in two categories where
\begin{itemize}
 \item statement $n$ is not an assignment statement or
 \item statement $n$ is an assignment statement.
\end{itemize}
We explain these two categories below.

{\em Statement $n$ is not an assignment statement.} Such a statement does not generate an alias pair.
Thus, demand for $\balpha$ will be raised to the predecessors of $n$. 
If it takes $k+1$ steps for the alias to reach node $n$ then it certainly takes $\leq k$ steps for it to reach the 
predecessors of node $n$.
From the inductive hypothesis, $\forall p \in \pred(n), \balpha\subseteq\dout^{\cd}_p, 
\balpha\times\bbeta\subseteq\aout^{\ex, k}_p \Rightarrow 
\balpha\times\bbeta\subseteq\aout^{\cd}_p$.
If the alias belongs to a predecessor and statement $n$ does not update the alias, then the alias 
is propagated to statement $n$.
Thus $\balpha\times\bbeta\subseteq\aout^{\cd}_n$.

{\em Statement $n$ is an assignment statement.}
Demand $\balpha\subseteq\dout^{\cd}_n$ and 
the alias $\balpha\times \bbeta\subseteq \aout^{\ex, k+1}$.
Now access expressions $\alpha$ and $\beta$ could take the combinations with respect to \lhs and \rhs{} as 
described in Figure~\ref{alpha-beta-l-r-precision}.
 We discuss these cases below.

 \begin{figure}[t]
 
 \begin{center}
 \begin{pspicture}(0,0)(120,55)

\putnode{a0}{origin}{50}{50}{\psframebox{$\alpha, \beta$}}

\putnode{a2}{a0}{-38}{-10}{\psframebox{$\alpha = r$}}
\putnode{g1}{a2}{19}{-10}{\psframebox[linestyle=none,fillstyle=solid, fillcolor=lpink]{\rule{30mm}{0mm}\rule{0mm}{7mm}}}
\putnode{c1}{a2}{9}{-10}{\psframebox{$\beta = l$}}
\putnode{c2}{a2}{-9}{-10}{\psframebox{$\beta \neq l$}}

\putnode{a1}{a0}{0}{-10}{\psframebox{$\alpha = l$}}
\putnode{b1}{a1}{-9}{-10}{\psframebox{$\beta = r$}}
\putnode{b2}{a1}{9}{-10}{\psframebox{$\beta \neq r$}}

\putnode{a3}{a0}{38}{-10}{\psframebox{$\alpha \neq l, \alpha \neq r$}}
\putnode{d1}{a3}{-14}{-10}{\psframebox{$\beta = l$}}
\putnode{d2}{a3}{0}{-10}{\psframebox{$\beta = r$}}
\putnode[l]{g2}{a3}{7}{-10}{\psframebox[linestyle=none,fillstyle=solid, fillcolor=lightblue]{\rule{20mm}{0mm}\rule{0mm}{7mm}}}
\putnode{d3}{a3}{18}{-10}{\psframebox{$\beta \neq l, \beta \neq r$}}


\putnode{e1}{g1}{0}{-9}{\psframebox[linestyle=none]{
\begin{tabular}{l}
 Generation of \\ direct alias 
\end{tabular}
}}

\putnode{e2}{g2}{14}{-9}{\psframebox[linestyle=none]{
\begin{tabular}{l}
No generation of alias.  \\ Only  propagation of  alias. 
\end{tabular}
}}

\putnode{e3}{d1}{-15}{-23}{\psframebox[linestyle=none]{
\begin{tabular}{l}
 Generation of derived alias \\ or propagation of alias.
\end{tabular}
}}

\ncline{->}{a0}{a1}
\ncline{->}{a0}{a2}
\ncline{->}{a0}{a3}

\ncline{->}{a1}{b1}
\ncline{->}{a1}{b2}

\ncline{->}{a2}{c1}
\ncline{->}{a2}{c2}

\ncline{->}{a3}{d1}
\ncline{->}{a3}{d2}
\ncline{->}{a3}{d3}


\nccurve[angleA=145,angleB=270,linestyle=dashed,dash=.6 .6,offsetA=10,nodesepA=5]{->}{e3}{c2}
\ncline[linestyle=dashed, dash=.6 .6]{->}{e3}{b2}
\ncline[linestyle=dashed, dash=.6 .6, offsetA=-8, nodesepA=4.2]{->}{e3}{d1}
\nccurve[angleA=45,angleB=270,linestyle=dashed,dash=.6 .6,offsetA=-12,nodesepA=6]{->}{e3}{d2}
 \end{pspicture}

{\small
\begin{itemize}
\item Cases $(\alpha = \rhs, \beta = \lhs)$ and $(\alpha = \lhs, \beta = \rhs)$
      correspond to aliases $\blhs\times\brhs$ reaching the exit of node $n$. These
     aliases are generated by $n$  directly without using the aliases
      reaching the entry of $n$.
\item Case $(\alpha \neq \lhs, \alpha \neq \rhs, 
\beta \neq \lhs, \beta \neq \rhs)$
 correspond to aliases
\text{$\balpha\times\bbeta$}
 that are propagated from the entry of $n$ to  the exit of $n$ 
      without any change. 
\item The remaining cases
      correspond to aliases that are either derived using the 
aliases reaching the entry of $n$ or are 
propagated from the entry of $n$ to the exit of $n$.
They are $\brhs\times\bbeta$, $\blhs\times\bbeta$, 
$\blhs\times \balpha$ and $\brhs\times\balpha$.
 
\end{itemize}
}

\caption{Given the aliases 
\text{$\balpha\times \bbeta$} reaching the exit of node $n$ in the
\ex method, the picture shows various
cases of the aliases reaching the exit of node $n$ in the 
\cd method (Lemma~\ref{lemma-cd-ex-equiv}).}
\label{alpha-beta-l-r-precision}
\end{center}
\end{figure}

 \begin{itemize}
 \item Generation of direct alias :
 \begin{itemize}
  \item $\alpha = \lhs, \beta = \rhs$
  
  If $\lhs \in \{x, x.f\}$ and $\rhs \in \{y, \&a, y.f, \new\ \tau\}$ then the alias would have been computed in single step.
  
  Let us look at the cases for other values of \lhs and \rhs.
  Since, $\balpha\subseteq\dout^{\cd}_n$, appropriate demand for \rhs{} will be raised at its predecessors.
If it takes $k+1$ steps for the alias $\balpha\times \bbeta\subseteq \aout^{\ex, k+1}_n$ 
to reach the exit of node $n$ then it certainly takes $\leq k$ steps for
the alias related to \rhs{} to reach the exit of the
predecessors of node $n$.
From the inductive hypothesis, the aliases of \rhs{} is present in both the methods \ex and \cd  at the exit of the predecessors of $n$.
Statement $n$ uses the alias for \rhs{} to compute \brhs.
Thus, $\balpha\subseteq\dout^{\cd}_n \Rightarrow \balpha\times\bbeta\subseteq\aout^{\cd}_n$.

  \item $\alpha = \rhs, \beta = \lhs$

  This is a dual of the previous case with \rhs{} and \lhs interchanged. It holds because aliasing is symmetric.
%
 \end{itemize}

\item Generation of derived alias or propagation of alias from predecessor:
\begin{itemize}

  \item $\alpha = \lhs, \beta \neq \rhs$
  
  Since, $\balpha\subseteq\dout^{\cd}_n$, appropriate demand for \lhs (if $\lhs \not\equiv x$) and \rhs{} will be raised to the predecessors of $n$.
If it takes $k+1$ steps for the alias $\balpha\times \bbeta\subseteq \aout^{\ex, k+1}_n$ 
to reach the exit of node $n$ then it certainly takes $\leq k$ steps for  the alias related to \lhs or \rhs{} to reach 
the predecessors of node $n$.
The alias at node $n$ is generated from the alias 
$\blhs\times\bbeta \subseteq \aout^{\ex, k}_p$ or $\brhs\times\bbeta \subseteq \aout^{\ex, k}_p$ at the predecessor $p$ of node $n$.
From the inductive hypothesis, the aliases for \lhs and \rhs{} at node $p$ is present in both the methods \ex and \cd.
Below are the cases which takes place with respect to the alias computed at predecessors of $n$ :
\begin{itemize}
 \item alias $\blhs\times\bbeta\subseteq\aout^{\cd}_p$ is propagated to node $n$ or
 \item alias $\brhs\times\bbeta\subseteq\aout^{\cd}_p$ is used to compute derived alias at node $n$.
\end{itemize}
Thus we get $\balpha\times\bbeta\subseteq\aout^{\cd}_n$.

  \item $\alpha = \rhs, \beta \neq \lhs$
 
  This is a dual of the previous case with \rhs{} and \lhs interchanged.  

  \item $\alpha \neq \lhs, \alpha \neq \rhs, \beta = \lhs$
  
  Since, $\balpha\subseteq\dout^{\cd}_n$, demand for \balpha will be raised to the predecessors of $n$.
If it takes $k+1$ steps for the alias $\balpha\times \bbeta\subseteq \aout^{\ex, k+1}_n$ 
to reach the exit of node $n$ then it certainly takes $\leq k$ steps for  
the alias related to \balpha to reach the exit of
the predecessors of node $n$.
The alias at node $n$ is generated from the alias 
$\blhs\times\balpha \subseteq \aout^{\ex, k}_p$ (if $\lhs\not\equiv x$) or 
$\brhs\times\balpha \subseteq \aout^{\ex, k}_p$ at the predecessor $p$ of node $n$.
From the inductive hypothesis, the alias with respect to \balpha at node $p$ is present in both the methods \ex and \cd.
Below are the cases which takes place with respect to the alias computed at predecessors of $n$ :
\begin{itemize}
 \item alias $\blhs\times\balpha\subseteq\aout^{\cd}_p$ is propagated to node $n$ or
 \item alias $\brhs\times\balpha\subseteq\aout^{\cd}_p$ is used to compute derived alias at node $n$.
\end{itemize}
Thus we get $\balpha\times\bbeta\subseteq\aout^{\cd}_n$.

  \item $\alpha \neq \lhs, \alpha \neq \rhs, \beta = \rhs$

  This is a dual of the previous case with \rhs{} and \lhs interchanged.  
\end{itemize}
\item Propagation of alias from predecessor :  
$\alpha \neq \lhs, \alpha \neq \rhs, \beta \neq \lhs, \beta \neq \rhs$
  
  Since, $\balpha\subseteq\dout^{\cd}_n$, demand for \balpha will be raised to the predecessors of $n$.
If it takes $k+1$ steps for the alias $\balpha\times \bbeta\subseteq \aout^{\ex, k+1}_n$ 
to reach the exit of node $n$ then it certainly takes $\leq k$ steps for  the alias $\balpha\times \bbeta$ to reach the
predecessors $p$ of node $n$.
From the inductive hypothesis, the alias for demand $\balpha$ at the predecessors of $n$ is present in both the methods \ex and \cd.
Statement $n$ propagates the alias computed at the predecessor.
Thus,
$\balpha\subseteq\dout^{\cd}_n \Rightarrow \balpha\times\bbeta\subseteq\aout^{\cd}_n$. 
 \end{itemize}
\end{proof}


\subsection{Demand-driven Alias Analysis is More Precise than Algorithm \ex for Java}
\label{sec:java.precision}

 Let \vecfjd be the vector of monotonic flow functions that computes
aliases and demands using demand-driven method for Java.  
This is obtained by
 incorporating minor changes in the data flow equations presented for
 \vecfid.  Let the resulting algorithm for demand-driven alias 
 analysis of Java programs be called $\jd$.  Note that there is no
 need for speculation in Java, so demand for address-of a variable
 will not be raised.  Java does not have access expressions of the
 form $*x$ or $x\rightarrow f$.  Thus the computation of abstract names changes
and is described in Table~\ref{table-dn-Java}.  We
 present the modified equations for \ldgen and \rdgen below.
 \begin{align*}
\ldgen(\rhs,\aliasing) & = 
\begin{cases}
	\{x\} \cup \dn(\rhs, \aliasing) & \rhs \equiv x.f 
	\\
	\{x\}	&	\rhs \equiv x
	\\
	\emptyset 	& 		\text{otherwise}
	\end{cases}
	\\
\rdgen(\lhs) & = \{x\mid \lhs \equiv x.f\} 
\end{align*}

\begin{thm}
Algorithm \jd is more precise than algorithm \ex for Java programs.
\end{thm}

\begin{proof}
Follows from lemma~\ref{lemma:java:jd-ex} and
lemma~\ref{lemma:java:jd-ex-eg}, proved below.
\end{proof}

\begin{lemma}
Algorithm \jd is at least as precise as algorithm \ex for Java programs.
  \label{lemma:java:jd-ex}
\end{lemma}

\begin{proof}
Follows from corollary~\ref{corr:dd-ex} of lemma~\ref{lemma-id-precise}
since function \vecfjd is a monotonic flow function, 
and $\vecgex(\aval) = \vecfjd{(\aval, \dval^{\bot})}$.

\end{proof}
\begin{table*}
\caption{Abstract names for different access expressions}
\label{table-dn-Java}
\footnotesize
 \begin{tabular}{|c|c|c|c|c|}
\hline
$\alpha$	
	& $x$
	& $x.f$
	& new $\tau$
	& null
		\\ \hline\hline

$\dn (\oneae, \aliasing)_{\text{\em\sf asb}}$
\rule[-.6em]{0em}{1.6em}
	& $\{ x\}$
	& $\{a.f\mid (x, \&a) \in \aliasing\}$
	& $\{\&\allocsite\}$
	& $\emptyset$
		\\ \hline
$\dn (\oneae, \aliasing)_{\text{\em\sf tba}}$
\rule[-.6em]{0em}{1.6em}
	& $\{ x\}$ 
	& $\{\tau.f\mid (x,\&tau) \in \aliasing\}$
	& $\{\&\tau\}$
	& $\emptyset$
		\\ \hline
 
 \end{tabular}

\end{table*}

 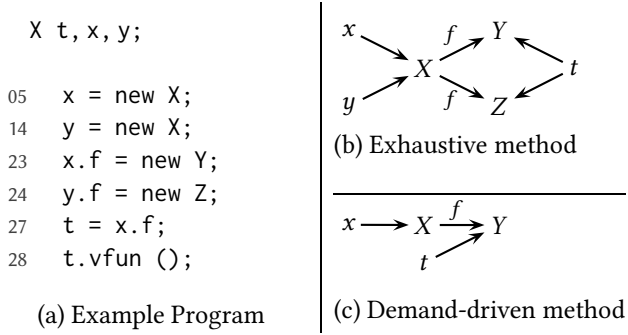
\begin{figure}[!t]
\centering
\begin{tabular}{@{}c|c@{}}
\begin{tabular}{@{}c@{}}

\setlength{\codeLineLength}{35mm}
	\begin{tabular}{@{}r@{}c}
	\codeLine{}{0}{X t,\,x,\,y; \rule{0em}{1.25em}}{white}
\\
	\codeLine{05}{1}{{x = new X;}}{white}
	\codeLine{14}{1}{{y = new X;}}{white}
	\codeLine{23}{1}{{x.f = new Y;}}{white}
	\codeLine{24}{1}{{y.f = new Z;}}{white}
	\codeLine{27}{1}{{t = x.f;}}{white}
	\codeLine{28}{1}{{t.vfun ();}}{white}
	
	\end{tabular}
	\\
	(a) Example Program
	\rule{0em}{1.75em}
	\end{tabular}
& 
\begin{tabular}{@{}l@{}}
 	\begin{pspicture}(0,0)(33,15)
	\psset{nodesep=-1.5}
	
	\putnode{x}{origin}{2}{12}{\pscirclebox[linestyle=none]{$x$}}
	\putnode{y}{x}{0}{-10}{\pscirclebox[linestyle=none]{$y$}}

	\putnode{A}{x}{10}{-5}{\pscirclebox[linestyle=none]{$X$}}
	\putnode{B}{A}{10}{5}{\pscirclebox[linestyle=none]{$Y$}}
	\putnode{C}{A}{10}{-5}{\pscirclebox[linestyle=none]{$Z$}}
	\putnode{t}{A}{20}{0}{\pscirclebox[linestyle=none]{$t$}}
	
	\ncline[nodesep=-1]{->}{x}{A}
	\ncline[nodesep=-1]{->}{y}{A}
	\ncline[nodesep=-1]{->}{A}{B}
	\aput[1pt](.4){\footnotesize$f$}
	\ncline[nodesep=-1]{->}{A}{C}
	\bput[1pt](.4){\footnotesize$f$}
	\ncline[nodesep=-1]{->}{t}{B}
	\ncline[nodesep=-1]{->}{t}{C}
	\end{pspicture} \\
	(b) Exhaustive method
\rule[-1.5em]{0em}{1em}
	\\\hline
		 	\begin{pspicture}(0,0)(30,12)
	\psset{nodesep=-1.5}
	
	\putnode{x}{origin}{2}{8}{\pscirclebox[linestyle=none]{$x$}}
	\putnode{t}{x}{10}{-5}{\pscirclebox[linestyle=none]{$t$}}

	\putnode{A}{x}{10}{0}{\pscirclebox[linestyle=none]{$X$}}
	\putnode{B}{A}{10}{0}{\pscirclebox[linestyle=none]{$Y$}}
	
	\ncline[nodesep=-1]{->}{x}{A}
	\ncline[nodesep=-1]{->}{t}{B}
	\ncline[nodesep=-1]{->}{A}{B}
	\aput[1pt](.4){\footnotesize $f$}
	\end{pspicture}
	\\
	(c) Demand-driven method
\end{tabular}

\end{tabular}
\caption{Example illustrating precision gain in demand-driven method over exhaustive method in Java.
Points-to graph using type-based abstraction for static resolution of
virtual call at line 28. 
Objects and allocation sites are annotated by their respective 
types in points-to graph shown in (b) \& (c).
Virtual function \emph{vfun} is defined in all the classes $\{X,Y,Z\}$.
Member $f$ is pointer to class $X$ and is declared in class $X$.
Class hierarchy is $X \rightarrow Y \rightarrow Z$.
}
\label{m-eg-java}
\end{figure}

 \begin{lemma}
Algorithm $\jd$ and 
algorithm $\ex$ are not equiprecise.
  \label{lemma:java:jd-ex-eg}
 \end{lemma}
 \begin{proof}
   We prove the lemma by providing an example to show that
algorithm \jd is strictly more precise than algorithm \ex.
Consider the
   example shown in Figure~\ref{m-eg-java}.  We refer to the alias
   and demand computed at line 28, where the virtual call statement is 
\text{$t.{\text{\em vfun}}()$}.  We consider the demand
   \text{$\dfun{28}{\jd} = (\{t\}, \emptyset)$}, 
and note that the alias pair
   \text{$(t, \&Z)\in \ain^{\ex}_{28}$}. 
However, \text{$(t, \&Z)\not\in \ain^{\jd}_{28}$}.  
Thus, algorithm \jd and algorithm \ex are not equiprecise.
\end{proof}

\section{Implementation and Measurements}
\label{sec:measurements}
In this section we present the implementation details, benchmark characteristics, and
efficiency and precision metrics for the different variants implemented.

\subsection{Implementation Details}
We have implemented the following three methods of alias analysis with data abstraction:
\begin{enumerate}[(a)]
\item  exhaustive method defined in Section~\ref{sec:precision-proof} (called \ex),
\item conventional demand-driven method defined in Section~\ref{sec:precision-proof} (called \cd), and 
\item our proposed demand-driven method with improved speculation defined in Section~\ref{sec:formulation} (called \id).
\end{enumerate}
We have used the following two data abstractions for modelling heap in
each of the above methods: 
\begin{enumerate}[(a)]
 \item type-based abstraction (called \tba), and 
 \item allocation-site-based abstraction (called \asb).
\end{enumerate}
This leads to six variants of alias analysis with data abstraction.
Our implementation is flow-sensitive, field-sensitive and context-insensitive for each of the six variants.

We have used a machine with 8 GB RAM with four 64-bit Intel i7 CPUs with 2.90GHz clock running Ubuntu 12.04.
We have implemented our analyses in GCC 4.7.2 as
interprocedural passes using the LTO (Link Time
Optimization) framework in order to perform whole program analysis.

\begin{table}[!t]
\begin{center}
\caption{Characteristics of benchmark programs. FP denotes indirect call-sites through function pointers}
\label{benchmark-details}
\begin{tabular}{|c|l|r|r|r|r|r|r|}
\hline
\multicolumn{2}{|c|}{\multirow{3}{*}{Program}}  & \multirow{3}{*}{LoC}  & \multicolumn{2}{c|}{\#
Functions} & \multicolumn{3}{c|}{\# Call sites} 
		\\ \cline{4-8}
\multicolumn{2}{|c|}{}			  &  & Total & Virtual & Total & Virtual & 
FP
	\\\hline
\hline

%
\multirow{4}{*}{
\begin{tabular}{c}
Open Source \\
GNU projects
\end{tabular}
} &
gengetopt	& 50280		& 247 	& 31		& 748	& 30 & 0	    	 	\\\cline{2-8}
& gperf 		& 6528	 	& 213 	& 22 		& 776 	& 31 & 0	   		\\\cline{2-8}
& gengen		& 18380		& 917	& 139		& 1061	& 55 & 2	   		\\\cline{2-8}
& motti		& 9294		& 1112	& 65		& 3383	& 243 & 1	 	  	\\\hline\hline
\multirow{4}{*}{SPEC CPU 2006} &
soplex		& 28277		& 1690	& 340		& 6324	& 402 & 2	   		\\\cline{2-8}
& povray		& 78705		& 2021	& 64		& 13058	& 105 & 52	   		\\\cline{2-8}
& dealII		& 92542		& 10074	& 1485		& 17932	& 134 & 8	   		\\\cline{2-8}
& omnetpp		& 27131		& 2505	& 747		& 8427	& 314 & 23	   		\\\hline\hline
%
\multirow{7}{*}{ITK} &
MSO                        &  63460 	& 4778	& 1601	& 6692 & 2554 & 4		\\\cline{2-8}
& RWSO                  &  61984	& 4602	& 1484	& 5193 & 1944 & 3	\\\cline{2-8}
%
& SOTISC &  34108	& 2614  & 720   & 2196  & 991  & 5	\\\cline{2-8}
& ICP3                      &  33439	& 3110  & 764   & 2331  & 943  & 8	\\\cline{2-8}
%
%
& ISO                         &  29108	& 2441  & 638   & 1295  & 582   & 3\\\cline{2-8}
& EMMME  &  21584	& 1824  & 440   & 1658  & 575   & 4\\\cline{2-8}
& ITLSA                 &  24863	& 2295  & 519   & 1402  & 529   & 5\\	\hline
%
%


\end{tabular}
\end{center}
\end{table}

\subsection{Benchmark Details}
\begin{table*}[!t]
\begin{center}
\caption{Analysis time (in ms) for different analyses}	
\label{time-data}
\begin{tabular}{|l|r|r|r||r|r|r|}
\hline
\multirow{2}{*}{Program} & \multicolumn{3}{c||}{\tba} &\multicolumn{3}{c|}{\asb} \\\cline{2-7}
& \multicolumn{1}{c|}{\id} & \multicolumn{1}{c|}{\cd} & \multicolumn{1}{c||}{\ex}& \multicolumn{1}{c|}{\id} & \multicolumn{1}{c|}{\cd} & \multicolumn{1}{c|}{\ex} 
\\\hline\hline

gengetopt-2.22.6	& 285		& 608		& 939		& 388 		& 630		& 1452		\\\hline
gperf-3.0.4 		& 94		& 82		& 108 		& 122		& 98		& 678		\\\hline
gengen-1.4.2		& 580		& 689		& 16384		& 681   	& 35076		& 38331     	\\\hline
motti-3.1.0 		& 1620	 	& 12945		& 870894 	& 4697	 	& 14529		& 909710 	\\\hline\hline
450.soplex		& 25781		& 108752	& 442668 	& 325808 	& 181503	& 540879 	\\\hline		
453.povray		& 12956		& 23898		& 44767		& 105171 	& 174075	& 370408 	\\\hline
447.dealII		& 120869	& 586811	& 2220179	& 358093 	& 610866	& 9535619 	\\\hline
471.omnetpp		& 141027	& 300127	& 408945	& 189240  	& 444046	& 8812114   	\\\hline\hline
MSO         	 	& 49702		& 87234		& 137123	& 136599 	& 320518	& 5581965 	\\\hline
RWSO 			& 49256		& 58887		& 128457	& 116221 	& 219636	& 3324397 	\\\hline
SOTISC      		& 3096		& 4686		& 5171		& 5112   	& 8172		& 359291	\\\hline
ICP3                   	& 1288		& 2877		& 3428		& 1320   	& 4395		& 424310	\\\hline
ISO               	& 1176		& 1271		& 1375		& 1349 		& 3052		& 90486		\\\hline
EMMME  		   	& 661		& 449		& 938		& 1067 		& 4352		& 77684		\\\hline
ITLSA                  	& 285 		& 240		& 611		& 372 		& 2392		& 6487		\\\hline

\end{tabular}
\end{center}
\end{table*}

We have evaluated the effectiveness of the six variants on a variety of C++ programs from 
open source GNU projects\footnote{
Available for download at : 
\htmladdnormallink{https://www.gnu.org/software/software.html}{https://www.gnu.org/software/software.html}
}, SPEC CPU 2006 benchmarks\footnote{
Details can be found at : 
\htmladdnormallink{https://www.spec.org/cpu2006/}{https://www.spec.org/cpu2006/}
}, and 
the Insight Segmentation and Registration Toolkit (ITK)\footnote{
Available for download at : 
\htmladdnormallink{https://itk.org/}{https://itk.org/}
} from the world of graphics.
SPEC CPU 2006 benchmarks are the standard set of benchmarks used for empirical measurements.
Since it contains few benchmarks for C++, we explored applications from open source 
GNU Projects. However, the applications we considered had less number of virtual function calls.
Our choice to include ITK toolkit
was dictated by the use of virtual function calls in C++ programs because our experiments
are done in the context of static resolution of virtual function calls.
The programs in ITK toolkit have a much larger number
of virtual function calls than the programs  
in the SPEC CPU 2006 benchmark suite.

The details of our programs
are shown in Table~\ref{benchmark-details}.
We have not considered \text{444.namd} and \text{473.astar} SPEC CPU 2006 benchmark programs because they do not contain virtual 
call-sites. 
Also, we do not present data for the benchmark 483.xalancbmk from SPEC CPU 2006 because the size of the program is 275kLoC and 
our current implementation does not scale to it. 

\subsection{Empirical Observations}

In this section we describe our observations.

\subsubsection{Time Measurements}
Table~\ref{time-data} presents the time taken by each of the six variants. It is clear that
our proposed method \id is far more efficient than \ex{} with a 
speedup factor greater than two in 11 and 14 cases using abstractions \tba and \asb respectively.
It also achieves a speedup as high as 537 in case of motti using \tba and a speedup of 321 in case of ICP3 using \asb.
It is faster than \cd in most cases (12 under \tba and 13 under \asb) and slower than \cd in some cases (three under \tba and 
two under \asb). 
In the cases where it is slower, the performance degradation is not greater than two; in the cases where it is faster, the
improvement is significant going to a factor 
of seven in case of motti using \tba and 51 in case of gengen using \asb.


\subsubsection{Precision Measurements}
In order to understand the precision of the variants implemented, we have measured the following data.
\begin{itemize}
\item Count of monomorphic call sites discovered (Table~\ref{devirt-data1}).
\item Count of virtual call edges identified in the call graph (Table~\ref{devirt-data2}).
\item Count of class types identified for objects used in the program (Table~\ref{class-count-data}).
\end{itemize}

The first metric indicates the number of cases in which a virtual call has been completely resolved 
eliminating
the need of indirect lookup in the virtual function table as such calls can be 
translated to direct calls. 
A larger number indicates more precise information.
The second metric reports the count of functions that can be invoked at the virtual function calls in the program.
Precision of this information has
an impact on the interprocedural analysis using such a call graph.
A smaller number indicates more precise information.
The third metric presents the count of types identified by all the objects used in the program.
Such an information is useful for type dependent clients like may-fail cast, or resolving other run time type information like use of
\emph{typeid} in C++ or \emph{instanceof} in Java~\cite{lhotak, Wang:2001:PCT:646158.679874, Sridharan:2013:AAO:2554511.2554523, Balatsouras:2013:CHC:2509136.2509530}.
For this metric also, a smaller number indicates more precise information.

\begin{table}
\begin{center}
\caption{Data related to virtual call-sites}
\label{devirt-data}

\begin{subtable}{0.45\linewidth}
               \centering

\caption{Count of monomorphic call-sites discovered. (A larger number indicates more precise information.)}
\label{devirt-data1}
\begin{tabular}{|l|r|r||r|r|}
\hline
\multirow{2}{*}{Program} & \multicolumn{2}{c||}{\tba} &\multicolumn{2}{c|}{\asb} \\\cline{2-5}
  & \multicolumn{1}{c|}{\id} & \multicolumn{1}{c||}{\ex/\cd} & \multicolumn{1}{c|}{\id} & \multicolumn{1}{c|}{\ex/\cd} \\\hline\hline
gengetopt & 30 & 30 & 30 & 30 \\\hline
 gperf  & 18 & 18 & 18 & 18 \\\hline
 gengen & 54 & 51 & 54 & 51\\\hline
 motti   & 235 & 231 & 235 & 231 \\\hline\hline
 soplex   & 364 & 351 & 364 & 356\\\hline
 povray & 98 & 85 & 98 & 85\\\hline
 dealII   & 121 & 117 & 121 & 117 \\\hline
 omnetpp  & 313 & 313 & 313 & 313\\\hline\hline
MSO                        &2402 & 2394 &  2402 & 2395\\\hline
RWSO                      & 1815  & 1810   & 1815 &1815  \\\hline
SOTISC 		          & 934  & 929 & 934   & 934 \\\hline
ICP3                        &887 & 807 & 887  & 887\\\hline    
ISO                          & 547  & 537 & 548  & 547\\\hline
EMMME                      & 532 & 491 & 532 & 532\\\hline   
ITLSA                        & 497 & 490  & 497 & 496\\\hline                  
\end{tabular}

           \end{subtable}%
           \hspace*{.8em}
           \begin{subtable}{0.48\linewidth}
               \centering
\caption{Count of virtual call edges in the call graph. (A smaller number indicates more precise information.)}
\label{devirt-data2}
\begin{tabular}{|l|r|r||r|r|}
\hline
\multirow{2}{*}{Program} & \multicolumn{2}{c||}{\tba} &\multicolumn{2}{c|}{\asb} \\\cline{2-5}
  & \multicolumn{1}{c|}{\id} & \multicolumn{1}{c||}{\ex/\cd} & \multicolumn{1}{c|}{\id} & \multicolumn{1}{c|}{\ex/\cd} \\\hline\hline
gengetopt & 30 & 30 & 30 & 30 \\\hline
 gperf  & 50 & 50 & 50 & 50 \\\hline
 gengen & 56 & 59 & 56 & 59\\\hline
 motti   & 271 & 278 & 271 & 278 \\\hline\hline
 soplex   & 505 & 538 & 505 & 523\\\hline
 povray & 112 & 126 & 112 & 126\\\hline
 dealII   & 148 & 160 & 148 & 160 \\\hline
 omnetpp  & 315 & 315 & 315 & 315\\\hline\hline
MSO                        &2801 & 2888 &  2800 & 2886\\\hline
RWSO                      & 2135  & 2221   & 2134 & 2210  \\\hline
SOTISC 		          & 1058  & 1088 & 1057   & 1087 \\\hline
ICP3                        &1009 & 1090 & 1009  & 1010\\\hline    
ISO                          & 640  & 649 & 618  & 639\\\hline
EMMME                      & 637 & 678 & 629 & 629\\\hline   
ITLSA                        & 565 & 575  & 562 & 566\\\hline                  
\end{tabular}

           \end{subtable}
\end{center}
\end{table}

\begin{table}
\begin{center}
\caption{Count of class types identified for the objects used in the program. (A smaller number indicates more precise information.)}
\label{class-count-data}
\begin{tabular}{|l|r|r||r|r|}
\hline
\multirow{2}{*}{Program} & \multicolumn{2}{c||}{\tba} &\multicolumn{2}{c|}{\asb} \\\cline{2-5}
  & \multicolumn{1}{c|}{\id} & \multicolumn{1}{c||}{\ex/\cd} & \multicolumn{1}{c|}{\id} & \multicolumn{1}{c|}{\ex/\cd} \\\hline\hline
gengetopt 	& 328 		& 443 		& 162	 			& 320	 	\\\hline
gperf  		& 550		& 725 		&  398 				& 653 		\\\hline
gengen 		& 84230 	& 88707 	&  65102 			& 66129		\\\hline
motti   	& 22992		& 46274 	&  22520			& 23291		\\\hline\hline
soplex   	& 218547 	& 308779 	& 185053			& 219290 	\\\hline
povray 		& 14422		& 27945 	& 13919 			& 17027		\\\hline
dealII   	& 391729	& 992147	& 390136			& 404188	\\\hline
omnetpp  	& 61072		& 69375		& 60041 			& 60540		\\\hline\hline
MSO             & 170907	& 171018	& 170311 			& 170982	\\\hline
RWSO            & 139798	& 139839	& 106182 			& 108202	\\\hline
SOTISC          & 16739		& 16890		& 16010 			& 16053		\\\hline
ICP3            & 14189		& 23194		& 9412 				& 9595		\\\hline    
ISO             & 6859		& 7051		& 6749 				& 6819		\\\hline
EMMME           & 10301		& 10375		& 10129				& 10239		\\\hline   
ITLSA           & 7946		& 7988 		& 7706 				& 7747		\\\hline                  
\end{tabular}
\end{center}
\end{table}

The count of monomorphic call-sites identified, the count of virtual call edges discovered and the count of class types identified for
the objects used in the program are identical for \cd and \ex.
This concurs well with the formal proof 
of this equivalence presented in Section~\ref{sec:precision-proof}.

It is well known that \asb is more precise than \tba.
This is evidenced by our data as well. However, the improvement of \asb over \tba is marginal for our type dependent client.
For example,
the count of monomorphic call-sites identified by \asb-\id is 548 as against 547 identified by
 \tba-\id in case of ISO.
Similarly, the count of virtual call edges discovered in the call graph by \asb-\id is marginally less than that by \tba-\id.
Thus, type-based abstraction suffices for type dependent clients.

The count of monomorphic call-sites discovered by \id is larger,
the count of virtual call edges discovered in the call graph by \id is smaller, and
the count of class types discovered by \id is far smaller than \cd and \ex using both the abstractions.
It is interesting to note that with type-based abstraction, 
we identify 24\% fewer types in 7 cases with the reduction increasing significantly to
 50\% and 60\%  for \texttt{motti} and \texttt{dealII} programs respectively.
 Similar traits are seen when allocation-site-based abstraction is
used---the reduction in the number of types is 39\% and 49\% 
 for \texttt{gperf} and \texttt{gengetopt} programs respectively.
This is in concurrence with the formal proof 
 that \id is more precise than \cd and \ex (Section~\ref{sec:precision-proof}).

\section{Related Work}
\label{sec:related.work}

We have classified the factors governing precision and efficiency of pointer analysis in Section~\ref{sec:intro}.
Various dimensions related to precision of pointer analysis have been explored in~\cite{Ryder:2003:DPR:1765931.1765945}.
They have not considered the dimension of the quantum of information required.
Our classification focuses on the information computed by an analysis, because our work shows
 that demand-driven method
can be more precise than the exhaustive method.
We have discussed literature related to bidirectional analysis in detail in Section~\ref{sec:formalizing-bidirectional}.
In this section we focus on the literature related to demand-driven pointer analysis, devirtualization, type-based analysis and control flow abstractions.

\subsection{Demand-driven Methods}
The information computed by demand-driven methods is governed
by
a client or an application. 
In some cases, a demand can also be raised internally by an analysis.
Some client driven applications desire quick 
results~\cite{Sridharan:2006:RCP:1133981.1134027, Guyer:2003:CPA:1760267.1760284, 
Sridharan:2005:DPA:1094811.1094817, Yan:2011:DCA:2001420.2001440}, whereas 
taint analysis that deals with security of information flow require highly precise results~\cite{spth_et_al:LIPIcs:2016:6116,Arzt:2014:FPC:2594291.2594299,Huang:2016:DSD:2950290.2950348}.
 In our analysis, demand is not raised by a client but is governed by the application
which specifies the objects whose type information is sought for.

A demand-driven method has always been considered to enhance the efficiency of an analysis as it computes
only the information required to meet a given set of demands.
The work described in~\cite{Sridharan:2005:DPA:1094811.1094817} is client driven and 
their aim is to provide results within a given time limit.
This restriction affects the precision of the information being computed.
With respect to precision, the work described in ~\cite{Agrawal:2002:EDD:647478.727927,Heintze:2001:DPA:378795.378802,Zheng:2008:DAA:1328438.1328464,Duesterwald:1997:PFD:267959.269970,Saha:2005:IDP:1069774.1069785}
suggests that 
the precision of demand-driven methods is similar to that of the corresponding exhaustive methods.

A demand-driven points-to analysis~\cite{Heintze:2001:DPA:378795.378802}
based on Andersen's flow-insensitive points-to analysis~\cite{Andersen94programanalysis}
has been used to construct call graphs in the presence of function 
pointers.
They have proved precision equivalence between demand-driven method and the exhaustive method.
This holds true for the subset of C language 
they have considered which does not deal with objects on the heap.

Liveness-based points-to analysis~\cite{lfcpa} has also been formulated with respect to the language which does not deal with 
objects on heap.
They perform the analysis in a demand-driven fashion by computing points-to information for the pointers that 
are live.
This computes a smaller amount of points-to information at each program point by 
safely ignoring the statements that are not related to the pointer variables
that are live.
Both these works~\cite{lfcpa,Heintze:2001:DPA:378795.378802} are governed by
the philosophy of raising demands for pointer indirection $*x$ appearing on \text{\sf lhs}, which
we refer to as the conventional speculation strategy.
Extending the notion of conventional speculation of demands to a language that deals with objects on heap,
results in achieving precision equivalent to the exhaustive method.
However, there is a loss of precision as compared to our proposed speculation,
as discussed in Section~\ref{sec:causes.of.imprecision} and Section~\ref{sec:precision-proof}.

Most of the demand-driven 
methods~\cite{spth_et_al:LIPIcs:2016:6116,Yan:2011:DCA:2001420.2001440,Sridharan:2005:DPA:1094811.1094817,Sridharan:2006:RCP:1133981.1134027}
draw comparison in terms of precision and efficiency with other demand-driven methods and 
do not compare with their exhaustive counterparts.
This is because the main focus with respect to demand-driven methods has been to improve the efficiency of the analysis.
However, our work emphasizes on the precision gain of demand-driven method along with achieving efficiency.

\subsection{Devirtualization}
Virtual function calls can be resolved using a type analysis based on class 
hierarchies~\cite{Aigner:1996:EVF:646154.679544,Bacon:1996:FSA:236337.236371,Dean:1995:OOP:646153.679523,Tip:2000:SPC:353171.353190}.
These analyses are efficient but are imprecise.
Hence, pointer analyses have also been used to construct call graphs in the
presence of
function pointers~\cite{Heintze:2001:DPA:378795.378802, callahan, p6, p5, p7, ryder} 
or virtual functions~\cite{Diwan:2001:UTA:383721.383732,Sundaresan:2000:PVM:353171.353189,Pande94, Pande95,Harini.types}.
Our focus is different from these methods because we wish to show the precision gain 
that can be achieved by a demand-driven method.

As noted by Hind ~\cite{Hind:2001:PAH:379605.379665}, it is more useful to create generic solutions instead
of solutions for a specific application. Our formulation is indeed generic and is useful
 for a wide range of demand-driven problems.
Specifically, the variant using the type-based abstraction is suitable for applications that rely on type
information like call graph construction, devirtualization, and identifying safety of casts (determining if the
cast may fail at runtime). 
All these require run time type information (RTTI) which needs the 
type of pointees where the pointer could point-to. 

\subsection{Type-based Abstraction}
Type-based abstraction suffices for type dependent applications.
A survey of type-based analysis is described in~\cite{Palsberg:2001:TAA:379605.379635}.
They discuss the fundamental type-based analysis for object oriented programs like CHA and RTA.
The work described in~\cite{Diwan:2001:UTA:383721.383732} have found significant improvement with the use of types for 
 virtual function resolution for  Modula-3. 
However, a language like C++ requires a more powerful alias analysis that considers pointer 
dereferences, the address-of operator, and pointer arithmetic.
Structure-sensitive points-to analysis~\cite{Balatsouras2016} uses type information to filter the points-to sets that results 
from analysis imprecision. 
Thus, they create multiple abstract objects per allocation-site based on the type information.
Types have also been looked at in a different light by using it to distinguish between calling contexts~\cite{lhotak}.
Their type-sensitivity is scalable and achieves a good precision as compared to object-sensitivity.
The use of types instead of allocation-sites makes the analysis more efficient 
without compromising on precision, 
as evidenced by our results as well.
Type-based analyses are simple as the type information is easily available and makes the analysis 
efficient~\cite{Palsberg:2001:TAA:379605.379635,Diwan:2001:UTA:383721.383732, lhotak}.

\subsection{Control Flow Abstractions}
Precision and efficiency requirement of an analysis helps select the right control flow abstractions.
Two well known analysis in the literature are Andersen's subset based points-to analysis~\cite{Andersen94programanalysis}
and Steensgaard's equality based points-to analysis~\cite{Steensgaard:1996:PAA:237721.237727}.
Both these analysis are flow-insensitive, computing imprecise results of varying degree.
Other analyses can be built on top of such quick and imprecise analyses to improve 
the precision~\cite{Hardekopf:2011:FPA:2190025.2190075,Hasti:1998:USS:277650.277668,Tan:2017:EPP:3062341.3062360}.
Path-sensitive analysis built on top of flow- and context-sensitive 
level by level analysis~\cite{level-by-level} is discussed in ~\cite{Sui:2011:SSP:2183641.2183659}.

Flow-sensitivity and context sensitivity increases the precision of an analysis as discussed in 
~\cite{lfcpa,Kahlon:2008:BTS:1375581.1375613,Zhu:2005:TSF:1065579.1065798,Li:2013:PSC:2491894.2466483,level-by-level}.
Various strategies have been explored to achieve context-sensitivity.
They include performing cloning based context sensitive analysis ~\cite{Whaley:2004:CCP:996841.996859},
refinement based context sensitive analysis ~\cite{Sridharan:2006:RCP:1133981.1134027}, and
imposing a restriction on the length of the call-strings for efficiency~\cite{Shivers91control-flowanalysis}.
However, restriction on the length of the call strings could introduce cycles in the program which could make the analysis slow 
as described in ~\cite{Oh:2010:AML:1840756.1840757} 
indicating that the need to compromise on precision to achieve efficiency might not always be true.

Object sensitivity~\cite{Milanova:2005:POS:1044834.1044835} is a different form of context sensitivity which is efficient and
improves on the precision of flow-insensitive analysis for Java.
Object sensitive analysis~\cite{Milanova:2005:POS:1044834.1044835} distinguishes between procedures based on the invoking object.
It uses allocation-site-based abstraction to model heap objects. As 
rightly noted by~\cite{lhotak}, apt name for object sensitive analysis would be allocation-site analysis.
Object sensitive analysis is based on the insight (which is left implicit) that data abstractions 
can be used for control flow abstractions.
This holds true for object oriented languages.


\section{Conclusions and Future Work}
\label{sec:conclusions}
Demand-driven methods are more efficient than the exhaustive methods because they compute only the information required
for meeting a set of demands.
In the presence of data abstraction, this reduction in the information required, 
can reduce the imprecision of a demand-driven method.
The conventional demand-driven methods raise additional demands speculatively
in the presence of indirect assignments. Such a 
speculation is inevitable for languages like C or C++ because indirect assignments could
modify variables whose addresses are taken.
We have proposed an alternative method that does not speculate 
a
 demand for indirect assignments.
Instead, it raises demands for address-of a variable which does not affect the precision of an analysis.
Thus, a demand-driven method is not only efficient as is always thought of, but it can also 
improve the precision of an analysis.

A demand-driven method is inherently bidirectional with demands flowing against the control flow and information 
related to the demands flowing along the control flow.
We classify bidirectional dependencies in order to understand the bidirectional flows of information.
We define \mop solution for bidirectional analyses which mandates defining paths for bidirectional analyses.
Since there is flow of information in both the directions, the definition of paths in \mop for unidirectional
analyses is not applicable for bidirectional analyses.
We define paths in \mop for bidirectional analyses as \emph{qualified control flow paths}.
This notion of paths is applicable to all demand-driven methods including slicing, taint analysis and 
liveness based points-to analysis. 
The definition of \mop facilitates proving the soundness of bidirectional methods; indeed we use it
to prove the soundness of our method.

We have formulated our demand-driven method as a bidirectional data flow analysis and
 implemented it for C++. We have compared its effectiveness with that
of the conventional demand-driven method and the exhaustive method.
We have used the following precision metric:
\begin{inparaenum}[(a)]
\item the count of monomorphic call-sites discovered,
\item the count of virtual call edges identified in the call graph, and
\item the count of class types identified for objects used in the program.
\end{inparaenum}
We have found that the precision of the 
conventional demand-driven
method and the exhaustive method for the demands raised is equivalent. 
Our proposed demand-driven method is more precise 
than the other two variants on all factors in the metric.
It is also more efficient than the exhaustive method and the conventional
demand-driven method with 
a small number of exceptions.

Our concepts are generic and are applicable to Java too.
We demonstrate this by formulating
our analysis for Java.
In case of Java, there is no need for any speculation due to the
absence of a pointer to a variable because Java does not support the
address-of (\&) operator thereby prohibiting indirect assignment to variables.
Thus a demand-driven method 
is expected to be more precise than
its exhaustive counterpart in case of Java. 
We have established this claim formally.


Some possible future directions of this work are as follows :
exploring the precision and efficiency trade-of by making 
the demand-driven alias analysis
context-sensitive~\cite{call_string.vbt,vasco};
creating a pointer-based slicing method that computes the relevant pointer
information precisely on a need basis; and 
using the demand-driven alias analysis for applications 
like taint analysis that 
require highly precise pointer information.

\begin{acks}
Swati Jaiswal is partially supported by a TCS Research Fellowship.
\end{acks}



\bibliography{demand-wisely}
\bibliographystyle{ACM-Reference-Format}

\appendix
\section{Detailed Working of Demand-driven Approach Using Different Speculation Strategy}\label{detailed-working}
We present working of the conventional demand-driven alias
analysis~\cite{Heintze:2001:DPA:378795.378802,lfcpa,Hirzel:2002:UTL:586088.586089} 
in Figure~\ref{m-eg-working-lfcpa}.
Conventional methods raise demand for indirect assignment statements of the form $*p = x$ and $y\rightarrow f = new\ Z$,
speculating that such a demand could be an alias of the demand reaching that statement.
We have seen that such a speculation is necessary to ensure soundness in Example~\ref{exmp:demand.imprecision.1} in Section~\ref{sec:speculation}.
The virtual call statement at line 28 raises demand for $ t$, which in turn raises demand for $z$ at line 27.
Statement at line 15 kills the demand for $ z$ and raises demand for \texttt{p} speculating that $p$ could point to \texttt{z}.
Points-to information $\{p\rightarrow z\}$ is propagated in the forward direction.
Since \texttt{p} points-to \texttt{z}, demand for \texttt{x} will be raised by line 15.
Line 05 computes the points-to information $x\rightarrow X$. 
Line 15 uses that points-to information
and computes $z\rightarrow X$.
Now abstract name for $ z\rightarrow f$ can be computed at line 27 and thus demand for abstract name $ X.f$ is raised.
Speculating $y\rightarrow f$ at line 24 to be an alias of $X.f$, demand for $y$ is raised.
Because $x$ points-to $X$, points-to information $X\!\stackrel{f}\rightarrow Y$ is computed at line 23.
The spurious information that $ y$ points-to $ X$ results in computing the points-to
information $X\!\stackrel{f} \rightarrow Z$ at line 24. 
The spurious node alias between $x$ and $y$ results in spurious link aliases.
Thus imprecise points-to information $t \rightarrow \{X, Y\}$ is computed.

We present working of our proposed improved demand-driven method in Figure~\ref{m-eg-working-our-demand}.
The virtual call statement at line 28 raises demand for $ t$ which in turn raises demand for $ z$ and $ \&z$ at line 27.
Instead of speculating and raising demands at indirect assignment statements, we raise demand for address-of a variable,
when address of such a variable is taken in the program. 
Address of a variable seeks to find pointers to it. Thus demand for $\&z$ helps compute the points-to information
$q\rightarrow z$ and $p\rightarrow z$ at lines 03 and 04 respectively.
Since $ p$ points-to $z$, demand for $x$ will be raised by line 15.
This identifies points-to information $x\rightarrow X$ at line 05.
This information is used by line 15 to compute points-to information $z\rightarrow X$.
Now abstract name for $z\rightarrow f$ can be identified at line 27 and thus demand for abstract name $ X.f$ is raised.
Thus points-to information that $X\!\stackrel{f} \rightarrow Y$ is computed by line 23.
Note that since demand for $y$ has not been raised, abstract name for $y\rightarrow f$  cannot be deduced.
Thus statement in line 24 is not used to compute points-to information.
This helps in eliminating spurious link aliases by avoiding spurious node aliases which gets introduced due to data abstractions.
Thus, precise points-to information $t\rightarrow \{Y\}$ is computed.

\begin{figure*}[!t]
\centering

\renewcommand{\codeLine}[4]{\protect#1\protect#3}
\footnotesize
\begin{tabular}[t]{|l|c|c|c|c|c|c|}
\hline
\multirow{3}{*}{Program}
	& \multicolumn{6}{c|}{Data flow information at the program points just after the numbered statements}
	\\ \cline{2-7}
	& \multicolumn{2}{c|}{Round \#1}
	& \multicolumn{2}{c|}{Round \#2}
	& \multicolumn{2}{c|}{Round \#3}
	\\ \cline{2-7}
	& \multicolumn{1}{@{}c@{}|}{\;\raisebox{-.3mm}{Demand}\;} & \multicolumn{1}{c|}{\raisebox{-.3mm}{PTG edges}}
	& \multicolumn{1}{@{}c@{}|}{\;\raisebox{-.3mm}{Demand}\;} & \multicolumn{1}{@{}c@{}|}{\raisebox{-.3mm}{PTG edges}}
	& \multicolumn{1}{@{}c@{}|}{\;\raisebox{-.3mm}{Demand}\;} & \multicolumn{1}{@{}c@{}|}{\raisebox{-.3mm}{PTG edges}}
	\\ \hline

	03 : \texttt{q = \&z;}
	& \rnode{n01}{$\emptyset$}
	& \rnode{n02}{$\emptyset$}
	& \rnode{n03}{$\emptyset$}
	& \rnode{n04}{$\emptyset$}
	& \rnode{n05}{$\{\texttt{X.f}\}$}
	& \rnode{n06}{$\emptyset$}
	\\ \hline
	04 : \texttt{p = \&z;}
	& \rnode{n11}{$\{\tt p\}$}
	& \rnode{n12}{$\{\tt p\!\rightarrow z\}$}
	& \rnode{n13}{$\emptyset$}
	& \rnode{n14}{$\emptyset$}
	& \rnode{n15}{$\{\tt X.f\}$}
	& \rnode{n16}{$\emptyset$}
	\\ \hline
	05 : \texttt{x = new X;}
	& \rnode{n21}{$\{\tt p\}$}
	& \rnode{n22}{$\{\tt p\!\rightarrow z\}$}
	& \rnode{n23}{$\{\tt x\}$}
	& \rnode{n24}{$\{\tt x\!\rightarrow X\}$}
	& \rnode{n25}{$\{\tt X.f\}$}
	& \rnode{n26}{$\emptyset$}
	\\ \hline
	14 : \texttt{y = new X;}
	& \rnode{n31}{$\{\tt p\}$}
	& \rnode{n32}{$\{\tt p\!\rightarrow z\}$}
	& \rnode{n33}{$\{\tt x\}$}
	& \rnode{n34}{$\{\tt x\!\rightarrow X\}$}
	& \rnode{n35}{$\{\tt X.f,y\}$}
	& \rnode{n36}{$\{\tt y\!\rightarrow X\}$}
	\\ \hline
	15 : \texttt{*p = x;}
	& \rnode{n41}{$\{\tt z\}$}
	& \rnode{n42}{$\{\tt p\!\rightarrow z\}$}
	& \rnode{n43}{$\emptyset$}
	& \rnode{n44}{$\{\tt x\!\rightarrow X,z\!\rightarrow X\}$}
	& \rnode{n45}{$\{\tt X.f,y\}$}
	& \rnode{n46}{$\{\tt y\!\rightarrow X\}$}
 	\\ \hline
	23 : \texttt{x->f = new Y;}
	& \rnode{n51}{$\{\tt z\}$}
	& \rnode{n52}{$\{\tt p\!\rightarrow z\}$}
	& \rnode{n53}{$\emptyset$}
	& \rnode{n54}{$\{\tt x\!\rightarrow X,z\!\rightarrow X\}$}
	& \rnode{n55}{$\{\tt X.f,y\}$}
	& \rnode{n56}{$\{\tt y\rightarrow X,X\!\stackrel{f}{\rightarrow}\!Y\}$}
 	\\ \hline

	24 : \texttt{y->f = new Z;}
	& \rnode{n61}{$\{\tt z\}$}
	& \rnode{n62}{$\{\tt p\!\rightarrow z\}$}
	& \rnode{n63}{$\emptyset$}
	& \rnode{n64}{$\{\tt x\!\rightarrow X,z\!\rightarrow X\}$}
	& \rnode{n65}{$\{\tt X.f\}$}
	& \rnode{n66}{$\{\tt y\!\rightarrow X,X\!\stackrel{f}{\rightarrow}\!\{Y, Z\}\}$}
	\\ \hline

	27 : \texttt{t = z->f;}
	& \rnode{n71}{$\{\tt t\}$}
	& \rnode{n72}{$\{\tt p\!\rightarrow z\}$}
	& \rnode{n73}{$\emptyset$}
	& \rnode{n74}{$\{\tt x\!\rightarrow X,z\!\rightarrow X\}$}
	& \rnode{n75}{$\emptyset$}
	& \rnode{n76}{$\{\tt y\!\rightarrow X, X\!\stackrel{f}{\rightarrow}\!\{Y, Z\}, t\!{\rightarrow}\{Y, Z\}\}$}
	\\ \hline

	28 : \texttt{t->vfun ();}

	& \rnode{n81}{}
	& \rnode{n82}{}
	& \rnode{n83}{}
	& \rnode{n84}{}
	& \rnode{n85}{}
	& \rnode{n86}{}
	\\ \hline
\end{tabular}
\caption{Multiple rounds of demand and points-to graph (PTG) propagation with conventional demand-driven method for static resolution of
virtual call at line 28 of the program in Figure~\ref{m-eg-type}.
The final PTG at a program point is the union of all edges added in each round at that program point.
	Demands are propagated in the backward direction whereas PTG edges are propagated in the forward direction.
}
\label{m-eg-working-lfcpa}
\end{figure*}

\begin{figure*}[!t]
\centering
\footnotesize
\begin{tabular}[t]{|l|c|c|c|c|c|c|}
\hline
\multirow{3}{*}{Program}
	& \multicolumn{6}{c|}{Data flow information at the program points just after the numbered statements}
	\\ \cline{2-7}
	& \multicolumn{2}{c|}{Round \#1}
	& \multicolumn{2}{c|}{Round \#2}
	& \multicolumn{2}{c|}{Round \#3}
	\\ \cline{2-7}

	& \multicolumn{1}{@{}c@{}|}{\;\raisebox{-.3mm}{Demand}\;} & \multicolumn{1}{@{}c@{}|}{\raisebox{-.3mm}{PTG edges}}
	& \multicolumn{1}{@{}c@{}|}{\;\raisebox{-.3mm}{Demand}\;} & \multicolumn{1}{@{}c@{}|}{\raisebox{-.3mm}{PTG edges}}
	& \multicolumn{1}{@{}c@{}|}{\;\raisebox{-.3mm}{Demand}\;} & \multicolumn{1}{@{}c@{}|}{\raisebox{-.3mm}{PTG edges}}
	\\ \hline

	03 : \texttt{q = \&z;}
	& \rnode{n01}{$\{\tt z, \&z\}$}
	& \rnode{n02}{$\{\tt q\!\rightarrow z\}$}
	& \rnode{n03}{$\emptyset$}
	& \rnode{n04}{$\emptyset$}
	& \rnode{n05}{$\{\texttt{X.f}\}$}
	& \rnode{n06}{$\emptyset$}
	\\ \hline
	04 : \texttt{p = \&z;}
	& \rnode{n11}{$\{\tt z, \&z\}$}
	& \rnode{n12}{$\{\tt q\!\rightarrow z, p\!\rightarrow z\}$}
	& \rnode{n13}{$\emptyset$}
	& \rnode{n14}{$\emptyset$}
	& \rnode{n15}{$\{\tt X.f\}$}
	& \rnode{n16}{$\emptyset$}
	\\ \hline
	05 : \texttt{x = new X;}
	& \rnode{n21}{$\{\tt z, \&z\}$}
	& \rnode{n22}{$\{\tt  q\!\rightarrow z, p\!\rightarrow z\}$}
	& \rnode{n23}{$\{\tt x\}$}
	& \rnode{n24}{$\{\tt x\!\rightarrow X\}$}
	& \rnode{n25}{$\{\tt X.f\}$}
	& \rnode{n26}{$\emptyset$}
	\\ \hline
	14 : \texttt{y = new X;}
	& \rnode{n31}{$\{\tt z, \&z\}$}
	& \rnode{n32}{$\{\tt  q\!\rightarrow z, p\!\rightarrow z\}$}
	& \rnode{n33}{$\{\tt x\}$}
	& \rnode{n34}{$\{\tt x\!\rightarrow X\}$}
	& \rnode{n35}{$\{\tt X.f\}$}
	& \rnode{n36}{$\emptyset$}
	\\ \hline
	15 : \texttt{*p = x;}
	& \rnode{n41}{$\{\tt z, \&z\}$}
	& \rnode{n42}{$\{\tt  q\!\rightarrow z, p\!\rightarrow z\}$}
	& \rnode{n43}{$\emptyset$}
	& \rnode{n44}{$\{\tt x\!\rightarrow X,z\!\rightarrow X\}$}
	& \rnode{n45}{$\{\tt X.f\}$}
	& \rnode{n46}{$\emptyset$}
 	\\ \hline
	23 : \texttt{x->f = new Y;}
	& \rnode{n51}{$\{\tt z, \&z\}$}
	& \rnode{n52}{$\{\tt  q\!\rightarrow z, p\!\rightarrow z\}$}
	& \rnode{n53}{$\emptyset$}
	& \rnode{n54}{$\{\tt x\!\rightarrow X,z\!\rightarrow X\}$}
	& \rnode{n55}{$\{\tt X.f\}$}
	& \rnode{n56}{$\{\tt X\!\stackrel{f}{\rightarrow}\!Y\}$}
 	\\ \hline

	24 : \texttt{y->f = new Z;}
	& \rnode{n61}{$\{\tt z, \&z\}$}
	& \rnode{n62}{$\{\tt  q\!\rightarrow z, p\!\rightarrow z\}$}
	& \rnode{n63}{$\emptyset$}
	& \rnode{n64}{$\{\tt x\!\rightarrow X,z\!\rightarrow X\}$}
	& \rnode{n65}{$\{\tt X.f\}$}
	& \rnode{n66}{$\{\tt X\!\stackrel{f}{\rightarrow}\!Y\}$}
	\\ \hline

	27 : \texttt{t = z->f;}
	& \rnode{n71}{$\{\tt t\}$}
	& \rnode{n72}{$\{\tt  q\!\rightarrow z, p\!\rightarrow z\}$}
	& \rnode{n73}{$\emptyset$}
	& \rnode{n74}{$\{\tt x\!\rightarrow X,z\!\rightarrow X\}$}
	& \rnode{n75}{$\emptyset$}
	& \rnode{n76}{$\{\tt X\!\stackrel{f}{\rightarrow}\!Y, t\!\rightarrow Y\}$}
	\\ \hline

	28 : \texttt{t->vfun ();}

	& \rnode{n81}{}
	& \rnode{n82}{}
	& \rnode{n83}{}
	& \rnode{n84}{}
	& \rnode{n85}{}
	& \rnode{n86}{}
	\\ \hline

\end{tabular}
\caption{Multiple rounds of demand and points-to graph (PTG) propagation with our demand-driven method for static resolution of
virtual call at line 28 of the program in Figure~\ref{m-eg-type}.
The final PTG at a program point is the union of all edges added in each round at that program point.
	Demands are propagated in the backward direction whereas PTG edges are propagated in the forward direction.
}
\label{m-eg-working-our-demand}
\end{figure*}

\end{document}